\documentclass[11pt,USenglish]{article}

\usepackage{geometry}
\newgeometry{vmargin={1in}, hmargin={1in,1in}}  

\usepackage{tikz}
\usetikzlibrary{shapes,arrows,arrows.meta,shadows,positioning,matrix,patterns,decorations.markings,decorations.pathreplacing,svg.path,calc,angles,quotes} 
\usepackage{color}
\usepackage{amssymb}
\usepackage[integrals]{wasysym}
\usepackage{bm}
\usepackage{wrapfig}
\usepackage{amsthm}
\usepackage{amsmath}
\usepackage{amsfonts}
\usepackage{tabu}
\usepackage{enumerate}
\usepackage[mathscr]{eucal}
\usepackage{caption}
\usepackage{authblk}
\usepackage{sectsty}

\usepackage{etoolbox}   
\newtoggle{abstract}    
\togglefalse{abstract} 

\theoremstyle{plain}
\newtheorem{theorem}{Theorem}
\newtheorem{lemma}[theorem]{Lemma}
\newtheorem{proposition}[theorem]{Proposition}

\theoremstyle{definition}
\newtheorem{definition}[theorem]{Definition}

\def\cE{\mathscr{E}}
\newcommand{\sgn}{\operatorname{sgn}}

\def\eps{\varepsilon}
\def\cT{\frac{5}{\sqrt{3}}-1}

\title{The Stretch Factor of Hexagon-Delaunay Triangulations}

\author[1]{Michael Dennis\thanks{michael\_dennis@cs.berkeley.edu}}

\author[2]{Ljubomir Perkovi\'{c}\thanks{lperkovic@cs.depaul.edu (corresponding author)}}

\author[2]{Duru T\"{u}rko\u{g}lu\thanks{duru@cs.uchicago.edu}}

\affil[1]{Computer Science Division, University of California at Berkeley, Berkeley, CA, USA}
\affil[2]{School of Computing, DePaul University, Chicago, IL, USA}

\date{}

\newcommand{\lowerbound}{
\begin{tikzpicture}[scale=1,transform shape]
\node (h) [draw, color=gray, shape border rotate=30, minimum size=1in, regular polygon, regular polygon sides=6] at (0,0) {};

\foreach \anchor/\placement/\label in
    {corner 1/above/$N$, corner 2/above left/$W_N$, corner 3/below left/$W_S$, corner 4/below/$S$, corner 5/below right/$E_S$, corner 6/above right/$E_N$}
\draw[shift=(h.\anchor)] plot[mark=x,mark size=1.9pt] coordinates{(0,0)} 
node [\placement] {{\scriptsize\label}};

\foreach \anchor/\placement/\label in
    {side 1/above/$n_w$, side 2/left/$w$, side 3/below/$s_w$, side 4/below/$s_e$, side 5/right/$e$, side 6/above/$n_e$}
\draw[shift=(h.\anchor)] plot coordinates{(0,0)} 
node [\placement] {{\scriptsize\label}};

\end{tikzpicture}\hspace{1cm}
\begin{tikzpicture}[scale=2]

\node [draw, color=gray, shape border rotate=30, minimum size=2.2cm, regular polygon, regular polygon sides=6,label=-135:{{\color{gray}{$H_1$}}}] at (0.475,-.2044) (H) {};
\node [draw, color=gray, shape border rotate=30, minimum size=2.2cm, regular polygon, regular polygon sides=6,label=45:{{\color{gray}{$H_k$}}}] at (0.525,.7803) (H) {};

\node at (0, 0) [draw,fill, inner sep = 1pt, circle, label=180:{$p=p_0$}] (h) {};
\node at (0.95, -0.48) [draw,fill, inner sep = 1pt, circle, label=0:{$q_0$}] (l) {};

\draw (h) -- (l);

\foreach \x in {1, ..., 12}
{
  \node at (\x*0.004167, \x*0.0879) [draw, fill, inner sep = 1pt, circle] (h) {};
  \draw (h) -- (l);
  \draw (h) -- (\x*0.004167-0.004167, \x*0.0879-0.0879);
  \node at (\x*0.004167+0.95, \x*0.0879-0.48) [draw,fill, inner sep = 1pt, circle] (l) {};
  \draw (h) -- (l);
  \draw (l) -- (\x*0.004167-0.004167+0.95, \x*0.0879-0.0879-0.48);
}

\node at (0.05, 1) [inner sep = 1pt, circle, label=180:{$p_k$}] {};
\node at (1, 0.58) [inner sep = 1pt, circle, label=0:{$q=q_k$}] (h) {};

\end{tikzpicture}\hspace{0.5cm}
\begin{tikzpicture}[scale=1.5]
\node [draw, color=gray, shape border rotate=30, minimum size=1.18125in, regular polygon, regular polygon sides=6,label=-135:{{\color{gray}{$H_1$}}}] at (0.866025,-.5) (H) {};

\node at (0,-0.2) [fill, circle, inner sep=1.2pt, label=180:{$p=p_0$}] (a) {};
\path (a) ++(70:0.27) node (p1) [draw,circle,inner sep=1.2pt,fill=black,label=180:$p_1$] {};
\path [draw] (p1) ++(70:0.27) node (p2) [draw,circle,inner sep=1.2pt,fill=black,label=180:$p_2$] {};

\node at (1.73205,-1) [fill, circle, inner sep=1.2pt, label=0:$q_0$] (c2) {};
\path (c2) ++(70:0.27) node (q1) [draw,circle,inner sep=1.2pt,fill=black,label=0:$q_1$] {};
\path (q1) ++(70:0.27) node (q2) [draw,circle,inner sep=1.2pt,fill=black,label=0:$q_2$] {};

\draw (p2) -- (p1) -- (a) -- (c2) -- (p1) -- (q1) -- (p2) -- (q2) -- (q1) -- (c2);

\end{tikzpicture}
}

\newcommand{\regular}{
\begin{tikzpicture}[scale=1, transform shape]
\clip (0,-1.14) rectangle (8,3);

\node [color=red!25, shape border rotate=30, minimum size=0.7in, regular polygon, regular polygon sides=6] at (0,0.33) (Hm3) {};

\node [color=blue!25, shape border rotate=30, minimum size=0.58in, regular polygon, regular polygon sides=6] at (0.47,0.45) (Hm2) {};

\node [color=red!25, shape border rotate=30, minimum size=0.78in, regular polygon, regular polygon sides=6] at (1.53,-0.05) (Hm1) {};

\node [draw, color=blue!25, shape border rotate=30, minimum size=0.81in, regular polygon, regular polygon sides=6, label={[label distance=-0.1cm]-120:{\color{blue!25}{$H_1$}}}] at (2,0.77) (H1) {};

\node at (1.7,0.7) {$T_1$};

\node [draw, color=red!25, shape border rotate=30, minimum size=1.2in, regular polygon, regular polygon sides=6, label={[label distance=-0.1cm]120:{\color{red!25}{$H_2$}}}] at (2.93,1.18) (H2) {};

\node at (2.4,1.18) {$T_2$};

\node [draw,color=blue!25, shape border rotate=30, minimum size=1.45in, regular polygon, regular polygon sides=6, label={[label distance=-0.1cm]-60:{\color{blue!25}{$H_3$}}}] at (3.99,0.71) (H3) {};

\node at (3.9,0.75) {$T_3$};

\node [draw, color=red!25, shape border rotate=30, minimum size=1.35in, regular polygon, regular polygon sides=6, label={[label distance=-0.1cm]85:{\color{red!25}{$H_4$}}}] at (4.1,0.9) (H4) {};

\node at (4.85,1.35) {$T_4$};

\node [draw, color=blue!25, shape border rotate=30, minimum size=1in, regular polygon, regular polygon sides=6, label={[label distance=-0.1cm]-30:{\color{blue!25}{$H_5$}}}] at (5.675,1.47) (H5) {};

\node at (5.85,1.5) {$T_5$};

\node [color=red!25, shape border rotate=30, minimum size=1.43in, regular polygon, regular polygon sides=6] at (6.06,3.3) (H6) {};

\node [color=blue!25, shape border rotate=30, minimum size=1.6in, regular polygon, regular polygon sides=6] at (6.25,3.625) (H7) {};

\node [color=red!25, shape border rotate=30, minimum size=1.355in, regular polygon, regular polygon sides=6] at (7.24,3.36) (H8) {};

\node [color=blue!25, shape border rotate=30, minimum size=1.45in, regular polygon, regular polygon sides=6] at (7.6,3.27) (H9) {};

\node [color=red!25, shape border rotate=30, minimum size=1.45in, regular polygon, regular polygon sides=6] at (8,3.5) (H10) {};

\node at (Hm3.side 3) [yshift=.4cm,circle,inner sep=1.4pt] (s) {};

\node at (H10.side 6) [yshift=-3.1cm,circle,inner sep=1.4pt] (t) {};

\draw [gray, dotted] (s) -- (t);

\node at (Hm3.corner 6) (a) {};
\node at (Hm3.corner 1) (b) {};
\node at (Hm2.corner 1) (c) {};
\node at (Hm2.corner 2) (d) {};

\node at (intersection of a--b and c--d) [circle,inner sep=1.4pt] (h1) {};

\node at (Hm2.corner 4) (a) {};
\node at (Hm2.corner 5) (b) {};
\node at (Hm1.corner 2) (c) {};
\node at (Hm1.corner 3) (d) {};

\node at (intersection of a--b and c--d) [circle,inner sep=1.4pt] (l1) {};

\node at (1.11,0.8) [draw,fill,circle,inner sep=1.4pt,label=180:{$u_0$}] (h2) {};

\node at (H1.corner 1) (a) {};
\node at (H1.corner 2) (b) {};
\node at (H2.corner 2) (c) {};
\node at (H2.corner 3) (d) {};

\node at (intersection of a--b and c--d) [draw,fill,circle,inner sep=1.4pt,label=90:{$u_1$}] (h4) {};

\node at (H2.corner 3) (a) {};
\node at (H2.corner 4) (b) {};
\node at (H3.corner 2) (c) {};
\node at (H3.corner 3) (d) {};

\node at (intersection of a--b and c--d) [draw,fill,circle,inner sep=1.4pt,label=-90:{$l_0,l_1,l_2$}] (l3) {};

\node at (3.6,2.32) [draw,fill,circle,inner sep=1.4pt,label=90:{$u_2,u_3$}] (h5) {};

\node at (H4.corner 6) (a) {};
\node at (H4.corner 1) (b) {};
\node at (H5.corner 1) (c) {};
\node at (H5.corner 2) (d) {};

\node at (intersection of a--b and c--d) [draw,fill,circle,inner sep=1.4pt,label=45:{$u_4$}] (h7) {};

\node at (H4.corner 5) (a) {};
\node at (H4.corner 6) (b) {};
\node at (H5.corner 3) (c) {};
\node at (H5.corner 4) (d) {};

\node at (intersection of a--b and c--d) [draw,fill,circle,inner sep=1.4pt,label=-45:{$l_3,l_4,l_5$}] (l6) {};

\node at (6.78,1.9) [draw,fill,circle,inner sep=1.4pt,label=0:{$u_5$}] (l8) {};

\node at (4.48,3.6) [circle,inner sep=1.4pt] (h9) {};

\node at (H8.corner 6) (a) {};
\node at (H8.corner 1) (b) {};
\node at (H9.corner 1) (c) {};
\node at (H9.corner 2) (d) {};

\node at (intersection of a--b and c--d) [circle,inner sep=1.4pt] (h10) {};

\node at (8.34,1.85) [circle,inner sep=1.4pt] (l11) {};

\draw (l3) -- (h2) -- (h4) -- (l3) -- (h5) -- (h4);
\draw (l3) -- (l6) -- (h5) -- (h7) -- (l6) -- (l8) -- (h7); 

\node (tmp) at ($(h7)+(150:5)$) {};
\node (tmp2) at ($(h2)+(90:4)$) {};

\coordinate (z) at (intersection of h7--tmp and h2--tmp2);

\node (tmp) at ($(h7)+(330:5)$) {};
\node (tmp2) at ($(l11)+(270:4)$) {};

\coordinate (z) at (intersection of h7--tmp and l11--tmp2);

\end{tikzpicture}
}

\newcommand{\discretedNdS}{
\begin{tikzpicture}[scale=1]
\clip (-2.5,-1.55) rectangle (6.5,1.55);
\node (h) [draw, color=gray, shape border rotate=30, minimum size=1in, regular polygon, regular polygon sides=6,label={[label distance=-1.5cm]0:{\color{gray}{$H_i$}}}] at (0,0) {};

\node at (h.side 1) [xshift=-0.2cm,yshift=-0.12cm,draw,circle,fill,inner sep=1.4pt,label=90:{\color{blue}{\large +}}] (s1) {};
\node at (h.side 2) [draw,circle,fill,inner sep=1.4pt,label=135:{\color{blue}{\large +}}] (s2) {};

\draw (s1) -- (h.corner 1) [->,>=stealth',blue,label=-45:{$+$}];
\draw ($(s2) + (-0.1cm,0)$) -- ($(h.corner 2) + (-0.1cm,0.052cm)$) -- ($(h.corner 1) + (0,0.1cm)$) [->,>=stealth',blue,label=-45:{$+$}];

\node at (h.side 6) [draw,circle,fill,inner sep=1.4pt,label=90:{\color{red}{\large --}}] (s1) {};
\node at (h.side 5) [yshift=-0.12cm,draw,circle,fill,inner sep=1.4pt,label=45:{\color{red}{\large --}}] (s2) {};

\draw (s1) -- (h.corner 1) [->,>=stealth',red];
\draw ($(s2) + (0.1cm,0)$) -- ($(h.corner 6) + (0.1cm,0.052cm)$) -- ($(h.corner 1) + (0,0.1cm)$) [->,>=stealth',red];

\node (h) [draw, color=gray, shape border rotate=30, minimum size=1in, regular polygon, regular polygon sides=6,label={[label distance=-1.5cm]0:{\color{gray}{$H_i$}}}] at (5,0) {};

\node at (h.side 2) [draw,circle,fill,inner sep=1.4pt,label=-135:{\color{blue}{\large +}}] (s2) {};
\node at (h.side 3) [xshift=+0.2cm,yshift=-0.12cm,draw,circle,fill,inner sep=1.4pt,label=-90:{\color{blue}{\large +}}] (s3) {};

\draw (s3) -- (h.corner 4) [->,>=stealth',blue,label=-45:{$+$}];
\draw ($(s2) + (-0.1cm,0)$) -- ($(h.corner 3) + (-0.1cm,-0.052cm)$) -- ($(h.corner 4) + (0,-0.1cm)$) [->,>=stealth',blue,label=-45:{$+$}];

\node at (h.side 5) [yshift=-0.12cm,draw,circle,fill,inner sep=1.4pt,label=-45:{\color{red}{\large --}}] (s2) {};
\node at (h.side 4) [draw,circle,fill,inner sep=1.4pt,label=-90:{\color{red}{\large --}}] (s3) {};

\draw (s3) -- (h.corner 4) [->,>=stealth',red];
\draw ($(s2) + (0.1cm,0)$) -- ($(h.corner 5) + (0.1cm,-0.052cm)$) -- ($(h.corner 4) + (0,-0.1cm)$) [->,>=stealth',red];
\end{tikzpicture}\hspace{1.34cm}
\begin{tikzpicture}[scale=1]
\clip (-2.35,-1.45) rectangle (6.5,1.65);
\node (h) [draw, color=gray, shape border rotate=30, minimum size=1in, regular polygon, regular polygon sides=6,label={[label distance=-1.5cm]0:{\color{gray}{$H_i$}}}] at (0,0) {};

\node at (h.side 2) [draw,circle,fill,inner sep=1.4pt,label=180:{$u_{i-1}$}] (uj) {};
\node at (h.side 6) [draw,circle,fill,inner sep=1.4pt,label=0:{$u_i$}] (ujp) {};

\draw (uj) -- (h.corner 2) -- (h.corner 1) [->,>=stealth',blue];
\draw (ujp) -- (h.corner 1) [->,>=stealth',red];

\draw (uj) -- (ujp);

\node at (-1.45,1.1) [blue] {$p_N(u_{i-1},i)$};
\node at (1.05,1.45) [red] {$p_N(u_i,i)$};
\end{tikzpicture}
}

\newcommand{\definitions}{
\begin{tikzpicture}[scale=1, transform shape]
\clip (-2.2,-1.14) rectangle (12,5.66);

\node [draw, color=red!25, shape border rotate=30, minimum size=0.7in, regular polygon, regular polygon sides=6] at (0,0.33) (Hm3) {};

\node [draw, color=blue!25, shape border rotate=30, minimum size=0.58in, regular polygon, regular polygon sides=6] at (0.47,0.45) (Hm2) {};

\node [draw,color=red!25, shape border rotate=30, minimum size=0.78in, regular polygon, regular polygon sides=6] at (1.53,-0.05) (Hm1) {};

\node [draw, color=blue!25, shape border rotate=30, minimum size=0.81in, regular polygon, regular polygon sides=6] at (2,0.77) (H1) {};

\node [draw, color=red!25, shape border rotate=30, minimum size=1.2in, regular polygon, regular polygon sides=6] at (2.93,1.18) (H2) {};

\node [draw,color=blue!25, shape border rotate=30, minimum size=1.45in, regular polygon, regular polygon sides=6] at (3.99,0.71) (H3) {};

\node [draw, color=red!25, shape border rotate=30, minimum size=1.35in, regular polygon, regular polygon sides=6] at (4.1,0.9) (H4) {};

\node [draw, color=blue!25, shape border rotate=30, minimum size=1in, regular polygon, regular polygon sides=6] at (5.675,1.47) (H5) {};

\node [draw,color=red!25, shape border rotate=30, minimum size=1.43in, regular polygon, regular polygon sides=6] at (6.06,3.3) (H6) {};

\node [draw,color=blue!25, shape border rotate=30, minimum size=1.6in, regular polygon, regular polygon sides=6] at (6.25,3.625) (H7) {};

\node [draw,color=red!25, shape border rotate=30, minimum size=1.355in, regular polygon, regular polygon sides=6] at (7.24,3.36) (H8) {};

\node [draw,color=blue!25, shape border rotate=30, minimum size=1.45in, regular polygon, regular polygon sides=6] at (7.6,3.27) (H9) {};

\node [draw, color=red!25, shape border rotate=30, minimum size=1.45in, regular polygon, regular polygon sides=6] at (8,3.5) (H10) {};

\node at (Hm3.side 2) [yshift=-0.32cm,draw,circle,fill,inner sep=1.4pt,label=180:{$u_0,l_0,s$}] (s) {};

\node at (H10.side 5) [yshift=-0.6cm,draw,circle,fill,inner sep=1.4pt,label=0:{$t,u_{12},l_{12}$}] (t) {};

\draw [gray, dotted] (s) -- (t);

\node at (Hm3.corner 6) (a) {};
\node at (Hm3.corner 1) (b) {};
\node at (Hm2.corner 1) (c) {};
\node at (Hm2.corner 2) (d) {};

\node at (intersection of a--b and c--d) [draw,fill,circle,inner sep=1.4pt,label=90:{$u_1$}] (h1) {};

\node at (Hm2.corner 4) (a) {};
\node at (Hm2.corner 5) (b) {};
\node at (Hm1.corner 2) (c) {};
\node at (Hm1.corner 3) (d) {};

\node at (intersection of a--b and c--d) [draw,fill,circle,inner sep=1.4pt,label=-90:{$l_1,l_2$}] (l1) {};

\node at (1.11,0.7) [draw,fill,circle,inner sep=1.4pt,label=90:{$u_2,u_3$}] (h2) {};

\node at (H1.corner 1) (a) {};
\node at (H1.corner 2) (b) {};
\node at (H2.corner 2) (c) {};
\node at (H2.corner 3) (d) {};

\node at (intersection of a--b and c--d) [draw,fill,circle,inner sep=1.4pt,label=90:{$u_4$}] (h4) {};

\node at (H2.corner 3) (a) {};
\node at (H2.corner 4) (b) {};
\node at (H3.corner 2) (c) {};
\node at (H3.corner 3) (d) {};

\node at (intersection of a--b and c--d) [draw,fill,circle,inner sep=1.4pt,label=-90:{$l_3,l_4,l_5$}] (l3) {};

\node at (3.6,2.32) [draw,fill,circle,inner sep=1.4pt,label=90:{$u_5,u_6$}] (h5) {};

\node at (H4.corner 6) (a) {};
\node at (H4.corner 1) (b) {};
\node at (H5.corner 1) (c) {};
\node at (H5.corner 2) (d) {};

\node at (intersection of a--b and c--d) [draw,fill,circle,inner sep=1.4pt,label=45:{$u_7,u_8$}] (h7) {};

\node at (H4.corner 5) (a) {};
\node at (H4.corner 6) (b) {};
\node at (H5.corner 3) (c) {};
\node at (H5.corner 4) (d) {};

\node at (intersection of a--b and c--d) [draw,fill,circle,inner sep=1.4pt,label=-45:{$l_6,l_7$}] (l6) {};

\node at (6.78,1.9) [draw,fill,circle,inner sep=1.4pt,label=-90:{$l_8,l_9,l_{10}$}] (l8) {};

\node at (4.48,3.6) [draw,fill,circle,inner sep=1.4pt,label=180:{$u_9$}] (h9) {};

\node at (H8.corner 6) (a) {};
\node at (H8.corner 1) (b) {};
\node at (H9.corner 1) (c) {};
\node at (H9.corner 2) (d) {};

\node at (intersection of a--b and c--d) [draw,fill,circle,inner sep=1.4pt,label=90:{$u_{10},u_{11}$}] (h10) {};

\node at (8.34,1.85) [draw,fill,circle,inner sep=1.4pt,label=-45:{$l_{11}$}] (l11) {};

\draw (h1) -- (s) -- (l1) -- (h1) -- (h2) -- (l1) -- (l3) -- (h2) -- (h4) -- (l3) -- (h5) -- (h4);
\draw (l3) -- (l6) -- (h5) -- (h7) -- (l6) -- (l8) -- (h7) -- (h9) -- (l8) -- (h10) -- (h9);
\draw (l8) -- (l11) -- (h10) -- (t) -- (l11); 

\node (tmp) at ($(h7)+(150:5)$) {};
\node (tmp2) at ($(h2)+(90:4)$) {};

\coordinate (z) at (intersection of h7--tmp and h2--tmp2);

\draw [dashed, red] (h2) -- (z) -- (h7);

\node (tmp) at ($(h4)+(150:5)$) {};
\coordinate (z) at (intersection of h4--tmp and h2--tmp2);
\draw [dotted, red] (h4) -- (z) -- (h2);

\node (tmp) at ($(h5)+(150:5)$) {};
\node (tmp2) at ($(h4)+(90:5)$) {};
\coordinate (z) at (intersection of h5--tmp and h4--tmp2);
\draw [dotted, red] (h4) -- (z) -- (h5);

\node (tmp) at ($(h7)+(150:5)$) {};
\node (tmp2) at ($(h5)+(90:5)$) {};
\coordinate (z) at (intersection of h7--tmp and h5--tmp2);
\draw [dotted, red] (h7) -- (z) -- (h5);

\node (tmp) at ($(h7)+(330:5)$) {};
\node (tmp2) at ($(l11)+(270:4)$) {};

\coordinate (z) at (intersection of h7--tmp and l11--tmp2);

\draw [dashed, red] (l11) -- (z) -- (h7);

\node (tmp2) at ($(l8)+(270:5)$) {};
\coordinate (z) at (intersection of h7--tmp and l8--tmp2);
\draw [dotted, red] (l8) -- (z) -- (h7);

\node (tmp) at ($(l8)+(330:5)$) {};
\node (tmp2) at ($(l11)+(270:5)$) {};
\coordinate (z) at (intersection of l8--tmp and l11--tmp2);
\draw [dotted, red] (l8) -- (z) -- (l11);

\end{tikzpicture}
}

\newcommand{\caseA}{
\begin{tikzpicture}[scale=.9, transform shape]
\clip (-0.3,-0.8) rectangle (10,4.3);

\coordinate (t) at (10,3.333);
\coordinate (s) at (0,0);
\draw [gray, dotted] (s) -- (t);

\node [draw, gray!25, shape border rotate=30, minimum size=0.58in, regular polygon, regular polygon sides=6] at (0.47,0.45) (Hm2) {};

\node [color=gray!50, dashed, shape border rotate=30, minimum size=0.78in, regular polygon, regular polygon sides=6] at (1.53,-0.05) (Hm1) {};

\node [draw, color=red!25, shape border rotate=30, minimum size=0.81in, regular polygon, regular polygon sides=6] at (2,0.77) (H0) {};

\node [color=gray!50, shape border rotate=30, minimum size=0.8in, regular polygon, regular polygon sides=6] at (1.73,0.61) (H1) {};

\node [draw, color=red!25, shape border rotate=30, minimum size=1.2in, regular polygon, regular polygon sides=6] at (2.93,1.18) (H2) {};

\node [color=gray!50, shape border rotate=30, minimum size=1.45in, regular polygon, regular polygon sides=6] at (3.99,0.71) (H3) {};

\node [draw, color=red!25, shape border rotate=30, minimum size=1.35in, regular polygon, regular polygon sides=6] at (4.1,0.9) (H4) {};

\node [draw, color=blue!25, shape border rotate=30, minimum size=1in, regular polygon, regular polygon sides=6] at (5.675,1.47) (H5) {};

\node [color=gray!50, shape border rotate=30, minimum size=0.906in, regular polygon, regular polygon sides=6] at (5.775,2.27) (H6) {};

\node [draw, color=gray!25, shape border rotate=30, minimum size=0.906in, regular polygon, regular polygon sides=6] at (10.37,2.8) (H7) {};

\node at (1.11,0.7) [draw,fill,circle,inner sep=1.4pt,label=180:{$p=u_{r'}$}] (h0) {};

\node at (H1.corner 1) (a) {};
\node at (H1.corner 2) (b) {};
\node at (H2.corner 2) (c) {};
\node at (H2.corner 3) (d) {};

\node at (intersection of a--b and c--d) [draw,fill,circle,inner sep=1.4pt,label=-4:{$u_{r'+1}$}] (h1h2) {};

\node at (H2.corner 3) (a) {};
\node at (H2.corner 4) (b) {};
\node at (H3.corner 2) (c) {};
\node at (H3.corner 3) (d) {};

\node at (intersection of a--b and c--d) [draw,fill,circle,inner sep=1.4pt] (l2l3) {};

\node at (3.6,2.32) [draw,fill,circle,inner sep=1.4pt,label=177:{$u_{r-1}$}] (h3h4) {};

\node at (H4.corner 6) (a) {};
\node at (H4.corner 1) (b) {};
\node at (H5.corner 1) (c) {};
\node at (H5.corner 2) (d) {};

\node at (intersection of a--b and c--d) [draw,fill,circle,inner sep=1.4pt,label=90:{$u_r$}] (h5h6) {};

\node at (6.78,1.9) [draw,fill,circle,inner sep=1.4pt,label=0:{$l_s$}] (l6) {};

\node at (H4.corner 5) (a) {};
\node at (H4.corner 6) (b) {};
\node at (H5.corner 3) (c) {};
\node at (H5.corner 4) (d) {};

\node at (intersection of a--b and c--d) [draw,fill,circle,inner sep=1.4pt] (l4l5) {};

\node (tmp) at ($(h3h4)+(90:4)$) {};
\node (tmp2) at ($(h5h6)+(150:4)$) {};

\coordinate (z3) at (intersection of h3h4--tmp and h5h6--tmp2);

\node (tmp) at ($(h3h4)+(150:4)$) {};
\node (tmp2) at ($(h1h2)+(90:4)$) {};

\coordinate (z2) at (intersection of h3h4--tmp and h1h2--tmp2);

\node (tmp) at ($(h1h2)+(150:4)$) {};
\node (tmp2) at ($(h0)+(90:4)$) {};

\coordinate (z1) at (intersection of h1h2--tmp and h0--tmp2);

\draw [dashed, red] (h5h6) -- (z3) -- (h3h4) -- (z2) -- (h1h2) -- (z1) -- (h0);

\node (tmp) at ($(h5h6)+(150:5)$) {};
\node (tmp2) at ($(h0)+(90:4)$) {};

\coordinate (z) at (intersection of h5h6--tmp and h1h2--tmp2);

\draw [dotted, red] (h5h6) -- (z) -- (h0);

\draw [thick,red] (h0) -- (h1h2) -- (h3h4) -- (h5h6);
\draw (h0) -- (l2l3);
\draw (h1h2) -- (l2l3) -- (h1h2);
\draw (h3h4) -- (l2l3) -- (l4l5) -- (h3h4);
\draw (h5h6) -- (l4l5) -- (l6);
\draw [thick, blue] (l6) -- (h5h6);

\coordinate (tmp) at ($(h5h6)+(-30:8)$);
\coordinate (tmp2) at ($(l6)+(-90:6)$);

\coordinate (z4) at (intersection of h5h6--tmp and l6--tmp2);

\draw [dashed, blue] (h5h6) -- (z4) -- (l6);

\coordinate (tmp) at ($(l6)+(-30:3)$);
\node at ($(tmp)+(90:2.25)$) [draw,fill,circle,inner sep=1.4pt,label=-135:{$l_{s'}=q$}] (tmp2) {};

\coordinate (tmp3) at ($(z4)+(-30:6)$);
\coordinate (tmp4) at ($(tmp)+(-90:4)$);

\coordinate (z5) at (intersection of z4--tmp3 and tmp--tmp4);

\draw [dotted, red] (h5h6) -- (z5) -- (tmp2);

\draw [dashed, red] (l6) -- (tmp) -- (tmp2);
\end{tikzpicture}
}

\def\tikzbox{
\clip (-0.42,-1.5) rectangle (4.05,3);
\coordinate (t) at (4,2.5);
\coordinate (s) at (-.5,0);

\node [transform shape,draw, color=red!25, shape border rotate=30, minimum size=0.7in, regular polygon, regular polygon sides=6] at (0.51,1.2) (H0) {};

\node at (0.16,0.51) [transform shape,inner sep=1.4pt, draw,circle,fill=red,label=180:{$u_r$}] (h0) {};

\draw [gray, dotted] (s) -- (t);

\node [transform shape,draw, color=blue!25, shape border rotate=30, minimum size=0.8in, regular polygon, regular polygon sides=6] at (.87,1.35) (H1) {};

\node [draw, color=red!25, shape border rotate=30, minimum size=0.95in, regular polygon, regular polygon sides=6] at (1.2,1.35) (H2) {};

\node [draw, color=blue!25, shape border rotate=30, minimum size=0.95in, regular polygon, regular polygon sides=6] at (1.3,1.41) (H3) {};

\node [draw, color=red!25, shape border rotate=30, minimum size=0.75in, regular polygon, regular polygon sides=6] at (1.74,1.41) (H4) {};

\node [draw, color=blue!25, shape border rotate=30, minimum size=0.685in, regular polygon, regular polygon sides=6] at (1.96,1.53) (H5) {};

\node [draw, color=red!25, shape border rotate=30, minimum size=0.6in, regular polygon, regular polygon sides=6] at (2.9,1.59) (H6) {};

\node at (H0.corner 4) (a) {};
\node at (H0.corner 5) (b) {};
\node at (H1.corner 3) (c) {};
\node at (H1.corner 4) (d) {};

\node at (intersection of a--b and c--d) [transform shape,draw,fill,circle,inner sep=1.4pt,label=270:{$l_{r+1}$}] (l1) {};

\node at (H0.corner 1) (az) {};
\node at (H0.corner 2) (bz) {};
\node at (H1.corner 2) (cz) {};
\node at (H1.corner 3) (dz) {};

\node at (intersection of az--bz and cz--dz) [transform shape,draw,fill,circle,inner sep=1.4pt] (h1) {};

\node at (.47,2.14) [draw,fill,circle,inner sep=1.4pt] (h2) {};

\node at (2.02,0.62) [draw,fill,circle,inner sep=1.4pt] (l2) {};

\node at (1.82,2.32) [draw,fill,circle,inner sep=1.4pt] (h3) {};

\node at (2.57,1.02) [draw,fill,circle,inner sep=1.4pt] (l3) {};

\node at (2.47,2.1) [transform shape,draw,fill,circle,inner sep=1.4pt,label=90:{$u_{s-1}$}] (h4) {};

\node at (3.45,2.03) [transform shape,draw,fill=red,circle,inner sep=1.4pt,label=0:{$l_s$}] (l4) {};

\coordinate (z) at (3.45, -1.39);

\draw [transform shape, dashed, red] (h0) -- (z);
\draw [transform shape, dashed, red] (l4) -- (z);

\draw [blue,thick] (l1) -- (h0);
\draw (h0) -- (h1) -- (l1) -- (h2) -- (h1);
\draw (l1) -- (l2) -- (h2) -- (h3) -- (l2) -- (l3) -- (h3) -- (h4) -- (l3) -- (l4) -- (h4);

}

\newcommand{\caseB}{
\begin{tikzpicture}
\begin{scope}
\tikzbox
\end{scope}
\begin{scope}[shift={(6,1.82)},yscale=1,xscale=1,rotate=-60]
\tikzbox
\end{scope}
\end{tikzpicture}
}

\newcommand{\othercase}{
\begin{tikzpicture}[scale=1, transform shape]
\clip (-0.42,-0.5) rectangle (4.05,2.9);
\coordinate (t) at (4,2.5);
\coordinate (s) at (-.5,0);

\node [transform shape,draw, color=red!25, shape border rotate=30, minimum size=0.7in, regular polygon, regular polygon sides=6] at (0.51,1.2) (H0) {};

\node at (0.16,0.51) [transform shape,inner sep=1.4pt, draw,circle,fill=red,label=180:{$u_r$}] (h0) {};

\draw [gray, dotted] (s) -- (t);

\node [transform shape,draw, color=blue!25, shape border rotate=30, minimum size=0.8in, regular polygon, regular polygon sides=6] at (.87,1.35) (H1) {};

\node [draw, color=red!25, shape border rotate=30, minimum size=0.95in, regular polygon, regular polygon sides=6] at (1.2,1.35) (H2) {};

\node [draw, color=blue!25, shape border rotate=30, minimum size=0.95in, regular polygon, regular polygon sides=6] at (1.3,1.41) (H3) {};

\node [draw, color=red!25, shape border rotate=30, minimum size=0.75in, regular polygon, regular polygon sides=6] at (1.74,1.41) (H4) {};

\node [draw, color=blue!25, shape border rotate=30, minimum size=0.685in, regular polygon, regular polygon sides=6] at (1.96,1.53) (H5) {};

\node [draw, color=red!25, shape border rotate=30, minimum size=0.6in, regular polygon, regular polygon sides=6] at (2.9,1.59) (H6) {};

\node at (H0.corner 4) (a) {};
\node at (H0.corner 5) (b) {};
\node at (H1.corner 3) (c) {};
\node at (H1.corner 4) (d) {};

\node at (intersection of a--b and c--d) [transform shape,draw,fill,circle,inner sep=1.4pt,label=270:{$l_{r+1}$}] (l1) {};

\node at (H0.corner 1) (az) {};
\node at (H0.corner 2) (bz) {};
\node at (H1.corner 2) (cz) {};
\node at (H1.corner 3) (dz) {};

\node at (intersection of az--bz and cz--dz) [transform shape,draw,fill,circle,inner sep=1.4pt] (h1) {};

\node at (.47,2.14) [draw,fill,circle,inner sep=1.4pt] (h2) {};

\node at (H4.corner 4) [draw,fill,circle,inner sep=1.4pt,label=-45:{$l_{t}$}] (l2) {};

\node at (1.82,2.32) [draw,fill,circle,inner sep=1.4pt,label=90:{$u_{t}$}] (h3) {};

\node at (2.57,1.02) [draw,fill,circle,inner sep=1.4pt] (l3) {};

\node at (2.47,2.1) [transform shape,draw,fill,circle,inner sep=1.4pt,label=90:{$u_{s-1}$}] (h4) {};

\node at (3.45,2.03) [transform shape,draw,fill=red,circle,inner sep=1.4pt,label=0:{$l_s$}] (l4) {};

\node (tmp) at ($(h0)+(-30:10)$) [transform shape] {};
\node (tmp2) at ($(l2)+(-90:10)$) [transform shape] {};

\coordinate (z) at (intersection of h0--tmp and l2--tmp2);

\draw [transform shape, dashed, red] (h0) -- (z);
\draw [transform shape, dashed, red] (l2) -- (z);

\node (tmp) at ($(l2)+(90:10)$) [transform shape] {};
\node (tmp2) at ($(h3)+(150:10)$) [transform shape] {};

\coordinate (z) at (intersection of l2--tmp and h3--tmp2);

\draw [transform shape, dashed, red] (l2) -- (z);
\draw [transform shape, dashed, red] (h3) -- (z);

\node (tmp) at ($(h3)+(-30:10)$) [transform shape] {};
\node (tmp2) at ($(l4)+(-90:10)$) [transform shape] {};

\coordinate (z) at (intersection of h3--tmp and l4--tmp2);

\draw [transform shape, dashed, red] (h3) -- (z);
\draw [transform shape, dashed, red] (l4) -- (z);

\draw [blue,thick] (l1) -- (h0);
\draw (h0) -- (h1) -- (l1) -- (h2) -- (h1);
\draw (l1) -- (l2) -- (h2) -- (h3) -- (l2) -- (l3) -- (h3) -- (h4) -- (l3) -- (l4) -- (h4);
\end{tikzpicture}
}

\newcommand{\notation}{
\begin{tikzpicture}[scale=1, transform shape]
\clip (0.6,-1.55) rectangle (6.5,3);

\coordinate (t) at (10,3.333);
\coordinate (s) at (0,0);
\draw [gray, dotted] (s) -- (t);

\node [color=gray!50, shape border rotate=30, minimum size=0.8in, regular polygon, regular polygon sides=6] at (1.73,0.61) (H1) {};

\node [draw, color=red!25, shape border rotate=30, minimum size=1.2in, regular polygon, regular polygon sides=6,label=135:{\color{red!25}{$H_i$}}] at (2.93,1.18) (H2) {};

\node [draw,circle,fill,color=red!25,inner sep=0.5pt,label=90:{\scriptsize\color{red!25}{$c_i$}}] at (2.93,1.18) {};

\node [draw, color=blue!25, shape border rotate=30, minimum size=1.45in, regular polygon, regular polygon sides=6,label=45:{\color{blue!25}{$H_{i+1}$}}] at (3.99,0.71) (H3) {};

\node [draw,circle,fill,color=blue!25,inner sep=0.5pt,label=90:{\scriptsize\color{blue!25}{$c_{i+1}$}}] at (3.99,0.71) {};

\node [draw, color=gray, shape border rotate=30, minimum size=1.245in, regular polygon, regular polygon sides=6,label={[label distance=-0.9cm]-80:{\color{gray}{$H(x)$}}}] at (3.7,0.8) (H25) {};

\node [color=red!25, shape border rotate=30, minimum size=1.35in, regular polygon, regular polygon sides=6] at (4.1,0.9) (H4) {};

\node [draw,circle,fill,color=gray,inner sep=0.5pt] at (3.7,0.8) (c) {};

\coordinate (cx) at (3.7,-1.1);
\draw [dashed,gray] (c) -- (cx); 
\node at (cx) [label={[gray]-90:{{\scriptsize $x$}}}] {};

\coordinate (wx) at ($(H25.corner 3) + (0,-1.105)$);
\draw [dashed,gray] (H25.corner 3) -- (wx); 
\node at (wx) [label={[label distance=-0.089cm,gray]-90:{{\scriptsize $w(x)$}}}] {};

\coordinate (ex) at ($(H25.corner 5) + (0,-1.105)$);
\draw [dashed,gray] (H25.corner 5) -- (ex); 
\node at (ex) [label={[label distance=-0.089cm,gray]-90:{{\scriptsize $e(x)$}}}] {};

\draw [gray,<->,{Stealth-Stealth},label=90:${\tt r}(x)$] ($(cx)+(0,0.2)$) -- ($(wx)+(0,0.2)$) node [midway, above] {\scriptsize ${\tt r}(x)$};
\draw [->] (1,-1.1) -- (6,-1.1);

\node at (H1.corner 1) (a) {};
\node at (H1.corner 2) (b) {};
\node at (H2.corner 2) (c) {};
\node at (H2.corner 3) (d) {};

\node at (intersection of a--b and c--d) [draw,fill,circle,inner sep=1.4pt,label=180:{\scriptsize $u_{i-1}$}] (h1h2) {};

\node at (H2.corner 3) (a) {};
\node at (H2.corner 4) (b) {};
\node at (H3.corner 2) (c) {};
\node at (H3.corner 3) (d) {};

\node at (intersection of a--b and c--d) [draw,fill,circle,inner sep=1.4pt,label=-135:{\scriptsize $\ell(x)=l_i$}] (l2l3) {};

\node at (3.6,2.32) [draw,fill,circle,inner sep=1.4pt,label={[label distance=0.2cm]90:{\scriptsize $u(x)=u_i$}}] (h3h4) {};

\node at (H3.side 5) [draw,yshift=-.3cm,fill,circle,inner sep=1.4pt,label=0:{\scriptsize $l_{i+1}$}] (l4l5) {};

\draw (l2l3) -- (h1h2) -- (h3h4) -- (l4l5) -- (l2l3) -- (h3h4);
\end{tikzpicture}
}

\newcommand{\dNdS}{
\begin{tikzpicture}[scale=1]
\clip (-1.8,-1.55) rectangle (6.5,1.55);
\node (h) [draw, color=gray, shape border rotate=30, minimum size=1in, regular polygon, regular polygon sides=6,label={[label distance=-1.5cm]0:{\color{gray}{$H(x)$}}}] at (0,0) {};

\node at (h.side 1) [xshift=-0.2cm,yshift=-0.12cm,draw,circle,fill,inner sep=1.4pt,label=90:{\color{blue}{\large +}}] (s1) {};
\node at (h.side 2) [draw,circle,fill,inner sep=1.4pt,label=135:{\color{blue}{\large +}}] (s2) {};

\draw (s1) -- (h.corner 1) [->,>=stealth',blue,label=-45:{$+$}];
\draw ($(s2) + (-0.1cm,0)$) -- ($(h.corner 2) + (-0.1cm,0.052cm)$) -- ($(h.corner 1) + (0,0.1cm)$) [->,>=stealth',blue,label=-45:{$+$}];

\node at (h.side 6) [draw,circle,fill,inner sep=1.4pt,label=90:{\color{red}{\large --}}] (s1) {};
\node at (h.side 5) [yshift=-0.12cm,draw,circle,fill,inner sep=1.4pt,label=45:{\color{red}{\large --}}] (s2) {};

\draw (s1) -- (h.corner 1) [->,>=stealth',red];
\draw ($(s2) + (0.1cm,0)$) -- ($(h.corner 6) + (0.1cm,0.052cm)$) -- ($(h.corner 1) + (0,0.1cm)$) [->,>=stealth',red];

\node (h) [draw, color=gray, shape border rotate=30, minimum size=1in, regular polygon, regular polygon sides=6,label={[label distance=-1.5cm]0:{\color{gray}{$H(x)$}}}] at (5,0) {};

\node at (h.side 2) [draw,circle,fill,inner sep=1.4pt,label=-135:{\color{blue}{\large +}}] (s2) {};
\node at (h.side 3) [xshift=+0.2cm,yshift=-0.12cm,draw,circle,fill,inner sep=1.4pt,label=-90:{\color{blue}{\large +}}] (s3) {};

\draw (s3) -- (h.corner 4) [->,>=stealth',blue,label=-45:{$+$}];
\draw ($(s2) + (-0.1cm,0)$) -- ($(h.corner 3) + (-0.1cm,-0.052cm)$) -- ($(h.corner 4) + (0,-0.1cm)$) [->,>=stealth',blue,label=-45:{$+$}];

\node at (h.side 5) [yshift=-0.12cm,draw,circle,fill,inner sep=1.4pt,label=-45:{\color{red}{\large --}}] (s2) {};
\node at (h.side 4) [draw,circle,fill,inner sep=1.4pt,label=-90:{\color{red}{\large --}}] (s3) {};

\draw (s3) -- (h.corner 4) [->,>=stealth',red];
\draw ($(s2) + (0.1cm,0)$) -- ($(h.corner 5) + (0.1cm,-0.052cm)$) -- ($(h.corner 4) + (0,-0.1cm)$) [->,>=stealth',red];
\end{tikzpicture}\hspace{1.34cm}
\begin{tikzpicture}[scale=1]
\clip (-2.35,-1.45) rectangle (6.5,1.65);
\node (h) [draw, color=gray, shape border rotate=30, minimum size=1in, regular polygon, regular polygon sides=6,label={[label distance=-2cm]0:{\color{gray}{$H(x)=H_i$}}}] at (0,0) {};

\node at (h.side 2) [draw,circle,fill,inner sep=1.4pt,label=180:{$u_{i-1}$}] (uj) {};
\node at (h.side 6) [draw,circle,fill,inner sep=1.4pt,label=0:{$u_i$}] (ujp) {};

\draw (uj) -- (h.corner 2) -- (h.corner 1) [->,>=stealth',blue];
\draw (ujp) -- (h.corner 1) [->,>=stealth',red];

\draw (uj) -- (ujp);

\node at (-1.45,1.1) [blue] {$p_N(u_{i-1},x)$};
\node at (1.05,1.45) [red] {$p_N(u_i,x)$};
\end{tikzpicture}
}

\newcommand{\UandL}{
\begin{tikzpicture}[scale=1, transform shape]
\clip (-1.2,-1.3) rectangle (5.7,2.75);

\node [draw, color=red!25, shape border rotate=30, minimum size=0.7in, regular polygon, regular polygon sides=6] at (0,0.33) (Hm3) {};

\node [draw, color=blue!25, shape border rotate=30, minimum size=0.58in, regular polygon, regular polygon sides=6] at (0.47,0.45) (Hm2) {};

\node [draw,color=red!25, shape border rotate=30, minimum size=0.78in, regular polygon, regular polygon sides=6] at (1.53,-0.05) (Hm1) {};

\node [draw, color=blue!25, shape border rotate=30, minimum size=0.81in, regular polygon, regular polygon sides=6] at (2,0.77) (H1) {};

\node [draw, color=red!25, shape border rotate=30, minimum size=1.2in, regular polygon, regular polygon sides=6] at (2.93,1.18) (H2) {};

\node [draw, color=gray!50, shape border rotate=30, minimum size=1.2in, regular polygon, regular polygon sides=6] at ($(2.93,1.18)+ (1*0.3, -.577*0.3)$) (H3) {};

\node at (Hm3.side 2) [yshift=-0.32cm,draw,circle,fill,inner sep=1.4pt,label=180:{$p$}] (s) {};

\node at (10,3.333) [draw,fill,circle,inner sep=1.4pt,label=180:{$t$}] (t) {};

\draw [gray, dotted] (s) -- (t);

\coordinate (o) at (-1,-.8);
\coordinate (x) at (4,-.8);
\draw [->] (o) -- (x);

\coordinate (N) at (H3.corner 1);
\draw [gray, dashed] (N) -- (N |- 52,-.8);
\node at (N |- 52,-.8) [label={[label distance=-0.05cm]-90:{\color{gray}{$x$}}}] {};
\draw [gray, dashed] (s) -- (s |- 52,-.8);
\node at (s |- 52,-.8) [label={[label distance=-0.2cm]-90:{\color{gray}{${\tt x}(p)$}}}] {};

\node at (Hm3.corner 6) (a) {};
\node at (Hm3.corner 1) (b) {};
\node at (Hm2.corner 1) (c) {};
\node at (Hm2.corner 2) (d) {};

\node at (intersection of a--b and c--d) [draw,fill,circle,inner sep=1.4pt] (h1) {};

\node at (Hm2.corner 4) (a) {};
\node at (Hm2.corner 5) (b) {};
\node at (Hm1.corner 2) (c) {};
\node at (Hm1.corner 3) (d) {};

\node at (intersection of a--b and c--d) [draw,fill,circle,inner sep=1.4pt] (l1) {};

\node at (1.11,0.7) [draw,fill,circle,inner sep=1.4pt] (h2) {};

\node at (H1.corner 1) (a) {};
\node at (H1.corner 2) (b) {};
\node at (H2.corner 2) (c) {};
\node at (H2.corner 3) (d) {};

\node at (intersection of a--b and c--d) [draw,fill,circle,inner sep=1.4pt] (h4) {};

\node at (H2.corner 3) (a) {};
\node at (H2.corner 4) (b) {};
\node at (H1.corner 4) (c) {};
\node at (H1.corner 5) (d) {};

\node at (intersection of a--b and c--d) [draw,fill,circle,inner sep=1.4pt,label=0:{$\ell(x)=l_i$}] (l3) {};

\node at ($(3.6,2.32) + (1*0.4, -.577*0.4)$) [draw, fill,circle,inner sep=1.4pt,label=0:{$u(x)=u_i$}] (h5) {};

\draw (h1) -- (l1) -- (h2) -- (l3) -- (h4);
\draw (l3) -- (h5);
\draw [red,dashed] (s) -- (h1) -- (h2) -- (h4) -- (h5);
\draw [blue,dashed] (s) -- (l1) -- (l3);

\node at (H2.corner 2) [red,xshift=-1.25cm,yshift=-0.37cm] {$d_{T_{1i}}(p,u_i)$};
\node at (Hm1.corner 3) [blue,xshift=-.4cm,yshift=0cm] {$d_{T_{1i}}(p,l_i)$};

\draw (h5) -- (H3.corner 1) [->,>=stealth',red];
\node at (H3.side 6) [red,xshift=0.5cm,yshift=0.37cm] {$p_N(x) (-)$};

\draw (l3) -- (H3.corner 4) [->,>=stealth',blue];
\node at (H3.side 3) [blue,xshift=-0.5cm,yshift=-0.3cm] {$p_S(x) (+)$};

\end{tikzpicture}\hspace{-0.2cm}
\begin{tikzpicture}[scale=1, transform shape]
\clip (-1.4,-1.3) rectangle (5.7,2.75);

\node [draw, color=red!25, shape border rotate=30, minimum size=0.7in, regular polygon, regular polygon sides=6] at (0,0.33) (Hm3) {};

\node [draw, color=blue!25, shape border rotate=30, minimum size=0.58in, regular polygon, regular polygon sides=6] at (0.47,0.45) (Hm2) {};

\node [draw,color=red!25, shape border rotate=30, minimum size=0.78in, regular polygon, regular polygon sides=6] at (1.53,-0.05) (Hm1) {};

\node [draw, color=blue!25, shape border rotate=30, minimum size=0.81in, regular polygon, regular polygon sides=6] at (2,0.77) (H1) {};

\node [draw, color=red!25, shape border rotate=30, minimum size=1.2in, regular polygon, regular polygon sides=6] at (2.93,1.18) (H2) {};

\node [draw, color=gray!50, shape border rotate=30, minimum size=1.2in, regular polygon, regular polygon sides=6] at ($(2.93,1.18)+ (1*0.3, -.577*0.3)$) (H3) {};

\node at (Hm3.side 2) [yshift=-0.32cm,draw,circle,fill,inner sep=1.4pt,label=180:{$p$}] (s) {};

\node at (10,3.333) [draw,fill,circle,inner sep=1.4pt,label=180:{$t$}] (t) {};

\draw [gray, dotted] (s) -- (t);

\coordinate (o) at (-1,-.8);
\coordinate (x) at (4,-.8);
\draw [->] (o) -- (x);

\coordinate (N) at (H3.corner 1);
\draw [gray, dashed] (N) -- (N |- 52,-.8);
\node at (N |- 52,-.8) [label={[label distance=-0.05cm]-90:{\color{gray}{$x$}}}] {};
\draw [gray, dashed] (s) -- (s |- 52,-.8);
\node at ($(s |- 52,-.8) + (-0.25,0)$) [label={[label distance=-0.2cm]-90:{\color{gray}{${\tt x}(p)$}}}] {};

\node at (Hm3.corner 6) (a) {};
\node at (Hm3.corner 1) (b) {};
\node at (Hm2.corner 1) (c) {};
\node at (Hm2.corner 2) (d) {};

\node at (intersection of a--b and c--d) [draw,fill,circle,inner sep=1.4pt] (h1) {};

\node at (Hm2.corner 4) (a) {};
\node at (Hm2.corner 5) (b) {};
\node at (Hm1.corner 2) (c) {};
\node at (Hm1.corner 3) (d) {};

\node at (intersection of a--b and c--d) [draw,fill,circle,inner sep=1.4pt] (l1) {};

\node at (1.11,0.7) [draw,fill,circle,inner sep=1.4pt] (h2) {};

\node at (H1.corner 1) (a) {};
\node at (H1.corner 2) (b) {};
\node at (H2.corner 2) (c) {};
\node at (H2.corner 3) (d) {};

\node at (intersection of a--b and c--d) [draw,fill,circle,inner sep=1.4pt] (h4) {};

\node at (H2.corner 3) (a) {};
\node at (H2.corner 4) (b) {};
\node at (H1.corner 4) (c) {};
\node at (H1.corner 5) (d) {};

\node at (intersection of a--b and c--d) [draw,fill,circle,inner sep=1.4pt,label=0:{$\ell(x)=l_i$}] (l3) {};

\node at ($(3.6,2.32) + (1*0.4, -.577*0.4)$) [draw, fill,circle,inner sep=1.4pt,label=0:{$u(x)=u_i$}] (h5) {};

\draw (h1) -- (l1) -- (h2) -- (l3) -- (h4);
\draw (l3) -- (h5);
\draw (s) -- (h1) -- (h2) -- (h4) -- (h5);
\draw (s) -- (l1) -- (l3);

\draw [red,dashed] (s) -- (Hm3.corner 2) -- (Hm3.corner 1) -- (h1) -- (Hm2.corner 1) --(Hm2.corner 6) -- (h2) -- (H1.corner 2) -- (h4) -- (H2.corner 2) -- (H2.corner 1) -- (h5);  

\draw [red] (h5) --  ($(h5) + (0, 0.1cm)$);
\draw ($(h5) + (0, 0.1cm)$) -- ($(H3.corner 1) + (0, 0.1cm)$) [->,>=stealth',red];

\draw [blue,dashed] (s) -- (Hm3.corner 3) -- (Hm3.corner 4) -- (l1) -- (Hm1.corner 3) --(Hm1.corner 4) -- (Hm1.corner 5) -- (l3);  

\draw (l3) -- (H3.corner 4) [->,>=stealth',blue];

\node at (H2.corner 2) [red,xshift=-1.25cm,yshift=-0.37cm] {$\bar{U}(x)-p_N(x)$};
\node at (Hm1.corner 3) [blue,xshift=-.3cm,yshift=-0.5cm] {$\bar{L}(x)-p_S(x)$};

\end{tikzpicture}
}

\newcommand{\transitiononezero}{
\begin{tikzpicture}[scale=1]
\node (h) [draw, color=gray, shape border rotate=30, minimum size=0.22in, regular polygon, regular polygon sides=6] at (0,0) {};
\node at (h.side 2) [draw,circle,fill,inner sep=0.75pt] (s1) {};
\node at (h.side 1) [draw,circle,fill,inner sep=0.75pt] (s2) {};
\draw (s1)--(s2);
\end{tikzpicture}
}

\newcommand{\transitiononefive}{
\begin{tikzpicture}[scale=1]
\node (h) [draw, color=gray, shape border rotate=30, minimum size=0.22in, regular polygon, regular polygon sides=6] at (0,0) {};
\node at (h.side 2) [draw,circle,fill,inner sep=0.75pt] (s1) {};
\node at (h.side 6) [draw,circle,fill,inner sep=0.75pt] (s2) {};
\draw (s1)--(s2);
\end{tikzpicture}
}

\newcommand{\transitiontwoone}{
\begin{tikzpicture}[scale=1]
\node (h) [draw, color=gray, shape border rotate=30, minimum size=0.22in, regular polygon, regular polygon sides=6] at (0,0) {};
\node at (h.side 3) [draw,circle,fill,inner sep=0.75pt] (s1) {};
\node at (h.side 2) [draw,circle,fill,inner sep=0.75pt] (s2) {};
\draw (s1)--(s2);
\end{tikzpicture}
}

\newcommand{\transitiontwozero}{
\begin{tikzpicture}[scale=1]
\node (h) [draw, color=gray, shape border rotate=30, minimum size=0.22in, regular polygon, regular polygon sides=6] at (0,0) {};
\node at (h.side 3) [draw,circle,fill,inner sep=0.75pt] (s1) {};
\node at (h.side 1) [draw,circle,fill,inner sep=0.75pt] (s2) {};
\draw (s1)--(s2);
\end{tikzpicture}
}

\newcommand{\transitiontwofive}{
\begin{tikzpicture}[scale=1]
\node (h) [draw, color=gray, shape border rotate=30, minimum size=0.22in, regular polygon, regular polygon sides=6] at (0,0) {};
\node at (h.side 3) [draw,circle,fill,inner sep=0.75pt] (s1) {};
\node at (h.side 6) [draw,circle,fill,inner sep=0.75pt] (s2) {};
\draw (s1)--(s2);
\end{tikzpicture}
}

\newcommand{\transitiontwofour}{
\begin{tikzpicture}[scale=1]
\node (h) [draw, color=gray, shape border rotate=30, minimum size=0.22in, regular polygon, regular polygon sides=6] at (0,0) {};
\node at (h.side 3) [draw,circle,fill,inner sep=0.75pt] (s1) {};
\node at (h.side 5) [draw,circle,fill,inner sep=0.75pt] (s2) {};
\draw (s1)--(s2);
\end{tikzpicture}
}

\newcommand{\transitionthreeone}{
\begin{tikzpicture}[scale=1]
\node (h) [draw, color=gray, shape border rotate=30, minimum size=0.22in, regular polygon, regular polygon sides=6] at (0,0) {};
\node at (h.side 4) [draw,circle,fill,inner sep=0.75pt] (s1) {};
\node at (h.side 2) [draw,circle,fill,inner sep=0.75pt] (s2) {};
\draw (s1)--(s2);
\end{tikzpicture}
}

\newcommand{\transitionthreezero}{
\begin{tikzpicture}[scale=1]
\node (h) [draw, color=gray, shape border rotate=30, minimum size=0.22in, regular polygon, regular polygon sides=6] at (0,0) {};
\node at (h.side 4) [draw,circle,fill,inner sep=0.75pt] (s1) {};
\node at (h.side 1) [draw,circle,fill,inner sep=0.75pt] (s2) {};
\draw (s1)--(s2);
\end{tikzpicture}
}

\newcommand{\transitionthreefive}{
\begin{tikzpicture}[scale=1]
\node (h) [draw, color=gray, shape border rotate=30, minimum size=0.22in, regular polygon, regular polygon sides=6] at (0,0) {};
\node at (h.side 4) [draw,circle,fill,inner sep=0.75pt] (s1) {};
\node at (h.side 6) [draw,circle,fill,inner sep=0.75pt] (s2) {};
\draw (s1)--(s2);
\end{tikzpicture}
}

\newcommand{\transitionthreefour}{
\begin{tikzpicture}[scale=1]
\node (h) [draw, color=gray, shape border rotate=30, minimum size=0.22in, regular polygon, regular polygon sides=6] at (0,0) {};
\node at (h.side 4) [draw,circle,fill,inner sep=0.75pt] (s1) {};
\node at (h.side 5) [draw,circle,fill,inner sep=0.75pt] (s2) {};
\draw (s1)--(s2);
\end{tikzpicture}
}

\newcommand{\transitionfourzero}{
\begin{tikzpicture}[scale=1]
\node (h) [draw, color=gray, shape border rotate=30, minimum size=0.22in, regular polygon, regular polygon sides=6] at (0,0) {};
\node at (h.side 5) [draw,circle,fill,inner sep=0.75pt] (s1) {};
\node at (h.side 1) [draw,circle,fill,inner sep=0.75pt] (s2) {};
\draw (s1)--(s2);
\end{tikzpicture}
}

\newcommand{\transitionfourfive}{
\begin{tikzpicture}[scale=1]
\node (h) [draw, color=gray, shape border rotate=30, minimum size=0.22in, regular polygon, regular polygon sides=6] at (0,0) {};
\node at (h.side 5) [draw,circle,fill,inner sep=0.75pt] (s1) {};
\node at (h.side 6) [draw,circle,fill,inner sep=0.75pt] (s2) {};
\draw (s1)--(s2);
\end{tikzpicture}
}

\newcommand{\growth}{
\begin{tikzpicture}[scale=1, transform shape]
\clip (1.3,-1) rectangle (6,2.8);

\node [draw, color=red!50, shape border rotate=30, minimum size=1.2in, regular polygon, regular polygon sides=6] at (2.93,1.18) (H1) {};

\node [draw,circle,fill,color=red!50,inner sep=0.5pt] at (2.93,1.18) (c1) {};

\node [draw, color=blue!50, shape border rotate=30, minimum size=1.2in, regular polygon, regular polygon sides=6] at (3.67,0.75) (H2) {};

\node [draw,circle,fill,color=blue!50,inner sep=0.5pt] at (3.67,0.75) (c2) {};

\node at (H2.side 6) [xshift=-0.3cm,yshift=0.18cm,draw,circle,fill,inner sep=1.4pt,label=0:{$u$}] (u) {};

\node at (H2.side 3) [xshift=-0.4cm,yshift=0.24cm,draw,circle,fill,inner sep=1.4pt,label=-180:{$l$}] (l) {};

\draw (u) -- (H2.corner 1) [->,>=stealth',blue];
\draw ($(u) + (-0.1cm,-0.052cm)$) -- ($(H1.corner 1) + (0,-0.1cm)$) [->,>=stealth',red];

\coordinate (c1r) at ($(c1)+(0:4)$);
\coordinate (c2u) at ($(c2)+(90:4)$);

\coordinate (z) at (intersection of c1--c1r and c2--c2u);

\draw [gray,<->,{Stealth-Stealth}] (c1) -- (z)  node [midway, above] {\scriptsize $\Delta x=1$};

\draw [densely dotted, gray] (c2) -- (z);

\coordinate (a) at (H1.corner 1);
\coordinate (b) at (H2.corner 1);

\draw [very thick,gray,<->,{Stealth-Stealth}] (a) -- (b)  node [midway, below left=-2pt] {\scriptsize $\Delta p_N(x)=\frac{2}{\sqrt{3}}$};

\coordinate (ar) at ($(a)+(0:4)$);
\coordinate (bu) at ($(b)+(90:4)$);

\coordinate (z) at (intersection of a--ar and b--bu);

\draw [gray,<->,{Stealth-Stealth}] (b) -- (z)  node [midway, right=-2pt] {\scriptsize $\Delta y(N(x)) = -\frac{1}{\sqrt{3}}$};

\coordinate (a) at (H1.side 2);
\coordinate (b) at (H2.side 2);

\coordinate (ar) at ($(a)+(0:4)$);
\coordinate (bu) at ($(b)+(90:4)$);

\coordinate (z) at (intersection of a--ar and b--bu);

\draw [gray,<->,{Stealth-Stealth}] (a) -- (z)  node [midway, below] {\scriptsize $\Delta w(x) = 1$};

\coordinate (a) at (H1.side 5);
\coordinate (b) at (H2.side 5);

\coordinate (ar) at ($(a)+(-90:4)$);
\coordinate (bu) at ($(b)+(180:4)$);

\coordinate (z) at (intersection of a--ar and b--bu);

\draw [gray,<->,{Stealth-Stealth}] (z) -- (b)  node [midway, below] {\scriptsize $\Delta e(x) = 1$};

\draw (l) -- (H2.corner 4) [->,>=stealth',blue];
\draw ($(l) + (0,-0.1cm)$) -- ($(H1.corner 4) + (0,-0.1cm)$) [->,>=stealth',red];

\coordinate (a) at (H1.corner 4);
\coordinate (b) at (H2.corner 4);

\draw [gray,<->,{Stealth-Stealth}] ($(H1.corner 4) + (0,-0.1cm)$) -- ($(H2.corner 4) + (0,-0.1cm)$)  node [midway, left=3pt] {\scriptsize $\Delta p_S(x)=\frac{2}{\sqrt{3}}$};

\coordinate (ar) at ($(a)+(0:4)$);
\coordinate (bu) at ($(b)+(90:4)$);

\coordinate (z) at (intersection of a--ar and b--bu);

\draw [gray,<->,{Stealth-Stealth}] ($(H2.corner 4) + (0,0.1cm)$) -- (z) node [midway, right=-2pt] {\scriptsize $\Delta y(S(x))=-\frac{3}{\sqrt{3}}$};
\end{tikzpicture}\hspace{-0.8cm}
\begin{tikzpicture}[scale=1, transform shape]
\clip (0.5,-1.55) rectangle (6.5,3.1);

\node [draw, color=red!50, shape border rotate=30, minimum size=1in, regular polygon, regular polygon sides=6] at (2.7,1.3) (H1) {};

\node [draw,circle,fill,color=red!50,inner sep=0.5pt] at (2.7,1.3) (c1) {};

\node [draw, color=blue!50, shape border rotate=30, minimum size=1.54in, regular polygon, regular polygon sides=6] at (3.3,0.965) (H2) {};

\node [draw,circle,fill,color=blue!50,inner sep=0.5pt] at (3.3,0.965) (c2) {};

\node at (H2.side 1) [xshift=-0.3cm,yshift=-0.18cm,draw,circle,fill,inner sep=1.4pt,label=180:{$u$}] (u) {};

\node at (H2.side 2) [draw,circle,fill,inner sep=1.4pt,label=-135:{$l$}] (l) {};

\draw (u) -- (H2.corner 1) [->,>=stealth',blue];
\draw ($(u) + (0.1cm,-0.052cm)$) -- ($(H1.corner 1) + (0,-0.1cm)$) [->,>=stealth',red];

\coordinate (c1r) at ($(c1)+(0:4)$);
\coordinate (c2u) at ($(c2)+(90:4)$);

\coordinate (z) at (intersection of c1--c1r and c2--c2u);

\draw [gray,<->,{Stealth-Stealth}] (c1) -- (z)  node [midway, above] {\scriptsize $\Delta x=1$};

\draw [densely dotted, gray] (c2) -- (z);

\coordinate (a) at (H1.corner 1);
\coordinate (b) at (H2.corner 1);

\draw [very thick,gray,<->,{Stealth-Stealth}] ($(a) + (0,-0.1cm)$) -- ($(b) + (0,-0.1cm)$)  node [midway, right=3pt] {\scriptsize $\Delta p_N(x)=\frac{2}{\sqrt{3}}$};

\coordinate (ar) at ($(a)+(90:4)$);
\coordinate (bu) at ($(b)+(180:4)$);

\coordinate (z) at (intersection of a--ar and b--bu);

\draw [gray,<->,{Stealth-Stealth}] (a) -- (z)  node [midway, left=-2pt] {\scriptsize $\Delta y(N(x)) = \frac{1}{\sqrt{3}}$};

\node at ($(H1.side 2) + (0.1cm,0)$) [gray] {\scriptsize $\Delta w(x) = 0$};

\coordinate (a) at (H1.side 5);
\coordinate (b) at (H2.side 5);

\coordinate (ar) at ($(a)+(-90:4)$);
\coordinate (bu) at ($(b)+(180:4)$);

\coordinate (z) at (intersection of a--ar and b--bu);

\draw [gray,<->,{Stealth-Stealth}] (z) -- (b)  node [midway, below] {\scriptsize $\Delta e(x) = 2$};

\draw (l) -- (H2.corner 3) -- (H2.corner 4) [->,>=stealth',blue];
\draw ($(l) + (-0.052cm,0)$) -- ($(H1.corner 3) + (-0.052cm,-0.1cm)$) -- ($(H1.corner 4) + (0,-0.1cm)$) [->,>=stealth',red];

\coordinate (a) at (H1.corner 4);
\coordinate (b) at (H2.corner 4);

\coordinate (ad) at ($(a)+(-90:4)$);
\coordinate (bl) at ($(b)+(150:4)$);

\coordinate (z) at (intersection of a--ad and b--bl);

\draw [gray,<->,{Stealth-Stealth}] ($(H1.corner 4) + (0,-0.1cm)$) -- ($(z) + (0,0.1cm)$) -- ($(H2.corner 4) + (0,0.1cm)$)  node [yshift=0.3cm, midway, left=5pt] {\scriptsize $\Delta p_S(x)=\frac{4}{\sqrt{3}}$};

\coordinate (ar) at ($(a)+(0:4)$);
\coordinate (bu) at ($(b)+(90:4)$);

\coordinate (z) at (intersection of a--ar and b--bu);

\draw [gray,<->,{Stealth-Stealth}] ($(H2.corner 4) + (0,0.1cm)$) -- (z) node [midway, right=-3pt] {\scriptsize $\Delta y(S(x))=-\frac{3}{\sqrt{3}}$};
\end{tikzpicture}\hspace{-1cm}
\begin{tikzpicture}[scale=1, transform shape]
\clip (1.1,-2.25) rectangle (5.9,2.5);

\node [draw, color=blue!50, shape border rotate=30, minimum size=1.2in, regular polygon, regular polygon sides=6] at (3.94,-0.57) (H2) {};

\node [draw,circle,fill,color=blue!50,inner sep=0.5pt] at (3.94,-0.57) (c1) {};

\node [draw, color=red!50, shape border rotate=30, minimum size=1.66in, regular polygon, regular polygon sides=6] at (3.435,0.305) (H1) {};

\node [draw,circle,fill,color=red!50,inner sep=0.5pt] at (3.435,0.305) (c2) {};

\node at (H1.side 5) [yshift=-0.4cm,draw,circle,fill,inner sep=1.4pt,label=0:{$u$}] (u) {};

\node at (H1.side 3) [xshift=0.5cm,yshift=-0.29cm,draw,circle,fill,inner sep=1.4pt,label=180:{$l$}] (l) {};

\draw ($(u) + (0.052cm,0cm)$) -- ($(H1.corner 6) + (0.052cm,0.1cm)$) -- ($(H1.corner 1) + (0,0.1cm)$) [->,>=stealth',red];
\draw (u) -- (H2.corner 6) -- (H2.corner 1) [->,>=stealth',blue];

\coordinate (c1r) at ($(c1)+(0:4)$);
\coordinate (c2u) at ($(c2)+(90:4)$);

\coordinate (z) at (intersection of c1--c1r and c2--c2u);

\draw [gray,<->,{Stealth-Stealth}] (c1) -- (z)  node [midway, above] {\scriptsize $\Delta x=1$};

\draw [densely dotted, gray] (c2) -- (z);

\coordinate (a) at (H1.corner 1);
\coordinate (b) at (H2.corner 1);

\coordinate (ar) at ($(a)+(-30:4)$);
\coordinate (bu) at ($(b)+(90:4)$);

\coordinate (z) at (intersection of a--ar and b--bu);

\draw [very thick,gray,<->,{Stealth-Stealth}] (a) -- (z) -- (b)  node [near end, right=0pt] {\scriptsize $\Delta p_N(x)=\frac{6}{\sqrt{3}}$};

\coordinate (ad) at ($(a)+(-90:4)$);
\coordinate (bl) at ($(b)+(180:4)$);

\coordinate (z) at (intersection of a--ad and b--bl);

\draw [gray,<->,{Stealth-Stealth}] (a) -- (z)  node [midway, left=-1pt] {\scriptsize $\Delta y(N(x)) = -\frac{1}{\sqrt{3}}$};

\node at (H2.side 5) [gray] {\scriptsize $\Delta e(x) = 0$};

\coordinate (a) at ($(H1.side 2) + (0,-0.2cm)$);
\coordinate (b) at (H2.side 2);

\coordinate (ar) at ($(a)+(0:4)$);
\coordinate (bu) at ($(b)+(90:4)$);

\coordinate (z) at (intersection of a--ar and b--bu);

\draw [gray,<->,{Stealth-Stealth}] (z) -- (a)  node [midway, above] {\scriptsize $\Delta w(x) = 2$};

\draw (l) -- (H2.corner 4) [->,>=stealth',blue];

\draw ($(l) + (0,-0.1cm)$) -- ($(H1.corner 4) + (0,-0.1cm)$) [->,>=stealth',red];

\coordinate (a) at (H1.corner 4);
\coordinate (b) at (H2.corner 4);

\draw [gray,<->,{Stealth-Stealth}] ($(H1.corner 4) + (0,-0.1cm)$) -- ($(H2.corner 4) + (0,-0.1cm)$)  node [midway, above right=-3pt] {\scriptsize $\Delta p_S(x)=\frac{2}{\sqrt{3}}$};

\coordinate (ad) at ($(a)+(-90:4)$);
\coordinate (bl) at ($(b)+(180:4)$);

\coordinate (z) at (intersection of a--ad and b--bl);

\draw [gray,<->,{Stealth-Stealth}] ($(H1.corner 4) + (0,-0.1cm)$) -- ($(z) + (0,-0.1cm)$) node [midway, left=0pt] {\scriptsize $\Delta y(S(x))=-\frac{3}{\sqrt{3}}$};

\end{tikzpicture}
}

\newcommand{\badintervals}{
\begin{tikzpicture}[scale=1, transform shape]
\clip (.6,-.67) rectangle (3.3,1.4);

\node [draw, color=red!50, shape border rotate=30, minimum size=0.6in, regular polygon, regular polygon sides=6] at (1.46,0.59) (H1) {};

\coordinate (c1) at (1.46,0.59);

\node [draw, color=blue!50, shape border rotate=30, minimum size=0.4in, regular polygon, regular polygon sides=6] at (2.83,0.49) (H2) {};

\coordinate (c2) at (2.83,0.49);

\node at (H1.side 6) [xshift=-0.15cm,yshift=0.09cm,draw,circle,fill,inner sep=1.4pt,label=0:{$u(x_l)$}] (ul) {};

\node at (H1.side 2) [yshift=-0.12cm,draw,circle,fill,inner sep=1.4pt,label=0:{$\ell(x_l)$}] (l) {};

\draw [->] (0.85,-0.35) -- (3.25,-0.35);

\coordinate (N1) at (H1.corner 1);
\coordinate (N2) at (H2.corner 1);

\draw [gray, dashed] (N1) -- (1.46,-0.35);
\draw [gray, dashed] (N2) -- (2.83,-0.35);

\node at (1.46,-0.35) [label={[label distance=-0.15cm]-90:{\color{gray}{$x_l$}}}] {};
\node at (2.83,-0.35) [label={[label distance=-0.15cm]-90:{\color{gray}{$x_r$}}}] {};

\draw (ul) -- (l);

\end{tikzpicture}\hspace{0.75cm}
\begin{tikzpicture}[scale=1, transform shape]
\clip (.6,-.67) rectangle (3.3,1.4);

\node [draw, color=red!50, shape border rotate=30, minimum size=0.6in, regular polygon, regular polygon sides=6] at (1.46,0.59) (H1) {};

\coordinate (c1) at (1.46,0.59);

\node [draw, color=blue!50, shape border rotate=30, minimum size=0.4in, regular polygon, regular polygon sides=6] at (2.83,0.49) (H2) {};

\coordinate (c2) at (2.83,0.49);

\node at (H1.side 2) [yshift=0.12cm,draw,circle,fill,inner sep=1.4pt,label=0:{$u(x_l)$}] (ul) {};

\node at (H1.side 4) [xshift=-0.15cm,yshift=-0.09cm,draw,circle,fill,inner sep=1.4pt,label=0:{$\ell(x_l)$}] (l) {};

\draw [->] (0.85,-0.35) -- (3.25,-0.35);

\coordinate (N1) at (H1.corner 1);
\coordinate (N2) at (H2.corner 1);

\draw [gray, dashed] (N1) -- (1.46,-0.35);
\draw [gray, dashed] (N2) -- (2.83,-0.35);

\node at (1.46,-0.35) [label={[label distance=-0.15cm]-90:{\color{gray}{$x_l$}}}] {};
\node at (2.83,-0.35) [label={[label distance=-0.15cm]-90:{\color{gray}{$x_r$}}}] {};

\draw (ul) -- (l);

\end{tikzpicture}\hspace{0.75cm}
\begin{tikzpicture}[scale=1, transform shape]
\clip (1,-.67) rectangle (3.6,1.4);

\node [draw, color=red!50, shape border rotate=30, minimum size=0.6in, regular polygon, regular polygon sides=6] at (2.83,0.49) (H2) {};

\coordinate (c1) at (1.46,0.59);

\node [draw, color=blue!50, shape border rotate=30, minimum size=0.4in, regular polygon, regular polygon sides=6] at (1.46,0.59) (H1) {};

\coordinate (c2) at (2.83,0.49);

\node at (H2.side 5) [yshift=0.12cm,draw,circle,fill,inner sep=1.4pt,label=180:{$u(x_r)$}] (ul) {};

\node at (H2.side 3) [yshift=-0.09cm,xshift=0.15cm,draw,circle,fill,inner sep=1.4pt,label=180:{$\ell(x_r)$}] (l) {};

\draw [->] (0.85,-0.35) -- (3.25,-0.35);

\coordinate (N1) at (H1.corner 1);
\coordinate (N2) at (H2.corner 1);

\draw [gray, dashed] (N1) -- (1.46,-0.35);
\draw [gray, dashed] (N2) -- (2.83,-0.35);

\node at (1.46,-0.35) [label={[label distance=-0.15cm]-90:{\color{gray}{$x_l$}}}] {};
\node at (2.83,-0.35) [label={[label distance=-0.15cm]-90:{\color{gray}{$x_r$}}}] {};

\draw (ul) -- (l);

\end{tikzpicture}\hspace{0.75cm}
\begin{tikzpicture}[scale=1, transform shape]
\clip (1,-.67) rectangle (3.6,1.4);

\node [draw, color=red!50, shape border rotate=30, minimum size=0.6in, regular polygon, regular polygon sides=6] at (2.83,0.49) (H2) {};

\coordinate (c1) at (1.46,0.59);

\node [draw, color=blue!50, shape border rotate=30, minimum size=0.4in, regular polygon, regular polygon sides=6] at (1.46,0.59) (H1) {};

\coordinate (c2) at (2.83,0.49);

\node at (H2.side 1) [xshift=0.15cm,yshift=0.09cm,draw,circle,fill,inner sep=1.4pt,label=180:{$u(x_r)$}] (ul) {};

\node at (H2.side 5) [yshift=-0.12cm,draw,circle,fill,inner sep=1.4pt,label=180:{$\ell(x_r)$}] (l) {};

\draw [->] (0.85,-0.35) -- (3.25,-0.35);

\coordinate (N1) at (H1.corner 1);
\coordinate (N2) at (H2.corner 1);

\draw [gray, dashed] (N1) -- (1.46,-0.35);
\draw [gray, dashed] (N2) -- (2.83,-0.35);

\node at (1.46,-0.35) [label={[label distance=-0.15cm]-90:{\color{gray}{$x_l$}}}] {};
\node at (2.83,-0.35) [label={[label distance=-0.15cm]-90:{\color{gray}{$x_r$}}}] {};

\draw (ul) -- (l);

\end{tikzpicture}
}

\newcommand{\gentlepath}{
\begin{tikzpicture}[scale=1, transform shape]

\node [draw, color=red!50, shape border rotate=30, minimum size=1.2in, regular polygon, regular polygon sides=6] at (2.93,1.18) (H1) {};

\coordinate (c1) at (2.93,1.18);

\node at (H1.side 6) [xshift=-0.3cm,yshift=0.18cm,draw,circle,fill,inner sep=1.4pt,label={[label distance=0pt]0:$u(x)$}] (ul) {};

\node at ($(ul)+(0,-0.85)$) [inner sep=0pt,outer sep=-0.2pt] (x) {};
\draw [gray,<->,{Stealth-Stealth}] (x) -- ($(x)+(0.95,0)$) node [below left] {\scriptsize $\delta_f(x)$};

\draw [gray,dashed] (ul) -- (ul |- 1, 1.65);

\node at (H1.side 2) [yshift=-0.24cm,draw,circle,fill,inner sep=1.4pt,label={[label distance=-3pt]135:$\ell(x)$}] (l) {};

\draw [thick] (ul) -- (l);

\node at (H1.corner 3) [draw,gray,circle,fill,inner sep=0.5pt,label=-180:{$w$}] (w) {};

\draw [->] (1,0.1) -- (4.5,0.1);

\coordinate (N1) at (H1.corner 1);

\node at (N1) [label={[label distance=-0.3cm]135:{\scriptsize $N(x)$}}] {};

\draw [gray] (ul) -- (w) -- (N1);

\draw [gray, dashed] (N1) -- (2.93,0.1);

\draw pic [draw=gray, "\color{gray}{\scriptsize $\theta$}", angle eccentricity=1.25,angle radius=1.05cm] {angle = ul--w--N1};

\node at (2.93,0.1) [label={[label distance=-0.15cm]-90:{\color{gray}{$x$}}}] {};

\draw [gray,<->,{Stealth-Stealth},label=90:${\tt r}(x)$] ($(w)+(2.64,0.2)$$) -- ($(w)+(1.32,0.2)$) node [midway, above] {\scriptsize ${\tt r}(x)$};

\draw (l) -- (H1.corner 3) -- (H1.corner 4) [->,>=stealth',blue,thick];
\draw (ul) -- (H1.corner 1) [->,>=stealth',red,thick];

\node at (H1.side 6) [red,xshift=0.75,yshift=0.6cm] {$p_N(x) (-)$};

\node at (H1.side 3) [blue,xshift=-0.8cm,yshift=-0.2cm] {$p_S(x) (+)$};

\end{tikzpicture}\hspace{1cm}
\begin{tikzpicture}[scale=1, transform shape]

\node [draw, color=red!50, shape border rotate=30, minimum size=1.2in, regular polygon, regular polygon sides=6] at (2.93,1.18) (H1) {};

\coordinate (c1) at (2.93,1.18);

\node [draw, color=blue!50, shape border rotate=30, minimum size=0.8in, regular polygon, regular polygon sides=6] at (5.67,0.99) (H2) {};

\coordinate (c2) at (5.67,0.99);

\node at (H1.side 6) [xshift=-0.3cm,yshift=0.18cm,draw,circle,fill,inner sep=1.4pt,label=45:{$u(x_l)=u_i$}] (ul) {};

\node at (H1.side 2) [yshift=-0.24cm,draw,circle,fill,inner sep=1.4pt,label={[label distance=-3pt]135:$l_i = \ell(x_l)$}] (l) {};

\draw [gray,<->,{Stealth-Stealth}] ($(ul)+(0,-1.5)$) -- ($(ul)+(0.95,-1.5)$) node [midway, below] {\scriptsize $\delta_f(x_l)$};

\node at (H1.corner 3) (w) {};

\draw [->] (1,0.1) -- (7,0.1);

\coordinate (N1) at (H1.corner 1);
\coordinate (N2) at (H2.corner 1);

\coordinate (N12) at ($(ul)+(30:0.25)$);
\coordinate (N13) at ($(N12)+(-30:0.5)$);
\coordinate (N14) at ($(N13)+(-90:.35)$);
\coordinate (N15) at ($(N14)+(-30:0.8)$);

\coordinate (x) at ($(N15)+(30:.66)$);

\coordinate (xd) at ($(x)+(30:6)$);
\coordinate (N2d) at ($(N2)+(150:4)$);

\coordinate (z) at (intersection of N2--N2d and x--xd);

\draw [densely dotted,thick] (N1) -- (ul) -- (N12) -- (N13) -- (N14) -- (N15) -- (x) -- (z) -- (N2);

\coordinate (N1d) at ($(N1)+(-90:3)$);
\coordinate (ud) at ($(ul)+(-90:3)$);
\coordinate (N12d) at ($(N12)+(-90:3)$);
\coordinate (N15d) at ($(N15)+(-90:3)$);
\coordinate (zd) at ($(z)+(-90:3)$);
\coordinate (N2d) at ($(N2)+(-90:3)$);

\coordinate (N1b) at ($(N1)+(0,-0.9)$);
\coordinate (N1r) at ($(N1b)+(0:10)$);

\coordinate (z1) at (intersection of N1--N1d and N1b--N1r);
\coordinate (z2) at (intersection of ul--ud and N1b--N1r);
\coordinate (z3) at (intersection of N12--N12d and N1b--N1r);
\coordinate (z4) at (intersection of N15--N15d and N1b--N1r);
\coordinate (z5) at (intersection of z--zd and N1b--N1r);
\coordinate (z6) at (intersection of N2--N2d and N1b--N1r);

\draw [gray!50, densely dashed] (ul) -- ($(ul)+(0,-1.5)$);
\draw [gray!50, densely dashed] (N12) -- (z3);
\draw [gray!50, densely dashed] (z) -- (z5);
\draw [gray!50, densely dashed] (z4) -- (N15);

\draw [blue, very thick] (z1) -- (z2);
\draw [blue, very thick] (z3) -- (z4);
\draw [blue, very thick] (z5) -- (z6);

\node at (N1) [label={[label distance=-0.3cm]135:{\scriptsize $N(x_l)$}}] {};
\node at (N2) [label={[label distance=-0.3cm]45:{\scriptsize $N(x_r)$}}] {};

\draw [gray, dashed] (N1) -- (2.93,0.1);
\draw [gray, dashed] (N2) -- (5.67,0.1);

\node at (2.93,0.1) [label={[label distance=-0.15cm]-90:{\color{gray}{$x_l$}}}] {};
\node at (5.67,0.1) [label={[label distance=-0.15cm]-90:{\color{gray}{$x_r$}}}] {};

\draw [gray,<->,{Stealth-Stealth},label=90:${\tt r}(x_l)$] ($(w)$) -- ($(w)+(1.32,0)$) node [midway, above] {\scriptsize ${\tt r}(x_l)$};

\draw (ul) -- (l);

\node at (H2.side 6) [xshift=+0.3cm,yshift=-0.18cm,draw,circle,fill,inner sep=1.4pt,label={[label distance=0pt]30:$u_{i+s} = u(x_r)$}] (ur) {};

\draw [densely dotted,thick] (N2) -- (ur);

\coordinate (lr) at ($(l)+(0:6)$);
\coordinate (urd) at ($(ur)+(-90:4)$);

\coordinate (z) at (intersection of l--lr and ur--urd);

\draw [gray,<->,{Stealth-Stealth}] (ur) -- (z)  node [midway, right=2pt] {\scriptsize $y(u(x_r)) - y(\ell(x_l))$};

\draw (ul) -- (l);

\end{tikzpicture}
}

\newcommand{\linearworstcase}{
\begin{tikzpicture}[scale=1, transform shape]

\node [draw, color=red!50, shape border rotate=30, minimum size=1.2in, regular polygon, regular polygon sides=6] at (2.93,1.18) (H1) {};

\coordinate (c1) at (2.93,1.18);

\node [draw, color=blue!50, shape border rotate=30, minimum size=0.8in, regular polygon, regular polygon sides=6] at (5.67,0.99) (H2) {};

\coordinate (c2) at (5.67,0.99);

\node at (H1.side 6) [xshift=-0.3cm,yshift=0.18cm,draw,circle,fill,inner sep=1.4pt,label=0:{$u(x_l)$}] (ul) {};

\draw [gray,<->,{Stealth-Stealth},label=90:$\theta_f(x_l)$] ($(ul)+(0,-1.1)$) -- ($(ul)+(0.96,-1.1)$) node [below left] {\scriptsize $\delta_f(x_l)$};

\draw [gray,<->,{Stealth-Stealth},label=90:$\theta_f(x_l)$] ($(ul)+(-0.96,-1.1)$) -- ($(ul)+(0,-1.1)$) node [below left] {\scriptsize $\delta_f(x_l)$};

\draw [gray, dashed] (ul) -- ($(ul)+(0,-1.1)$);

\node at (H1.side 2) [yshift=-0.24cm,draw,circle,fill,inner sep=1.4pt,label=180:{$\ell(x_l)$}] (l) {};

\draw [gray,<->,{Stealth-Stealth},label=90:$\theta_f(x_l)$] ($(l)+(0,-0.5)$) -- ($(l)+(0.36,-0.5)$) node [very near start, below] {\scriptsize ${\tt r}(x_l) - \delta_f(x_l)$};

\node at (H2.side 5) [yshift=+0.18cm,draw,circle,fill,inner sep=1.4pt,label=45:{$u(x_r)$}] (ur) {};

\coordinate (N12) at ($(ul)+(30:0.25)$);

\coordinate (N13) at ($(N12)+(-30:0.35)$);
\coordinate (N14) at ($(N13)+(-90:.35)$);
\coordinate (N15) at ($(N14)+(-30:0.8)$);

\coordinate (x) at ($(N15)+(30:.66)$);

\coordinate (xd) at ($(x)+(30:6)$);
\coordinate (N2d) at ($(N2)+(150:4)$);

\coordinate (z) at (intersection of N2--N2d and x--xd);

\draw [densely dotted] (N1) -- (ul) -- (N12) -- (N13) -- (N14) -- (N15) -- (x) -- (z) -- (N2);

\draw [->] (1,0.1) -- (7,0.1);

\coordinate (N1) at (H1.corner 1);
\coordinate (N2) at (H2.corner 1);

\draw [gray, dashed] (N1) -- (2.93,0.1);
\draw [gray, dashed] (N2) -- (5.67,0.1);

\node at (2.93,0.1) [label={[label distance=-0.15cm]-90:{\color{gray}{$x_l$}}}] {};
\node at (5.67,0.1) [label={[label distance=-0.15cm]-90:{\color{gray}{$x_r$}}}] {};

\draw [blue,thick,<->,{Stealth-Stealth}] ($($(l)+(0.72,0)$)$) -- ($($(l)+(2.36,0.5)$)+(2.59,-0.5)$) node [midway, above] {};

\draw [red,thick,<->,{Stealth-Stealth},label=90:${\tt r}(x_l)$] ($($(l)+(0.36,-0.5)$)$) -- ($($(l)+(2.36,-0.5)$)+(1.7,0)$) node [midway, above] {};

\draw (ul) -- (l);
\end{tikzpicture}
\begin{tikzpicture}[scale=1, transform shape]

\node [draw, color=red!50, shape border rotate=30, minimum size=1.2in, regular polygon, regular polygon sides=6] at (2.93,1.18) (H1) {};

\coordinate (c1) at (2.93,1.18);

\node [draw, color=blue!50, shape border rotate=30, minimum size=0.8in, regular polygon, regular polygon sides=6] at (5.67,0.99) (H2) {};

\coordinate (c2) at (5.67,0.99);

\node at (H1.side 6) [xshift=-0.3cm,yshift=0.18cm,draw,circle,fill,inner sep=1.4pt,label=0:{$u(x_l)$}] (ul) {};

\draw [gray,<->,{Stealth-Stealth},label=90:$\theta_f(x_l)$] ($(ul)+(0,-1.65)$) -- ($(ul)+(0.96,-1.65)$) node [below left] {\scriptsize $\delta_f(x_l)$};
\draw [gray,<->,{Stealth-Stealth},label=90:$\theta_f(x_l)$] ($(ul)+(-0.96,-1.65)$) -- ($(ul)+(0,-1.65)$) node [below left] {\scriptsize $\delta_f(x_l)$};

\draw [gray, dashed] (ul) -- ($(ul)+(0,-1.65)$);

\node at (H1.side 2) [yshift=-0.24cm,draw,circle,fill,inner sep=1.4pt,label=180:{$\ell(x_l)$}] (l) {};

\node at (H2.corner 2) [draw,circle,fill,inner sep=1.4pt,label=0:{$u(x_r)$}] (ur) {};

\coordinate (N12) at ($(N1)+(-30:0.1)$);

\coordinate (N13) at ($(N12)+(-90:1)$);

\coordinate (xd) at ($(N13)+(-30:6)$);
\coordinate (N2d) at ($(N2)+(-150:4)$);

\coordinate (z) at (intersection of N2--N2d and N13--xd);

\draw [densely dotted] (N1) -- (N12) -- (N13) -- (z) -- (N2);

\coordinate (zd) at ($(z)+(0,-5)$);
\coordinate (N2d) at ($(N2)+(6,0)$);

\coordinate (zzz) at (intersection of zd--z and N2--N2d);

\draw [red,dashed,thick,<->,{Stealth-Stealth}] (N2) -- (zzz);

\draw [gray, dashed] (z) -- (zzz);

\draw [->] (1,0.1) -- (7,0.1);

\coordinate (N1) at (H1.corner 1);
\coordinate (N2) at (H2.corner 1);

\draw [gray, dashed] (N1) -- (2.93,0.1);

\draw [gray, dashed] (N2) -- (5.67,0.1);

\node at (2.93,0.1) [label={[label distance=-0.15cm]-90:{\color{gray}{$x_l$}}}] {};
\node at (5.67,0.1) [label={[label distance=-0.15cm]-90:{\color{gray}{$x_r$}}}] {};

\draw [blue,thick,<->,{Stealth-Stealth}] ($($(l)+(0.72,-0.5)$)$) -- ($($(l)+(1.47,-0.5)$)+(1.7,0)$) node [midway, above] {};

\draw (ul) -- (l);
\end{tikzpicture} 
\vspace{-0.5cm}\center{(a) \hspace{2.6in} (b)}
}

\newcommand{\induction}{
\begin{tikzpicture}[scale=1, transform shape]

\node at (0,0.5) [draw,circle,fill,inner sep=1.4pt,label=180:{$p$}] (p) {};
\coordinate (pu) at ($(p)+(0,1)$);

\draw [gray, dashed] (p) -- (pu);

\coordinate (c1) at (2.93,1.18);

\node [draw, color=red!50, shape border rotate=30, minimum size=1.2in, regular polygon, regular polygon sides=6] at (c1) (H1) {};

\node at (H1.side 6) [xshift=-0.3cm,yshift=0.18cm,draw,circle,fill,inner sep=1.4pt,label=0:{$u(x)$}] (ul) {};

\draw [gray,<->,{Stealth-Stealth}] ($(ul)+(0,-0.7)$) -- ($(ul)+(0.96,-0.7)$) node [below left] {\scriptsize $\delta_f(x)$};

\draw [gray, dashed] (ul) -- ($(ul)+(0,-0.7)$);

\node at (H1.side 2) [yshift=-0.24cm,draw,circle,fill,inner sep=1.4pt,label=180:{$\ell(x)$}] (l) {};

\coordinate (ld) at ($(l) + (0,-0.35)$);

\coordinate (z) at (1.97,1.5);
\coordinate (z2) at (2.93,1.377);

\draw [->] (-0.1,0.1) -- (4.5,0.1);

\coordinate (N1) at (H1.corner 1);

\draw [gray, dashed] (N1) -- (2.93,0.11);

\node at (2.93,0.1) [label={[label distance=-0.15cm]-90:{\color{gray}{{\scriptsize $x$}}}}] {};

\draw [blue,thick,<->,{Stealth-Stealth}] (pu) -- (z) node [midway, above] {};
\draw [blue,dotted](pu) -- ($(pu)+(-0.25,0)$) node [label={[label distance=-0.15cm]180:$\frac{10}{\sqrt{3}} - 2$}]{};

\draw [blue,thick,<->,{Stealth-Stealth}] (z2) -- ($(z2)+(-0.96,0)$) node [below right] {\scriptsize $\delta_f(x)$};
\draw [blue,dotted](z2) -- ($(z2)+(1.64,0)$) node [label={[label distance=-0.15cm]0:$\frac{6}{\sqrt{3}}$}]{};

\draw [blue, dashed] (z) -- ($(z2) + (-0.96,0)$);

\draw [gray,<->,{Stealth-Stealth}] ($(H1.side 5)+(0,-.5)$) -- ($(c1)+(0,-.5)$) node [midway, below] {\scriptsize ${\tt r}(x)$};

\end{tikzpicture}\hspace{0.65cm}\begin{tikzpicture}[scale=1, transform shape]

\node at (1,0.5) [draw,circle,fill,inner sep=1.4pt,label=180:{$p$}] (p) {};
\coordinate (pu) at ($(p)+(0,1)$);

\draw [gray, dashed] (p) -- (pu);

\coordinate (c1) at (2.93,1.18);

\coordinate (z) at (2.93, 1.5);
\coordinate (z3) at (3.89, .223);

\node [draw, color=red!50, shape border rotate=30, minimum size=1.2in, regular polygon, regular polygon sides=6] at (c1) (H1) {};

\node at (H1.side 1) [xshift=0.3cm,yshift=0.18cm,draw,circle,fill,inner sep=1.4pt,label=180:{$u(x)$}] (ul) {};

\draw [gray,<->,{Stealth-Stealth}] ($(ul)+(0,-0.6)$) -- ($(ul)+(-0.96,-0.6)$) node [below right] {\scriptsize $\delta_b(x)$};

\draw [gray, dashed] (ul) -- ($(ul)+(0,-0.6)$);

\node at (H1.side 5) [yshift=-0.24cm,draw,circle,fill,inner sep=1.4pt,label=0:{$\ell(x)$}] (l) {};

\coordinate (ld) at ($(l) + (0,+0.35)$);

\draw [->] (-0.1,0.1) -- (4.5,0.1);

\coordinate (N1) at (H1.corner 1);

\draw [gray, dashed] (N1) -- (2.93,0.11);

\node at (2.93,0.1) [label={[label distance=-0.15cm]-90:{\color{gray}{{\scriptsize $x$}}}}] {};

\draw [blue,thick,<->,{Stealth-Stealth}] (pu) -- (z) node [midway, above] {};
\draw [blue,dotted](pu) -- ($(pu)+(-0.25,0)$) node [very near start, label={[label distance=0.1cm]180:$\frac{10}{\sqrt{3}} - 2$}]{};

\draw [blue,thick,<->,{Stealth-Stealth}] (z3) -- ($(z3)+(-0.96,0)$)  node [above right] {\scriptsize $\delta_b(x)$};
\draw [blue,dotted](z3) -- ($(z3)+(.8,0)$) node [label={[label distance=-0.15cm]0:$\frac{4}{\sqrt{3}}-2$}]{};

\draw [gray,<->,{Stealth-Stealth}] ($(H1.side 2)+(0,-.5)$) -- ($(c1)+(0,-.5)$) node [midway, below] {\scriptsize ${\tt r}(x)$};

\end{tikzpicture}
\center{(a) \hspace{2.6in} (b)}

}

\newcommand{\proofA}{
\center\begin{tikzpicture}[scale=1, transform shape]

\draw [->] (1,0.1) -- (6.4,0.1);

\node at (0,0.5) [draw,circle,fill,inner sep=1.4pt,label=180:{$p$}] (p) {};
\coordinate (pu) at ($(p)+(0,1)$);

\draw [gray, dashed] (p) -- (pu);

\coordinate (c1) at (2.93,1.18);

\node [draw, color=red!50, shape border rotate=30, minimum size=1.2in, regular polygon, regular polygon sides=6] at (c1) (H1) {};

\node at (H1.side 6) [xshift=-0.3cm,yshift=0.18cm,draw,circle,fill,inner sep=1.4pt,label=45:{$u(x_l)$}] (ul) {};

\draw [gray,<->,{Stealth-Stealth}] ($(ul)+(0,-0.8)$) -- ($(ul)+(0.96,-0.8)$) node [above left] {\scriptsize $\delta_f(x_l)$};

\draw [gray, dashed] (ul) -- ($(ul)+(0,-0.6)$);

\node at (H1.side 2) [yshift=-0.24cm,draw,circle,fill,inner sep=1.4pt,label=180:{$\ell(x_l)$}] (l) {};

\coordinate (z) at (1.97, 1.5);
\coordinate (z2) at (1.97, 1.38);
\coordinate (z3) at (2.93, 1.38);
\coordinate (z4) at (2.93, 1.5);

\draw [gray, dashed] (z) -- (z2);

\coordinate (N1) at (H1.corner 1);

\draw [gray, dashed] (N1) -- (2.93,0.11);

\node at (2.93,0.1) [label={[label distance=-0.15cm]-90:{\color{gray}{{\scriptsize $x_l$}}}}] {};

\draw [blue,thick,<->,{Stealth-Stealth}] ($(z)+(-1.96,0)$) -- (z) node [midway, above] {};
\draw [blue,dotted]($(z)+(-2.1,0)$) -- ($(z)+(-2,0)$) node [label={[label distance=-0.1cm]180:$\frac{10}{\sqrt{3}} - 2$}]{};

\draw [blue,thick,<->,{Stealth-Stealth}] (z2) -- (z3) node [below left] {\scriptsize $\delta_f(x_l)$};

\draw [gray,<->,{Stealth-Stealth}] ($(H1.side 5)+(0,-.5)$) -- ($(c1)+(0,-.5)$) node [midway, above] {\scriptsize ${\tt r}(x_l)$};

\node [draw, color=blue!50, shape border rotate=30, minimum size=0.8in, regular polygon, regular polygon sides=6] at (5.67,0.99) (H2) {};

\coordinate (c2) at (5.67,0.99);

\coordinate (z5) at (4.78, 1.5);
\coordinate (z6) at (4.78, 1.38);
\coordinate (z7) at (5.67, 1.38);

\draw [red,thick,<->,{Stealth-Stealth}] (z4) -- (z5);

\draw [blue,dotted](z) -- (z4) node {};

\draw [red,dotted](z7) -- (z3) node {};
\draw [red,dotted](z7) -- ($(z7)+(1,0)$) node [label={[label distance=-0.15cm]0:$\frac{6}{\sqrt{3}}$}]{};

\node at (5.67,0.1) [label={[label distance=-0.2cm]-45:{\color{gray}{{\scriptsize $x_r$}}}}] {};
\draw [gray, dashed] (N2) -- (5.67,0.1);

\draw [red,thick,<->,{Stealth-Stealth}] (z6) -- (z7) node [below left] {\scriptsize ${\tt r}(x_r)$};

\coordinate (rwx) at ($(H2.corner 3) + (0,-0.37)$);
\draw [dashed,gray] (H2.corner 3) -- (rwx); 
\node at (rwx) [label={[label distance=-0.125cm,gray]-90:{{\scriptsize $w(x_r)$}}}] {};

\end{tikzpicture}
}

\newcommand{\proofB}{
\begin{tikzpicture}[scale=1, transform shape]

\node at (0.4,0.5) [draw,circle,fill,inner sep=1.4pt,label=180:{$p$}] (p) {};
\coordinate (pu) at ($(p)+(0,1)$);

\draw [gray, dashed] (p) -- (pu);

\coordinate (c1) at (2.93,1.18);
\coordinate (c2) at ($(c1)+(.9,.5196)$);

\node [draw, color=red!50, shape border rotate=30, minimum size=1.2in, regular polygon, regular polygon sides=6] at (c1) (H1) {};

\node [draw, color=blue!50, shape border rotate=30, minimum size=1.2in, regular polygon, regular polygon sides=6] at (c2) (H2) {};

 \node at (H2.corner 2) [draw,circle,fill,inner sep=1.4pt,label={[label distance=-0.15cm]135:{$u_s=u_t$}}] (u) {};

 \node at (H1.corner 5) [draw,circle,fill,inner sep=1.4pt,label={[label distance=0.15cm]0:$l_s$}] (l) {};

\coordinate (N1) at (H1.corner 1);
\coordinate (N1x) at (2.93,0.1);

\draw [gray, dashed] (N1) -- (N1x);

\node at (N1x) [label={[label distance=-0.15cm]-90:{\color{gray}{{\scriptsize $x_l$}}}}] {};

\coordinate (N2) at (H2.corner 1);
\coordinate (N2x) at (3.83,0.1);

\draw [gray, dashed] (N2) -- (N2x);

\node at (N2x) [label={[label distance=-0.15cm]-90:{\color{gray}{{\scriptsize $x_r$}}}}] {};

\draw [gray,<->,{Stealth-Stealth}] ($(u)+(0,-0.8)$) -- ($(u)+(-0.89,-0.8)$) node [midway,above] {\scriptsize $\delta_b(x_l)$};
\draw [gray, dashed] (u) -- ($(u)+(0,-0.8)$);

\draw [gray,<->,{Stealth-Stealth}] ($(l)+(0,0.8)$) -- ($(l)+(0.89,+0.8)$) node [midway,above] {\scriptsize $\delta_f(x_r)$};
\draw [gray, dashed] (l) -- ($(l)+(0,0.8)$);

\coordinate (pr) at ($(pu)+(0:6)$);
\coordinate (w1) at (H1.side 2);
\coordinate (w1u) at ($(w1)+(90:6)$);
\coordinate (ud) at ($(u)+(-90:6)$);
\coordinate (lu) at ($(l)+(90:6)$);

\coordinate (z1) at (intersection of w1--w1u and pu--pr);
\coordinate (z3) at (intersection of u--ud and pu--pr);
\coordinate (z4) at (intersection of N1--N1x and pu--pr);
\coordinate (z2) at ($(z1)+(z4)-(z3)$);
\coordinate (z5) at (intersection of N2--N2x and pu--pr);
\coordinate (z6) at (intersection of l--lu and pu--pr);

\coordinate (o) at (0.1,0.1);
\coordinate (x) at (4.5,0.1);
\draw [->] (o) -- (x);

\draw [very thick] ($(z4)+(0,-1.41)$) -- ($(z5)+(0,-1.41)$) node [midway, label={[label distance=-0.1cm]-90:$I$}]{};

\draw [blue,thick,<->,{Stealth-Stealth}] (pu) -- (z4) node [midway, above] {};
\draw [blue,dotted](pu) -- ($(pu)+(-0.2,0)$) node [very near start, label={[label distance=-0.1cm]90:$\frac{10}{\sqrt{3}} - 2$}]{};

\draw [blue,thick,<->,{Stealth-Stealth}] ($(z4)+(0,-1.28)$) -- ($(z5)+(0,-1.28)$) node [midway,above] {\scriptsize $\delta_b(x_l)$};
\draw [blue,dotted]($(z4)+(-1.3,-1.28)$) -- ($(z4)+(0,-1.28)$) node [very near start,label={[label distance=0cm]180:$\frac{4}{\sqrt{3}}-2$}]{};

\draw [red,thick,<->,{Stealth-Stealth}] ($(z4)+(0,-.45)$) -- ($(z5)+(0,-.45)$) node [midway,above] {\scriptsize $\delta_f(x_r)$};
\draw [red,dotted]($(z4)+(0,-.45)$) -- ($(z4)+(-.6,-.45)$) node [very near end, label={[label distance=-0.15cm]180:$\frac{4}{\sqrt{3}}$}]{};

\coordinate (e) at (H1.corner 5);
\coordinate (ed) at ($(e) + (0,-1)$);
\coordinate (ex) at (intersection of o--x and e--ed);

\coordinate (w) at (H2.corner 3);
\coordinate (wd) at ($(w) + (0,-1)$);
\coordinate (wx) at (intersection of o--x and w--wd);

\end{tikzpicture}\hspace{0.6cm}
\begin{tikzpicture}
\coordinate (c1) at (2.93,1.18);
\coordinate (c2) at ($(c1)+(2.9,-.15196)$);

\node [draw, color=red!50, shape border rotate=30, minimum size=1.2in, regular polygon, regular polygon sides=6] at (c1) (H1) {};

\node [draw, color=blue!50, shape border rotate=30, minimum size=1in, regular polygon, regular polygon sides=6] at (c2) (H2) {};

 \node at (H1.side 6) [draw,shift={+(-0.075,0.0433cm)},circle,fill,inner sep=1.4pt,label={[label distance=0]180:{$u_i$}}] (ui) {};

 \node at (H1.side 2) [draw,yshift=-0.3cm,circle,fill,inner sep=1.4pt,label={[label distance=0cm]-135:$l_i$}] (li) {};

 \node at (H2.side 1) [draw,shift={+(0.15,0.0866cm)},circle,fill,inner sep=1.4pt,label={[label distance=0]90:{$u_j$}}] (uj) {};

 \node at (H2.side 5) [draw,circle,fill,inner sep=1.4pt,label={[label distance=0.15cm]0:$l_j$}] (lj) {};

\coordinate (o) at (1.5,0.1);
\coordinate (x) at (7.5,0.1);
\draw [->] (o) -- (x);

\coordinate (N1) at (H1.corner 1);
\coordinate (N1d) at ($(N1)+(-90:4)$);
\coordinate (N1x) at (intersection of o--x and N1--N1d);

\draw [gray, dashed] (N1) -- (N1x);

\node at (N1x) [label={[label distance=-0.15cm]-90:{\color{gray}{{\scriptsize $x_l$}}}}] {};

\coordinate (N2) at (H2.corner 1);
\coordinate (N2d) at ($(N2)+(-90:4)$);
\coordinate (N2x) at (intersection of o--x and N2--N2d);

\draw [gray, dashed] (N2) -- (N2x);

\node at (N2x) [label={[label distance=-0.15cm]-90:{\color{gray}{{\scriptsize $x_r$}}}}] {};

\node at (0.4,0.5) [draw,circle,fill,inner sep=1.4pt,label=180:{$p$}] (p) {};
\coordinate (pu) at ($(p)+(0,1)$);
\coordinate (pu2) at ($(p)+(0,0.88)$);
\coordinate (pu3) at ($(p)+(0,-0.28)$);
\draw [gray, dashed] (p) -- (pu);

\coordinate (pr) at ($(pu)+(0:6)$);
\coordinate (p2r) at ($(pu2)+(0:6)$);
\coordinate (p3r) at ($(pu3)+(0:6)$);

\coordinate (liu) at ($(li)+(90:6)$);
\coordinate (lid) at ($(li)+(-90:6)$);
\coordinate (uid) at ($(ui)+(-90:6)$);
\coordinate (ujd) at ($(uj)+(-90:6)$);
\coordinate (lju) at ($(lj)+(90:6)$);
\coordinate (ljd) at ($(lj)+(-90:6)$);

\coordinate (z1) at (intersection of li--liu and pu--pr);
\coordinate (z3) at (intersection of N1--N1x and pu--pr);
\coordinate (z4) at (intersection of ui--uid and pu--pr);
\coordinate (z2) at ($(z1)+(z4)-(z3)$);

\draw [blue,thick,<->,{Stealth-Stealth}] (pu) -- (z2) node [midway, above] {};
\draw [blue,dotted](pu) -- ($(pu)+(-0.2,0)$) node [very near start, label={[label distance=-0.1cm]90:$\frac{10}{\sqrt{3}} - 2$}]{};

\coordinate (z2d) at ($(z2)+(-90:6)$);
\coordinate (z5) at (intersection of z2--z2d and pu2--p2r);
\coordinate (z6) at (intersection of N1--N1x and pu2--p2r);

\draw [blue,thick,<->,{Stealth-Stealth}] (z5) -- (z6);
\draw [blue,dotted]($(z6)+(0.5,0)$) -- ($(z6)+(0,0)$) node [very near start,label={[label distance=-0.30cm]-30:$\frac{6}{\sqrt{3}}$}]{};

\coordinate (z7) at (intersection of N2--N2x and pu--pr);

\draw [red,thick,<->,{Stealth-Stealth}] (z3) -- (z7) node {};
\draw [red,dotted](z3) -- (z2) node {};

\coordinate (z8) at (intersection of z2--z2d and pu3--p3r);
\coordinate (z9) at (intersection of N1--N1x and pu3--p3r);

\draw [red,thick,<->,{Stealth-Stealth}] (z8) -- (z9) node [midway,above] {\scriptsize $\delta_f(x_l)$};
\draw [red,dotted]($(z9)+(0,0)$) -- ($(z9)+(1.1,0)$) node [very near end, label={[label distance=0.]0:$\frac{4}{\sqrt{3}}-2$}]{};

\coordinate (z10) at (intersection of uj--ujd and pu3--p3r);
\coordinate (z11) at (intersection of N2--N2x and pu3--p3r);
\coordinate (z12) at (intersection of lj--ljd and pu3--p3r);
\coordinate (z13) at ($(z12)-(z11)+(z10)$);

\draw [red,thick,<->,{Stealth-Stealth}] (z11) -- (z13) node [midway,above] {\scriptsize $\delta_b(x_r)$};
\draw [red,dotted]($(z11)+(0,0)$) -- ($(z11)+(-.5,0)$);

\coordinate (z13u) at ($(z13)+(90:6)$);

\draw [gray,<->,{Stealth-Stealth}] ($(z4)+(0.74,0.15)$) -- ($(z4)+(0,0.15)$) node [midway,above] {\scriptsize $\delta_f(x_l)$};
\draw [gray,<->,{Stealth-Stealth}] ($(z10)+(0,1.45)$) -- ($(z10)+(-0.7,1.45)$) node [midway, above]{\scriptsize $\delta_b(x_l)$};

\draw [gray, dashed] (z8) -- (z2);
\draw [gray, dashed] ($(z4)+(0,0.2)$) -- (ui);
\draw [gray, dashed] ($(z10)+(0,1.5)$) -- (uj);

\end{tikzpicture}
\center{(a) \hspace{2.6in} (b)}
}

\newcommand{\pathcost}{
\begin{tikzpicture}

\coordinate (c1) at (2.93,1.18);
\coordinate (c2) at ($(c1)+(2.9,-.15196)$);

\node [draw, color=red!50, shape border rotate=30, minimum size=1.2in, regular polygon, regular polygon sides=6] at (c1) (H1) {};

\node [draw, color=blue!50, shape border rotate=30, minimum size=1in, regular polygon, regular polygon sides=6] at (c2) (H2) {};

 \node at (H1.side 6) [draw,circle,fill,inner sep=1.4pt,label={[label distance=-0.1cm]-85:{$u_i$}}] (ui) {};

 \node at (H1.side 2) [draw,yshift=-0.5cm,circle,fill,inner sep=1.4pt,label={[label distance=0cm]135:$l_i$}] (li) {};

 \node at (H2.side 1) [draw,shift={+(0.15,0.0866cm)},circle,fill,inner sep=1.4pt,label={[label distance=0]-90:{$u_j$}}] (uj) {};

 \node at (H2.side 5) [draw,circle,fill,inner sep=1.4pt,label={[label distance=0cm]45:$l_j$}] (lj) {};

\draw (li) -- (ui);
\draw (lj) -- (uj);

\coordinate (o) at (1.5,0.1);
\coordinate (x) at (7.5,0.1);
\draw [->] (o) -- (x);

\coordinate (N1) at (H1.corner 1);
\coordinate (N1d) at ($(N1)+(-90:4)$);
\coordinate (N1x) at (intersection of o--x and N1--N1d);

\draw [gray, dashed] (N1) -- (N1x);

\node at (N1x) [label={[label distance=-0.15cm]-90:{\color{gray}{{\scriptsize $x_l$}}}}] {};

\coordinate (N2) at (H2.corner 1);
\coordinate (N2d) at ($(N2)+(-90:4)$);
\coordinate (N2x) at (intersection of o--x and N2--N2d);

\draw [gray, dashed] (N2) -- (N2x);

\node at (N2x) [label={[label distance=-0.15cm]-90:{\color{gray}{{\scriptsize $x_r$}}}}] {};

\coordinate (uid) at ($(ui)+(-90:6)$);
\coordinate (ujd) at ($(uj)+(-90:6)$);

\coordinate (p1) at (0,1.3);
\coordinate (p2) at (7,1.3);

\coordinate (z3) at (intersection of ui--uid and p1--p2);
\coordinate (z4) at (intersection of N1--N1x and p1--p2);
\coordinate (z5) at (intersection of N2--N2x and p1--p2);
\coordinate (z6) at (intersection of uj--ujd and p1--p2);

\draw [gray,<->,{Stealth-Stealth}] ($(z6)-(0.7,0)$) -- (z6) node [midway,below] {\scriptsize $\delta_b(x_l)$};

\draw [gray, dashed] (z6) -- (uj);
\draw [gray, dashed] (z3) -- (ui);

\draw [gray,<->,{Stealth-Stealth}] (z3) -- ($(z3)+(0.65,0)$) node [midway,below] {\scriptsize $\delta_f(x_r)$};

\node at (H1.side 2) [yshift=-0.5cm,inner sep=1.4pt,label=-135:{\color{blue}{\large +}}] {};
\node at (H1.side 6) [yshift=-0.05cm,inner sep=1.4pt,label=90:{\color{red}{\large $-$}}] {};
\node at (H1.side 3) [inner sep=1.4pt,label=-135:{\color{blue}{$p_S(l_i,x_l)$}}] {};
\node at (H1.side 6) [xshift=0.4cm, yshift=0.2cm,inner sep=1.4pt,label=90:{\color{red}{$p_N(u_i,x_l)$}}] {};
\node at (H2.side 5) [yshift=-0.3cm,inner sep=1.4pt,label=-45:{\color{red}{\large $-$}}] {};
\node at (H2.side 4) [inner sep=1.4pt,label=-45:{\color{red}{$p_S(l_j,x_r)$}}] {};
\node at (H2.side 1) [xshift=0.2cm,yshift=0.1cm,inner sep=1.4pt,label=90:{\color{blue}{\large $+$}}] {};
\node at (H2.side 1) [xshift=1.2cm,yshift=0.3cm,inner sep=1.4pt,label=90:{\color{blue}{$p_N(u_j,x_r)$}}] {};

\draw ($(li) + (-0.1cm,0)$) -- ($(H1.corner 3) + (-0.1cm,-0.052cm)$) -- ($(H1.corner 4) + (0,-0.1cm)$) [->,>=stealth',blue];
\draw ($(ui) + (0,0.1cm)$) -- ($(H1.corner 1) + (0,0.1cm)$) [->,>=stealth',red];
\draw ($(lj) + (0.1cm,0)$) -- ($(H2.corner 5) + (0.1cm,-0.052cm)$) -- ($(H2.corner 4) + (0,-0.1cm)$) [->,>=stealth',red];
\draw ($(uj) + (0cm,0.1cm)$) -- ($(H2.corner 1) + (0cm,0.1cm)$) [->,>=stealth',blue];

\node at ($(ui)+(30:0.25)$) [draw,circle,fill,inner sep=1.4pt,label={[label distance=0cm]0:{$u_{i+1}$}}] (u2) {};

\coordinate (u3) at ($(u2)+(-30:0.5)$);
\coordinate (u4) at ($(u3)+(30:.15)$);
\coordinate (u5) at ($(u4)+(-30:0.4)$);
\coordinate (x6) at ($(u5)+(30:2)$);
\coordinate (x7) at ($(uj)+(150:2)$);
\coordinate (u6) at (intersection of uj--x7 and u5--x6)
; 
\draw (ui) -- (u2) -- (u3) -- (u4) -- (u5) -- (u6) -- (uj);

\end{tikzpicture}
}

\newcommand{\anotherpathcost}{
\begin{tikzpicture}

\coordinate (c1) at (2.93,1.18);
\coordinate (c2) at ($(c1)+(2.9,-.15196)$);

\node [draw, color=red!50, shape border rotate=30, minimum size=1.2in, regular polygon, regular polygon sides=6] at (c1) (H1) {};

\node [draw, color=blue!50, shape border rotate=30, minimum size=1in, regular polygon, regular polygon sides=6] at (c2) (H2) {};

 \node at (H1.side 6) [draw,circle,fill,inner sep=1.4pt,label={[label distance=-0.1cm]-85:{$u_i$}}] (ui) {};

 \node at (H1.side 2) [draw,yshift=-0.5cm,circle,fill,inner sep=1.4pt,label={[label distance=0cm]135:$l_i$}] (li) {};

 \node at (H2.side 1) [draw,shift={+(0.15,0.0866cm)},circle,fill,inner sep=1.4pt,label={[label distance=0]-90:{$u_j$}}] (uj) {};

 \node at (H2.side 5) [draw,circle,fill,inner sep=1.4pt,label={[label distance=0cm]45:$l_j$}] (lj) {};

\draw (li) -- (ui);
\draw (lj) -- (uj);

\coordinate (o) at (1.5,0.1);
\coordinate (x) at (7.5,0.1);
\draw [->] (o) -- (x);

\coordinate (N1) at (H1.corner 1);
\coordinate (N1d) at ($(N1)+(-90:4)$);
\coordinate (N1x) at (intersection of o--x and N1--N1d);

\draw [gray, dashed] (N1) -- (N1x);

\node at (N1x) [label={[label distance=-0.15cm]-90:{\color{gray}{{\scriptsize $x_l$}}}}] {};

\coordinate (N2) at (H2.corner 1);
\coordinate (N2d) at ($(N2)+(-90:4)$);
\coordinate (N2x) at (intersection of o--x and N2--N2d);

\draw [gray, dashed] (N2) -- (N2x);

\node at (N2x) [label={[label distance=-0.15cm]-90:{\color{gray}{{\scriptsize $x_r$}}}}] {};

\coordinate (uid) at ($(ui)+(-90:6)$);
\coordinate (ujd) at ($(uj)+(-90:6)$);

\coordinate (p1) at (0,1.3);
\coordinate (p2) at (7,1.3);

\coordinate (z3) at (intersection of ui--uid and p1--p2);
\coordinate (z4) at (intersection of N1--N1x and p1--p2);
\coordinate (z5) at (intersection of N2--N2x and p1--p2);
\coordinate (z6) at (intersection of uj--ujd and p1--p2);

\node at (H1.side 2) [yshift=-0.5cm,inner sep=1.4pt,label=-135:{\color{blue}{\large +}}] {};
\node at (H1.side 6) [yshift=-0.05cm,inner sep=1.4pt,label=90:{\color{red}{\large $-$}}] {};
\node at (H1.side 3) [inner sep=1.4pt,label=-135:{\color{blue}{$p_S(l_i,x_l)$}}] {};
\node at (H1.side 6) [xshift=0.4cm, yshift=0.2cm,inner sep=1.4pt,label=90:{\color{red}{$p_N(u_i,x_l)$}}] {};
\node at (H2.side 5) [yshift=-0.3cm,inner sep=1.4pt,label=-45:{\color{red}{\large $-$}}] {};
\node at (H2.side 4) [inner sep=1.4pt,label=-45:{\color{red}{$p_S(l_j,x_r)$}}] {};
\node at (H2.side 1) [xshift=0.2cm,yshift=0.1cm,inner sep=1.4pt,label=90:{\color{blue}{\large $+$}}] {};
\node at (H2.side 1) [xshift=1.2cm,yshift=0.3cm,inner sep=1.4pt,label=90:{\color{blue}{$p_N(u_j,x_r)$}}] {};

\draw ($(li) + (-0.1cm,0)$) -- ($(H1.corner 3) + (-0.1cm,-0.052cm)$) -- ($(H1.corner 4) + (0,-0.1cm)$) [->,>=stealth',blue];
\draw ($(ui) + (0,0.1cm)$) -- ($(H1.corner 1) + (0,0.1cm)$) [->,>=stealth',red];
\draw ($(lj) + (0.1cm,0)$) -- ($(H2.corner 5) + (0.1cm,-0.052cm)$) -- ($(H2.corner 4) + (0,-0.1cm)$) [->,>=stealth',red];
\draw ($(uj) + (0cm,0.1cm)$) -- ($(H2.corner 1) + (0cm,0.1cm)$) [->,>=stealth',blue];

\node at ($(ui)+(30:0.25)$) [draw,circle,fill,inner sep=1.4pt,label={[label distance=0cm]0:{$u_{i+1}$}}] (u2) {};

\coordinate (u3) at ($(u2)+(-30:0.5)$);
\coordinate (u4) at ($(u3)+(30:.15)$);
\coordinate (u5) at ($(u4)+(-30:0.4)$);
\coordinate (x6) at ($(u5)+(30:2)$);
\coordinate (x7) at ($(uj)+(150:2)$);
\coordinate (u6) at (intersection of uj--x7 and u5--x6)
; 
\draw (ui) -- (u2) -- (u3) -- (u4) -- (u5) -- (u6) -- (uj);

\end{tikzpicture}
}

\newcommand{\mickey}{
\begin{tikzpicture}

\node [draw, color=blue!50, shape border rotate=30, minimum size=2in, regular polygon, regular polygon sides=6] at (2,0) (H) {};

\node [draw, color=red!50, shape border rotate=30, minimum size=1.4641in, regular polygon, regular polygon sides=6,label=-90:{{\color{red!50}{$H_1$}}}] at (1.41,1.02) (L) {};

\node [draw, color=green!50, shape border rotate=30, minimum size=1.4641in, regular polygon, regular polygon sides=6,label=-87:{{\color{green!50}{$H_n$}}}] at (2.59,1.02) (R) {};

\node at (H.corner 1) [draw,circle,fill,inner sep=1.4pt] (u4) {};
\node at (H.corner 3) [draw,circle,fill,inner sep=1.4pt] (l2) {};
\node at (H.corner 4) [draw,circle,fill,inner sep=1.4pt] (l3) {};
\node at (H.corner 5) [draw,circle,fill,inner sep=1.4pt] (l4) {};

\node at (L.corner 3) [draw,circle,fill,inner sep=1.4pt,label=180:$l_1$] (l1) {};
\node at (R.corner 5) [draw,circle,fill,inner sep=1.4pt,label=0:$l_{n-1}$] (l5) {};

\node at (L.corner 2) [draw,circle,fill,inner sep=1.4pt,label=180:$u_1$] (u2) {};
\node at (L.corner 1) [draw,circle,fill,inner sep=1.4pt] (u3) {};
\node at (R.corner 1) [draw,circle,fill,inner sep=1.4pt] (u5) {};
\node at (R.corner 6) [draw,circle,fill,inner sep=1.4pt,label=0:$u_{n-1}$] (u6) {};

\node at (L.corner 4) (LN) {};
\coordinate (LNL) at ($(LN)+(-180:6)$);
\coordinate (m) at (intersection of LN--LNL and u2--l2);

\draw [gray, dashed] (m) -- (LN) node [midway, label={[label distance=-0.1cm]-90:{\color{gray}{$1$}}}] {};

\coordinate (l4L) at ($(l4)+(-180:6)$);
\coordinate (m) at (intersection of l4--l4L and u4--l3);

\draw [gray, dashed] (m) -- (l4) node [near start, label={[label distance=-0.1cm]-90:{\color{gray}{$\frac{\sqrt{3}}{2} + \frac{1}{2}$}}}] {};

\coordinate (h) at (u3);
\coordinate (l) at (l1);

\draw [very thick] (h) -- (l) node [midway, label={[label distance=-0.1cm]180:{\color{gray}{$2$}}}] {};

\foreach \x in {.8,.6,.4,.2}
{
  \node at ($\x*(l1)+(l2)-\x*(l2)$) [draw, fill, inner sep = 1.4pt, circle] (l) {};
  \draw (h) -- (l);
  \node at ($\x*(u3)+(u4)-\x*(u4)$) [draw, fill, inner sep = 1.4pt, circle] (h) {};
  \draw (h) -- (l);
}

\draw (h) -- (l2) -- (u4) -- (l3);

\coordinate (h) at (u5);
\coordinate (l) at (l5);

\draw [very thick] (h) -- (l);

\foreach \x in {.8,.6,.4,.2}
{
  \node at ($\x*(l5)+(l4)-\x*(l4)$) [draw, fill, inner sep = 1.4pt, circle] (l) {};
  \draw (h) -- (l);
  \node at ($\x*(u5)+(u4)-\x*(u4)$) [draw, fill, inner sep = 1.4pt, circle] (h) {};
  \draw (h) -- (l);
}

\draw (h) -- (l4) -- (u4);
\draw [very thick] (u4) -- (u5) -- (u6) -- (l5) -- (l4) -- (l3) -- (l2) node [near end,label={[label distance=0cm]-90:{\color{gray}{$1 + \frac{1}{\sqrt{3}}$}}}] {} -- (u2) -- (u3) node [near start,label={[label distance=0cm]90:{\color{gray}{$\frac{2}{\sqrt{3}}$}}}] {} -- (u4);

\node at ($0.133975*(u2)+(l1)-0.133975*(l1)$) [draw, fill, inner sep = 1.4pt, circle,label=175:$s$] (p) {};
\node at ($0.133975*(u6)+(l5)-0.133975*(l5)$) [draw, fill, inner sep = 1.4pt, circle,label=5:$t$] (q) {};

\draw [dashed] (p) -- (q);

\draw (u2) -- (p) node [midway,label={[label distance=0cm]180:{\color{gray}{$1$}}}] {};
\end{tikzpicture}
}

\begin{document}


\maketitle

\begin{abstract}
The problem of computing the exact stretch factor (i.e., the tight bound on the
worst case stretch factor) of a Delaunay triangulation is one of the
longstanding open problems in computational geometry. Over the
years, a series of upper and lower bounds on the exact stretch factor have
been obtained but the gap between them is still large. An alternative approach
to solving the problem is to develop techniques for computing the exact stretch
factor of ``easier'' types of Delaunay triangulations, in particular those
defined using regular-polygons instead
of a circle. Tight bounds exist for Delaunay triangulations defined using
an equilateral triangle~\cite{Chew89} and a square~\cite{BGHP15}. In this paper,
we determine the exact stretch factor of Delaunay triangulations defined using
a regular hexagon: It is $2$.

We think that the main contribution of this paper are the two techniques we
have developed to compute tight upper bounds for the stretch factor of
Hexagon-Delaunay triangulations.
\end{abstract}


\section{Introduction}
\label{sec:intro}
In this paper we consider the problem of computing a tight bound
on the worst case stretch factor of a Delaunay triangulation. 

Given a set $P$ of points on the plane, the Delaunay triangulation $T$ on $P$
is a plane graph such that for every pair $u,v \in P$, $(u,v)$ is an edge of
$T$ if and only if there is a {\em circle} passing through $u$ and $v$ with no
point of $P$ in its interior\footnote{This definition assumes that points in
$P$ are in general position which we discuss in Section~\ref{sec:prelim}.}. 
In this paper, we refer to Delaunay triangulations defined using the
{\em circle} as
$\ocircle$-Delaunay triangulations. The $\ocircle$-Delaunay triangulation 
$T$ of $P$ is a plane subgraph of the complete, weighted Euclidean graph
$\cE^P$ on $P$ in which the weight of an edge is the Euclidean distance between
its endpoints. Graph $T$ is also a {\em spanner}, defined as a subgraph of
$\cE^P$ with
the property that the distance in the subgraph between any pair of points is
no more than a constant multiplicative ratio of the distance in $\cE^P$ between
the points. The constant ratio is referred to as the {\em stretch factor} (or
{\em spanning ratio}) of the spanner.

The problem of computing a tight bound on the worst case stretch factor of the
$\ocircle$-Delaunay triangulation has been open for more than three decades. 
In the 1980s, when $\ocircle$-Delaunay triangulations were not known to be
spanners, Chew considered related, ``easier'' structures. In
1986~\cite{Chew86}, Chew proved that a $\square$-Delaunay
triangulation---defined using a {\em fixed-orientation square}
instead of a circle---is a 
spanner with stretch factor at most $\sqrt{10}$. Following this, Chew
proved that the $\triangle$-Delaunay
triangulation---defined using a
{\em fixed-orientation equilateral triangle}
--- has a stretch factor of~$2$~\cite{Chew89}. Significantly, this bound is
tight: one can construct $\triangle$-Delaunay
triangulations with stretch factor arbitrarily close to 2.
Finally, Dobkin et al.~\cite{DFS90} showed that the $\ocircle$-Delaunay
triangulation is a spanner as well. The bound on the stretch factor they obtained was
subsequently improved by Keil and Gutwin~\cite{KG92} as shown in
Table~\ref{ta:related}.
The bound by Keil and Gutwin stood unchallenged
for many years until Xia recently improved the bound to below
2~\cite{Xia13}.

\begin{table}[!b]
\caption{Key stretch factor upper bounds (tight bounds are bold).}
\centering
\begin{tabular}{llr}
{\bf Paper} &  {\bf Graph} & {\bf Upper Bound} \\ \hline
\cite{DFS90} & $\ocircle$-Delaunay & $\pi(1+\sqrt{5})/2 \approx 5.08$ \\ \hline
\cite{KG92} & $\ocircle$-Delaunay & $4\pi/(3\sqrt{3}) \approx 2.41$ \\ \hline
\cite{Xia13} & $\ocircle$-Delaunay & $1.998$ \\ \hline \hline
\cite{Chew89} & $\triangle$-Delaunay & $\mathbf{2}$ \\ \hline \hline
\cite{Chew86} & $\square$-Delaunay & $\sqrt{10} \approx 3.16$ \\ \hline
\cite{BGHP15} & ${\square}$-Delaunay \hspace{3ex} &
$\mathbf{\sqrt{4+2\sqrt{2}}\approx 2.61}$ \\ \hline \hline
{\bf [This paper]} & {\Large\varhexagon}-Delaunay \hspace{3ex} &
$\mathbf{2}$\\ \hline
\end{tabular}
\label{ta:related}
\end{table}

On the lower bound side, some progress has been made on bounding the worst
case stretch factor of a $\ocircle$-Delaunay triangulation. The trivial lower
bound of $\pi/2 \approx 1.5707$ has been improved to
$1.5846$~\cite{BDLSV11} and then to $1.5932$~\cite{XZ11}. 

After three decades of research, we know that the worst case
stretch factor of $\ocircle$-Delaunay triangulations is somewhere between
$1.5932$ and $1.998$. Unfortunately, the techniques that have been
developed so far seem inadequate for proving a tight stretch factor bound.

Rather than attempting to improve further the bounds on
the stretch factor of $\ocircle$-Delaunay triangulations, we follow an
alternative approach. Just like Chew turned to $\triangle$- and
$\square$-Delaunay triangulations to develop insights useful for showing
that $\ocircle$-Delaunay triangulations are spanners, we make use of Delaunay
triangulations defined using regular polygons to develop techniques for
computing tight stretch factor bounds. Delaunay triangulations based on
regular polygons are known to be spanners (Bose et al.~\cite{BCCS08}).
Tight bounds are known for $\triangle$-Delaunay triangulations~\cite{Chew89}
and also for $\square$-Delaunay triangulations (Bonichon et al.~\cite{BGHP15})
as shown in Table~\ref{ta:related}. 

In this paper, we show that the worst case stretch factor of
{\Large\varhexagon}-Delaunay triangulations is 2. We present an overview of
our proof in Section~\ref{sec:main}. 
\iftoggle{abstract} 
{The overview makes use of three lemmas 
whose detailed proofs are ommitted; the proofs (briefly discussed in 
Sections~\ref{sec:gentle}, \ref{sec:proofA}, and  \ref{sec:proofB}) are
in the full version of the paper (Appendix).}
{The overview makes use of three
lemmas whose detailed proofs are in Sections~\ref{sec:gentle}, \ref{sec:proofA}, and  \ref{sec:proofB}.} 
We think that our main
contribution consists of two techniques that we use to compute tight upper 
bounds  on the stretch factor of particular types of
{\Large\varhexagon}-Delaunay triangulations. In Section~\ref{sec:conclusion}
we review the role of the techniques in the paper and explore their potential
to be applied to other kinds of Delaunay triangulations.

\section{Preliminaries}
\label{sec:prelim}
We consider a finite set $P$ of points in the two-dimensional plane with an
orthogonal coordinate system. The $x$- and $y$-coordinates of a point $p$ will
be denoted by ${\tt x}(p)$ and ${\tt y}(p)$, respectively. The
Euclidean graph $\cE^P$ of $P$ is the complete weighted graph embedded in the
plane whose nodes are identified with the points of $P$. For
every pair of nodes $p$ and $q$, the edge $(p,q)$ represents the
segment $[pq]$ and the weight of $(p,q)$ is the Euclidean
distance between $p$ and $q$ which is 
$d_2(p,q) = \sqrt{({\tt x}(p)-{\tt x}(q))^2+({\tt y}(p)-{\tt y}(q))^2}$.
Our arguments also use the $x$-coordinate distance between 
$p$ and $q$ which we denote as $d_x(p,q) = |{\tt x}(p)-{\tt x}(q)|$.

Let $T$ be a subgraph of $\cE^P$. The length of a path in $T$ is the sum of the
weights of the edges of the path and the distance $d_T(p,q)$ in $T$
between two points $p$ and $q$ is the length of the shortest
path in $T$ between them. $T$ is a $t$-spanner for some constant $t>0$ 
if for every pair of points $p,q$ of $P$, $d_T(p,q) \leq t \cdot d_2(p,q)$. The
constant $t$ is referred to as the {\em stretch factor} of $T$.

We define a {\em family of spanners} to be a set of graphs $T^P,$
one for every finite set $P$ of points in the plane, such that for some
constant $t>0$, every $T^P$ is a $t$-spanner of $\cE^P$. 
We say that the stretch factor $t$ is exact (tight) for the family
(or that the worst case stretch factor is $t$) if for every $\epsilon > 0$
there exists a set of points $P$ such that $T^P$ is {\em not} a
$(t-\epsilon)$-spanner of $\cE^P$. 

The families of spanners we consider are various types of Delaunay
triangulations on a set $P$ of points in the plane. Given a set $P$ of
points on the plane, we say that a convex, closed, simple curve in the plane is
empty if it contains no point of $P$ in its interior. The $\ocircle$-Delaunay
triangulation $T$ on $P$ is defined as follows: For every pair $u,v \in P$,
$(u,v)$ is an edge of $T$ if and only if there is an empty {\em circle} passing through
$u$ and $v$. (This definition assumes that
points are in general position which in the case of $\ocircle$-Delaunay
triangulations means that no four points of $P$ are co-circular.) If,
in the definition, 
{\em circle} is replaced by {\em fixed-orientation square} (e.g., a square
whose sides are axis-parallel) or by {\em fixed-orientation equilateral
triangle} then different triangulations are obtained: the $\square$-
and the $\triangle$-Delaunay triangulations.





If, in the definition of the $\ocircle$-Delaunay triangulation, we change
{\em circle} to {\em fixed-orientation regular hexagon}, then a
{\Large\varhexagon}-Delaunay triangulation is obtained. In this paper
we focus on such triangulations. While any fixed orientation of the
hexagon is possible, we choose w.l.o.g. the orientation that has two sides of
the hexagon parallel to the $y$-axis as shown in Fig.~\ref{fig:lower}-(a).
In the remainder of the paper, {\em hexagon} will always refer to a regular
hexagon with such an orientation. We find it useful to label the vertices of
the hexagon $N$, $E_N$, $E_S$, $S$, $W_S$, and $W_N$, in clockwise order and
starting with the
top one. We also label the sides $n_e$, $e$, $s_e$, $s_w$, $w$, and $n_w$ as
shown in  Fig.~\ref{fig:lower}-(a); we will sometimes refer to the $s_e$
and $s_w$ sides as the $s$ sides and to the $n_e$ and $n_w$ sides as the $n$
sides.

\iftoggle{abstract}
{\begin{figure}}
{\begin{figure}[!b]}
\center{\lowerbound

\hspace{-0.4cm} (a) \hspace{4.6cm} (b) \hspace{4.9cm} (c)
}
\caption{(a) The hexagon orientation and the side and vertex labels that we
use (b) A {\Large\varhexagon}-Delaunay triangulation with points
$p$, $q$, $p_k$, and $q_0$ having coordinates $(0,0)$,
$(1, \frac{1}{\sqrt{3}})$, $(\delta, \frac{2}{\sqrt{3}}-\sqrt{3}\delta)$,
and $(1-\delta, -\frac{1}{\sqrt{3}}+\sqrt{3}\delta)$, respectively.
\iftoggle{abstract}{For $\delta$ small enough, $d_T(p,q) \geq (2-\epsilon)d_2(p,q)$.}{} (c) A closer look at the bottom faces of this triangulation.}

\label{fig:lower}
\end{figure}

The definition of the {\Large\varhexagon}-Delaunay triangulation 
assumes that no four points lie on the boundary of an empty
hexagon. Our arguments also assume that no two points lie on a line whose
slope matches the slope of a side of the hexagon (i.e. slopes 
$\infty, \frac{1}{\sqrt{3}}, -\frac{1}{\sqrt{3}}$). 
The {\em general position} assumption we therefore make in this paper consists
of the above two restrictions. This assumption is made solely for the purpose of
simplifying the presentation; the arguments in the paper could be extended so
the results apply to all {\Large\varhexagon}-Delaunay triangulations.
Finally, we need to be aware that unlike the $\ocircle$-Delaunay
triangulation on $P$, the {\Large\varhexagon}-Delaunay (and also the
$\square$- and  $\triangle$-Delaunay) triangulation on $P$ may not contain
all edges on the convex hull of $P$. To handle this and simplify our
arguments, we add to $P$ six additional points, very close to but not exactly 
(in order to satisfy the above assumptions) at coordinates $(0,\pm M)$ and 
$(\pm M \cos(\pi/6), \pm M \sin(\pi/6))$ where $M > 50\max_{s,t \in P}d_2(s,t)$.
The {\Large\varhexagon}-Delaunay triangulation on this modified set of points
$P$, consisting of the original triangulation plus additional
edges between the new points and original points and also between the new
points themselves,  
includes the edges on the convex hull of $P$. Also, any path in
this triangulation between two points $s$ and $t$ from the original set $P$
with length bounded by $2d_2(s,t)$ cannot possibly use the added points. Thus
a proof of our main result for the modified triangulation will also be a proof
for the original one and so we assume that the {\Large\varhexagon}-Delaunay
triangulation on $P$ includes the edges on the convex hull of $P$.

We end this section with a lower bound, by Bonichon~\cite{Bonichon}, on the
worst case stretch factor of {\Large\varhexagon}-Delaunay triangulations.
\iftoggle{abstract}
{The
lower bound construction is illustrated in Fig.~\ref{fig:lower}-(b) and
Fig.~\ref{fig:lower}-(c). The proof is ommited but appears in the Appendix.
}{}
\begin{lemma}
\label{le:lower_bound}
For every $\eps > 0 $, there exists a set $P$ of points in the
plane such that the {\Large\varhexagon}-Delaunay triangulation on $P$
has stretch factor at least  $2 - \eps$.
\end{lemma}

\iftoggle{abstract}
{}{\begin{proof}
Let $k$ be some positive integer and let points $p=p_0$, $q=q_k$, $p_k$, and
$q_0$ have coordinates $(0,0)$, $(1, \frac{1}{\sqrt{3}})$,
$(\delta, \frac{2}{\sqrt{3}}-\sqrt{3}\delta)$,
and $(1-\delta, -\frac{1}{\sqrt{3}}+\sqrt{3}\delta)$, respectively, where 
$\delta = \frac{1}{k+2}$ (see Fig.~\ref{fig:lower}-(b)). Additional $k-1$ 
points $p_1,\dots,p_{k-1}$ are placed on line segment $[p_0p_k]$ and another
$k-1$ points $q_1,\dots,q_{k-1}$ on line segment $[q_0q_k]$ so that all segments
$[p_{i-1}p_i]$ and $[q_{i-1}q_i]$, for $i=1,\dots,k$,  have equal length. 

For every $i = 1, 2, \dots, k$, if $H_i$ is the hexagon of minimum width
$1-\delta$ with $p_{i-1}$ and $p_{i}$ on its $w$ and $n_w$ sides,
then $q_{i-1}$ is exactly the $E_S$ vertex of $H_i$
(e.g., refer to $p_0$, $p_1$, $q_0$, and $H_1$ in Fig.~\ref{fig:lower}-(c)).
This means that all points $q_j$ with
$j \not= i-1$ as well as all points $p_j$ with $j \not= i-1, i$ lie outside
of $H_i$. Therefore, for every $i = 1, 2, \dots, k$, points $p_{i-1}$, $p_i$,
and $q_{i-1}$ define a triangle in the {\Large\varhexagon}-Delaunay
triangulation $T$ on $P$. A similar argument shows that for every
$i = 1, 2, \dots, k$, points $q_{i-1}$, $q_i$, and $p_i$ define a triangle
in $T$ and so the triangulation is as shown in Fig.~\ref{fig:lower}-(b).

A shortest path from $p$ to $q$ in $T$ is, for example,
$p=p_{0}, p_1, \dots, p_k, q$ and so $d_T(p,q) = d_2(p_0,p_k) + d_2(p_k,q_k)$
which tends to $\frac{4}{\sqrt{3}}$ from below as $\delta \rightarrow 0$ or
$k \rightarrow \infty$.
The distance between $p$ and $q$, on the other hand is $\frac{2}{\sqrt{3}}$.
Therefore, for any $\epsilon > 0$, it is possible to choose $k$, $\delta$,
and a set $P$ of points such that the {\Large\varhexagon}-Delaunay
triangulation on $P$ has stretch factor at least $2 - \eps$.
\end{proof}
}

\section{Main result}
\label{sec:main}
In this section we state our main result and provide an overview of our proof.
\iftoggle{abstract}
{We start with a technical lemma that is used to prove the two
key lemmas needed for the main result.}
{We start by introducing a technical lemma that will be used to prove the two
key lemmas needed for the main result.}

\subsection{Technical lemma}
\label{sub:keylemma}

Let $T$ be the {\Large\varhexagon}-Delaunay triangulation on a set of
points $P$ in the plane. 
\begin{definition}
Let $T_1, T_2, \dots, T_n$ be a sequence of triangles of $T$ that a line $st$
of finite slope intersects. This sequence of triangles is said to be
{\em linear} w.r.t. line $st$ if for every $i=1,\dots,n-1$:
\begin{itemize}
\item triangles $T_i$ and $T_{i+1}$ share an edge, and 
\item line $st$ intersects the interior of that shared edge (not an endpoint).
\end{itemize}
\end{definition}

Our goal is to prove an upper bound on the length of the shortest path from the 
``leftmost'' point of $T_1$ to the ``rightmost'' point of $T_n$, when certain conditions hold. We 
introduce some notation and definitions, illustrated in 
Fig.~\ref{fig:regular}, to make this more precise.

We consider the $n-1$ shared triangle edges intersected by line $st$ from left
to right (where left and right are defined with respect to $x$-coordinates)
and label the endpoints of the $i$-th edge $u_i$ and $l_i$, with
$u_i$ being above line $st$ and $l_i$ below. We note that points typically get
multiple labels and identify a point with its label(s). 
If line $st$ goes through the vertex of $T_1$ other than $u_1$ and $l_1$,
we assign that vertex both labels $u_0$ and $l_0$ (as shown in
Fig.~\ref{fig:definitions}); otherwise, we assign labels
$u_0$ and $l_0$ to the endpoints of the edge of $T_1$ intersected by line $st$
other than $(u_1,l_1)$, with $u_0$ being above line $st$ (as shown in
Fig.~\ref{fig:regular}). Similarly, if line 
$st$ goes through the vertex of $T_n$ other than $u_{n-1}$ and $l_{n-1}$, 
we assign it both labels $u_n$ and $l_n$ (as shown in
Fig.~\ref{fig:definitions}); otherwise, we assign labels 
$u_n$ and $l_n$ to the endpoints of the other edge of $T_n$ intersected by
line $st$, with $u_n$ being above line $st$ (as shown in
Fig.~\ref{fig:regular}). Note that for $1 \leq i \leq n$:
\begin{itemize}
\item either $T_i = \triangle(u_i, l_i, l_{i-1})$, in which case we call
$l_{i-1}$ and $l_i$ the {\em left} and {\em right} vertices of $T_i$ 
\item or $T_i = \triangle(u_{i-1}, u_i, l_i)$, in which case we call $u_{i-1}$
and $u_i$ the {\em left} and {\em right} vertices of $T_i$.
\end{itemize}
Note that only one of the above holds, except for $T_1$ if $u_0=l_0$
(in which case both hold) or for $T_n$ if $u_n=l_n$ (in which case again both
hold). For every $i = 1, \dots, n$,
when $u_i = u_{i-1}$ or $l_i = l_{i-1}$ we call the corresponding vertex of $T_i$
the {\em base} vertex of $T_i$. Note that $T_1$ has no base vertex if $u_0=l_0$
and $T_n$ has no base vertex if $u_n=l_n$ (as is the case in 
Fig.~\ref{fig:definitions}). Let $U$ and $L$ be the sets of all point labels
$u_i$ and $l_i$, respectively, and let $T_{1n}$ be the union of
$T_1,T_2,\dotsc,T_{n}$ which we will refer to as a linear sequence of triangles
as well.

\iftoggle{abstract}
{\begin{figure}[!b]}
{\begin{figure}[!b]}
\center{\regular}
\caption{The dotted line ($st$) intersects the linear sequence of triangles
$T_1, T_2, \dots, T_5$. The vertices of each triangle $T_i$
($u_{i-1}$, $u_i$, $l_{i-1}$, $l_i$, two of which are equal) lie on the boundary
of the hexagon $H_i$. Note that $l_2$ is the left, $l_3$ is the right,
and $u_2=u_3$ is the base vertex of $T_3$, for example. The linear sequence is
regular since $T_1$ has a left induction vertex, $T_5$ has a right induction
vertex, and, for $i = 1, \dots, 4$, no $u_i$ is a $s$ vertex of $T_i$ or
$T_{i+1}$, no $l_i$ is a $n$ vertex of $T_i$ or $T_{i+1}$, and no $(u_i,l_i)$
is gentle.
}
\label{fig:regular}
\end{figure}

Let $H_i$, for $1 \leq i \leq n$, be the (empty) hexagon passing through the
vertices of $T_i$; note that $H_i$ is unique due to the general position
assumption. A vertex of $T_i$ is said to be a $w$, $e$, $n$,
or $s$ vertex of $T_i$ if it lies on the $w$ side, $e$ side, one of
the $n$ sides, or one of the $s$ sides, respectively, of $H_i$  
(see Fig.~\ref{fig:regular}).
A left vertex of $T_i$ that is a $w$ vertex of
$T_i$ is referred to as a {\em left} {\em induction vertex}
of $T_i$; similarly, a right vertex that is a $e$ vertex is referred to as a {\em right} {\em induction vertex}
of $T_i$. 

Note that a base vertex cannot be an induction vertex.



\begin{definition}
We call an edge $(u_i,l_j)$ {\em gentle} if its slope is between
$-\frac{1}{\sqrt{3}}$ and $\frac{1}{\sqrt{3}}$.
\end{definition}
In Fig.~\ref{fig:regular} no edge $(u_i,l_j)$ is gentle while in 
Fig.~\ref{fig:definitions} $(u_0,l_1)$ and $(u_8,l_8)$ are gentle.

\begin{definition}
The linear sequence of triangles $T_{1n}$ is {\em regular} if $T_1$ has a left
induction vertex, $T_n$ has a right induction vertex, and if, for
every $i = 1, \dots, n-1$:
\begin{itemize}
\item $u_i$ is {\em not} a $s$ vertex of $T_i$ and $T_{i+1}$,
\item $l_i$ is {\em not} a $n$ vertex of $T_i$ and $T_{i+1}$, and 
\item $(u_i,l_i)$ is not gentle.
\end{itemize} 
\end{definition}
The linear sequence in Fig.~\ref{fig:regular} is regular while the one in
Fig.~\ref{fig:definitions} is not (because $u_8$ lies on the $s_w$ side of
$H_9$---the red hexagon passing through the vertices of
$T_9 = \triangle(u_8, u_9, l_9)$---and also because edge $(u_8,l_8)$ is gentle). 

\iftoggle{abstract}
{The proof of the following technical lemma is discussed in Section~\ref{sec:proofA}.}
{The following technical lemma is proven in Section~\ref{sec:proofA}.}
\begin{lemma}[The Technical Lemma]
\label{le:mainlemmaA}
If $T_{1n}$ is a regular linear sequence of triangles then there is a path in
$T_{1n}$ from
the left induction vertex $p$ of $T_1$ to the right induction vertex $q$ of
$T_n$ of length at most  $\frac{4}{\sqrt{3}} d_x(p,q)$.
\end{lemma}

Actually, what we show in Section~\ref{sec:proofA} implies something stronger: 
If $T_{1n}$ is regular then the lengths of the {\em upper} path 
$p,u_0,\dots,u_n,q$ and of the {\em lower} path $p,l_0,\dots,l_n,q$ 
add up to at most $\frac{8}{\sqrt{3}} d_x(p,q)$. 
It is useful to informally describe now the techniques we use to do this. For
that purpose we introduce, for a point $o$ on a side of $H_i$, functions
$p_N(o,i)$ and $p_S(o,i)$ as the {\em signed} shortest distances around the
perimeter of $H_i$ from $o$ to the $N$ vertex and $S$ vertex, respectively;
the sign is positive if $o$ lies on sides
$n_w$, $w$, or $s_w$ of $H_i$ and negative otherwise 
(see Fig.~\ref{fig:discretedNdS}-(a) and Fig.~\ref{fig:discretedNdS}-(b)).

\iftoggle{abstract}
{\begin{figure}[!b]}
{\begin{figure}[!b]}
\discretedNdS
\begin{center}
\vspace{-8pt}\hspace{-1.3cm} (a) \hspace{4.35cm} (b) \hspace{4.55cm} (c)
\end{center}
\caption{(a) The values of $p_N(o,i)$ are illustrated, for various points $o$
lying on the boundary of $H_i$, as signed hexagon arc lengths. (b) The 
values of $p_S(o,i)$ are illustrated similarly. (c) The length of edge
$(u_{i-1},u_i)$ is bounded by $p_N(u_{i-1}, i) - p_N(u_i, i)$.}
\label{fig:discretedNdS}
\end{figure}


Note that the length of each edge 
$(u_{i-1},u_i)$ (assuming $u_{i-1} \not= u_i$) can be
bounded by the distance
from $u_{i-1}$ to $u_i$ when traveling clockwise along the sides of $H_i$.
This distance is exactly $p_N(u_{i-1}, i) - p_N(u_i, i)$ as
illustrated in Fig.~\ref{fig:discretedNdS}-(c).
This motivates the following {\em discrete} function, defined for
$i=0,1,\dots,n$ and, for convenience's sake, 1) assuming that $p=u_0$ and
$q=u_n$ and 2) using an additional hexagon $H_{n+1}$ of radius $0$ centered
at point $q$:
\[\bar{U}(i) = \sum_{j=1}^{i} (p_N(u_{j-1}, j) - p_N(u_{j}, j)) + p_N(u_{i},i+1).\]

Function $\bar{U}(i)$ can be used to bound the length of upper path fragments;
in particular, $\bar{U}(n)$ bounds the length of the upper path from 
$p$ to $q$. A function $\bar{L}(i)$ bounding the length of the lower
path can be defined similarly. In Section~\ref{sec:proofA}, we will compute
an upper bound for $\bar{U}+\bar{L}$ by
1) switching the analysis from a discrete one to a continuous one, with
functions $p_N$ and $p_S$ defined not in terms of index $i$ but
in terms of coordinate $x$ for every $x$ between ${\tt x}(p)$ and ${\tt x}(q)$
and 2) analyzing the growth rates, with respect to $x$, of the continuous
functions $p_N$, $p_S$, and $\bar{U}+\bar{L}$. We
will show that (the continuous versions of) $p_N$ and $p_S$ are piecewise linear functions with growth rates
$\frac{2}{\sqrt{3}}$, $\frac{4}{\sqrt{3}}$, or $\frac{6}{\sqrt{3}}$, and that
$\bar{U}+\bar{L}$ is also piecewise linear with growth rate equal to the growth
rate of $p_N+p_S$ which can be $\frac{4}{\sqrt{3}}$, $\frac{6}{\sqrt{3}}$, or
$\frac{8}{\sqrt{3}}$. 
Lemma~\ref{le:mainlemmaA} will follow from the last (largest) growth rate.


With the technical lemma in hand, we can now state the first of the two
key lemmas that we need to prove our main result.

\subsection{The Amortization Lemma}

The first of our two key lemmas is a strengthening of the (Technical) Lemma~\ref{le:mainlemmaA}
under two restrictions. The first restriction is that $T_{1n}$ is 
defined with respect to a line $st$ whose slope $m_{st}$ is restricted to
$0 < m_{st} < \frac{1}{\sqrt{3}}$. With that restriction we get the following
properties:


\begin{lemma}
\label{lem:props}
Let $T_{1n}$ be a linear sequence with respect to line $st$ with slope $m_{st}$
such that $0 < m_{st} < \frac{1}{\sqrt{3}}$. For every $i$ s.t. 
$1 \leq i \leq n$:
\begin{itemize}
\item If $u_{i-1}$ lies on side $s_w$ of $H_{i}$ or $l_{i}$ lies on 
side $n_e$ of $H_{i}$ then $(u_{i-1},l_{i})$ is gentle.
\item If $l_{i-1}$ lies on side $n_w$ of $H_{i}$ or $u_{i}$ lies on the $s_e$
side of $H_{i}$ then $(l_{i-1},u_{i})$ is gentle.
\item None of the following can occur: $u_{i-1}$ lies on side $s_e$ of $H_{i}$,
$l_{i}$ lies on side $n_w$ of $H_{i}$, $l_{i-1}$ lies on side $n_e$ of $H_{i}$,
and $u_{i}$ lies on the $s_w$ side of $H_{i}$. 
\end{itemize}
\end{lemma}
\noindent Note, for example, that $u_8$ lies on side $s_w$ of hexagon
$H_9$ in Fig.~\ref{fig:definitions} and that edge $(u_8,l_9)$ is gentle. 
A gentle edge that satisfies one of the four cases in the first two bullet points of
Lemma~\ref{lem:props} will
be called {\em irregular}.

\begin{proof}
If $u_{i-1}$ lies on side $s_w$ of some hexagon $H_{i}$ then, since
$0 < m_{st} < \frac{1}{\sqrt{3}}$ and by general position assumptions, 
either $u_{i-1}=u_{i}$ and $l_{i-1}$ and $l_{i}$
must lie on sides $s_e$ and $e$ of $H_{i}$, respectively, or $l_{i-1}=l_{i}$ must
lie on side $s_e$ or $e$ of $H_{i}$. Either way, the slope of
the line going through $u_{i-1}$ and $l_{i}$ must be between 
$-\frac{1}{\sqrt{3}}$ and  $\frac{1}{\sqrt{3}}$. Similar arguments can be used
to handle the remaining three cases in the first two bullet points.

Let the left and right intersection points of line $st$ with hexagon $H_{i}$
be  $h_{i-1}$ and $h_{i}$. Note that when traveling clockwise along the sides of
$H_{i}$ the points will be visited in this order: 
$h_{i-1},u_{i-1},u_{i},h_{i}, l_{i}, l_{i-1}$.
If $u_{i-1}$ lies on side $s_e$ of $H_{i}$ then $i>1$ and, because
$0 < m_{st} < \frac{1}{\sqrt{3}}$, either $u_{i}$ (if $u_{i-1} \not= u_{i}$)
or $l_{i}$ (if $l_{i-1} \not= l_{i}$) would have to lie on side $s_e$
of $H_{i}$ as well, which violates our general position assumption for the
set of  points P. The remaining three cases are handled similarly.
\end{proof}


By the above lemma, under the restriction $0 < m_{st} < \frac{1}{\sqrt{3}}$,
if $T_{1n}$ has no gentle edge then it is regular and (Technical) 
Lemma~\ref{le:mainlemmaA} applies. A narrower but much stronger
version of (Technical) Lemma~\ref{le:mainlemmaA} applies as well if another
restriction is made. To state the second restriction we need
some additional terminology. 

Let $l_i \in L$ and $u_j \in U$. If $i \leq j$ and ${\tt x}(l_i) < {\tt x}(u_j)$
then we say that $l_i$ {\em occurs before} $u_j$, and if $j \leq i$ and
${\tt x}(u_j) < {\tt x}(l_i)$ then we say that $u_j$ {\em occurs before}
$l_i$.
\begin{definition}
Given points $l_i \in L$ and $u_j \in U$ such that one occurs before the other,
a path between them is \emph{gentle} if the length of the path is not greater
than $\sqrt{3}d_x(u_j,l_i)-({\tt y}(u_j)-{\tt y}(l_i))$.
\end{definition}
See Fig.~\ref{fig:definitions} for an illustration of a gentle path. Note that
a gentle edge is a gentle path (e.g., $(u_0,l_1)$ and 
$(u_8,l_8=l_9)$ in Fig.~\ref{fig:definitions}).

The following is the key to our proof of the main result of this paper:
\begin{lemma}[The Amortization Lemma]
\label{le:mainlemmaB}
Let $T_{1n}$ be a regular linear sequence with respect to line $st$ with slope
$m_{st}$. If $0 < m_{st} < \frac{1}{\sqrt{3}}$ and if $T_{1n}$ contains no
gentle path then there is a path in $T_{1n}$ from 
the left induction vertex $p$ of $T_1$ to the right induction vertex $q$ of
$T_n$ of length at most  $(\cT) d_x(p,q)$.
\end{lemma}
\iftoggle{abstract}{We will discuss the proof of the Amortization Lemma in 
Section~\ref{sec:proofB}; the proof builds 
on the analysis done in Section~\ref{sec:proofA} to prove (Technical) 
Lemma~\ref{le:mainlemmaA}.}{We will prove the Amortization Lemma in Section~\ref{sec:proofB} by building 
on the analysis done in Section~\ref{sec:proofA} to prove (Technical) 
Lemma~\ref{le:mainlemmaA}.} 
Instead of using function $\bar{U}$, however, we consider the discrete 
function
\begin{equation*}
U(i) = d_{T_{1i}}(p,u_i) + p_N(u_{i},i+1)
\end{equation*}
defined for $i=0,1,\dots,n$ and, for convenience's sake, 1) assuming that 
$p=u_0$ and $q=u_n$ and 2) using additional hexagon $H_{n+1}$ of radius $0$
centered at point $q$. An equivalent discrete function $L(i)$ using points 
$l_i$ instead of $u_i$ can be defined. Note that $U(n) + L(n)$ is exactly twice
the distance in
$T_{1n}$ from $p$ to $q$. To bound $U + L$, we will switch the analysis to
a continuous one just as we did for ${\bar U} + {\bar L}$. We will see that,
except for a finite number of discontinuities, the continuous version of $U+L$
has the same growth rate as ${\bar U} + {\bar L}$, which is the growth
rate of $p_N+p_S$. 
We will consider the intervals when the growth rate of (the continuous version of)
$U+L$ is higher than $2(\cT)$ (i.e., when its growth rate is $\frac{8}{\sqrt{3}}$) 
and we will amortize the extra $2 - \frac{2}{\sqrt{3}}$ growth over 
intervals when its growth rate is smaller than $2(\cT)$ (i.e., when its
growth rate is $\frac{4}{\sqrt{3}}$ or $\frac{6}{\sqrt{3}}$).
The amortization can usually be done because the intervals when the growth rate
of $U+L$ is large must be relatively short compared to intervals when
its growth is smaller, otherwise a gentle path can be shown to exist. To get our
tight bound however, we will need to do more and show that at certain points
(which are points of discontinuity)  we need to use ``cross-edges'' ($l_i,u_i)$.
This is because when the amortization is not
possible there is a long enough interval, say from hexagon
$H_i$ to hexagon $H_j$, when the growth rate of $U+L$ is mostly
$\frac{8}{\sqrt{3}}$.
It turns out that in that case one of $U$ or $L$ has growth rate
bounded by $\frac{2}{\sqrt{3}}$ (say, $U$) and the other ($L$) by 
$\frac{6}{\sqrt{3}}$. This means that path $l_i, l_{i+1}, \dots, l_j$ has
relatively large length with respect to $\Delta(x)$ 
and that $u_i, u_{i+1}, \dots, u_{j}$ is a relatively
short path that can be used to replace the long subpath $l_i, l_{i+1}, \dots, l_j$
with the shorter subpath $l_i, u_i, u_{i+1}, \dots, u_{j}, l_j$ in a path from $p$
to $q$. The $\cT$
stretch factor bound is the result of a min-max optimization between the two
subpaths from $l_i$ to $l_j$, and it is tight as we show in
Section~\ref{sec:conclusion}.

\iftoggle{abstract}
{\begin{figure}}
{\begin{figure}}
\center{\definitions}
\caption{
A gentle path from $u_2$
to $l_{11}$ is one whose length is at most 
$\sqrt{3}d_x(u_2,l_{11})-({\tt y}(u_2)-{\tt y}(l_{11}))$, i.e. the length of the
red dashed piecewise linear curve from $u_2$ to $l_{11}$ (consisting of two 
vertical segments and a third with slope $-\frac{1}{\sqrt{3}}$). The path 
$u_2=u_3,u_4,u_5=u_6,u_7=u_8,l_8=l_9=l_{10},l_{11}$, easily seen to be
bounded--in length--by the red dotted piecewise linear curve, is gentle.
This path can be extended with edge $(l_{11},t)$ to a canonical gentle path
from $u_2$ to $t$; the proof of (Main) Lemma~\ref{le:divide}, in this particular
case, combines the bound on the
length of this path together with the bound on the length of a path from $s$
to $u_2$ obtained via induction.}
\label{fig:definitions}
\end{figure}

Next we turn to the case when the sequence of triangles $T_{1n}$ contains
a gentle path.

\subsection{The Gentle Path Lemma}

Just as in the previous subsection, we consider a linear sequence of triangles
$T_{1n}$ defined with respect to a line $st$ with slope $m_{st}$ satisfying
$0 < m_{st} < \frac{1}{\sqrt{3}}$. We now consider the case when $T_{1n}$
contains a gentle path and state the other of our two key lemmas. We start 
with two definitions:
\begin{definition}
We say that linear sequence $T_{1n}$ is {\em standard}
if $T_1$ has a left induction vertex or $u_0=l_0$, $T_n$ has a right induction
vertex or $u_n=l_n$, and neither the base vertex of $T_1$ (if any) nor the base
vertex of $T_n$ (if any) is the endpoint of a gentle path in $T_{1n}$. 
\end{definition}
Note that if $u_0=l_0$ and $u_n=l_n$ both hold (i.e., line $st$ goes through
those points) then $T_{1n}$ is trivially standard because $T_1$ and $T_n$
cannot  have base vertices. 

\begin{definition}
Let $T_{1n}$ be a standard linear sequence.
A gentle path in $T_{1n}$ from $p$ to $q$, where $p$ occurs before $q$, is
{\em canonical in $T_{1n}$} (or simply {\em canonical} if $T_{1n}$ is clear
from the context) if $p$ is a right induction vertex of $T_{i}$ for some
$i \geq 1$  or $p$ is the left vertex of $T_1$ and if $q$ is a left induction
vertex of $T_{j}$ for some $j \leq n$ or $q$ is the right vertex of $T_n$.
\end{definition}
For example, the gentle path 
$u_2=u_3,u_4,u_5=u_6,u_7=u_8,l_8=l_9=l_{10},l_{11},l_{12}$
in Fig.~\ref{fig:definitions} is canonical.

The second key lemma, which we will use alongside (Amortization) 
Lemma~\ref{le:mainlemmaB} to prove our main result, is stated
\iftoggle{abstract}{next; its proof is discussed in Section~\ref{sec:gentle}.}{next and proven in Section~\ref{sec:gentle}.}
\begin{lemma}[The Gentle Path Lemma] 
\label{le:boundedGentelSections}
Let $T_{1n}$ be a linear sequence of triangles with respect to a line $st$ with
slope $m_{st}$ such that $0 < m_{st} < \frac{1}{\sqrt{3}}$. If $T_{1n}$ is
standard and contains a gentle path then the path can be extended
to a canonical gentle path in $T_{1n}$.
\end{lemma}
The main idea behind the proof of this lemma is that a gentle path between
$u_r \in U$ and $l_s \in L$ (where, say, $r \leq s$ and 
${\tt x}(u_r) < {\tt x}(u_s)$) in $T_{1n}$ can 
be extended using edge $(u_{r-1},u_r)$, unless $r=0$ or $u_r$ is a right
induction vertex of $T_{r}$, or using edge $(l_{s},l_{s+1})$, unless $s=n$ or
$l_s$ is a left induction vertex of $T_{s+1}$. In other words, a gentle path
can be extended unless it is canonical.

We are now ready to state our main result and provide a proof that uses
the two key lemmas.

\subsection{The main result and the Main Lemma}

\begin{theorem}
\label{th:main}
The stretch factor of a {\Large\varhexagon}-Delaunay triangulation is at most $2$.
\end{theorem}

To prove Theorem~\ref{th:main} we need to show that between any two points
$s$ and
$t$ of a set of points $P$ there is, in the {\Large\varhexagon}-Delaunay
triangulation $T$ on $P$, a path from $s$ to $t$ of length
at most $2d_2(s,t)$. Let $m_{st}$ be the slope of the line $st$ 
passing through $s$ and $t$.
Thanks to the hexagon's rotational and reflective symmetries as well as our general 
position assumptions, we can rotate
the plane around $s$ and possibly reflect the plane with respect to the $x$-axis
to ensure that 
$0 < m_{st} < \frac{1}{\sqrt{3}}$. Given this assumption, our main theorem
will follow from:
\begin{lemma}[The Main Lemma] \label{le:divide}
For every pair of points $s,t \in P$ with 
$0 < m_{st} < \frac{1}{\sqrt{3}}$:
\begin{equation}
\label{eq:main}
d_T(s,t) \leq \max\Bigl\{\cT, \sqrt{3}+m_{st}\Bigr\}d_x(s,t).
\end{equation}
\end{lemma}
Before we prove this lemma, we show that it implies the main theorem.
\begin{proof}[Proof of Theorem~\ref{th:main}]
W.l.o.g., we assume that $s$ has coordinates $(0,0)$, $t$ lies in the
positive quadrant, $m_{st} < \frac{1}{\sqrt{3}}$, and $d_2(s,t)=1$.
With these assumptions it follows that 
$\frac{\sqrt{3}}{2} < {\tt x}(t) = d_x(s,t) < 1$
and we need to show that $d_T(s,t) \leq 2$.

By Lemma \ref{le:divide}, either 
$d_T(s,t)\leq (\cT) d_x(s,t) \leq (\cT) < 2$ or
 \[d_T(s,t) \leq (\sqrt{3}+m_{st}) d_x(s,t) = \sqrt{3}d_x(s,t) + d_y(s,t) = \sqrt{3}d_x(s,t)+ \sqrt{1-d_x(s,t)^2}\]
 which attains its maximum, over the interval $[\frac{\sqrt{3}}{2},1]$, at 
$d_x(s,t)=\frac{\sqrt{3}}{2}$ giving $d_T(s,t) \leq 2$.
\end{proof}

We now turn to the proof of (Main) Lemma~\ref{le:divide}. We start by noting that if
there is a point $p$ of $P$ on the segment $[st]$
then~(\ref{eq:main}) would follow if~(\ref{eq:main}) holds for the pairs
of points $s,p$ and $p,t$; we can therefore assume that no point of $P$ other
than $s$ and $t$ lies on the segment $[st]$. We can also assume, as argued in
Section~\ref{sec:prelim}, that segment
$[st]$ does not intersect the outer face of the triangulation $T$.
We assume w.l.o.g. that
$s$ has coordinates $(0,0)$ and thus $t$ lies in the positive quadrant.

Let $T_1, T_2, T_3, \dots, T_n$ be the sequence of
triangles of the 
triangulation $T$ that line segment $[st]$ intersects
when moving from $s$ to $t$ (refer to Fig.~\ref{fig:definitions}). (Recall
that we assume that segment $[st]$ does not intersect the outer face of
$T$.)
Clearly, $T_{1n}$ is a linear
sequence of triangles and we assign labels $u_i$ and $l_i$ to the points
and define sets $U$ and $L$ as described in Subsection~\ref{sub:keylemma}.
We note that all arguments in the rest of this paper use only points
and edges in $T_{1n}$.

Note that, since $st$ must intersect the
interior of $H_1$, $s$ can only lie on the $n_w$, $w$, or $s_w$ sides of $H_1$;
by Lemma~\ref{lem:props}, if $s$ lies on the $n_w$ side of $H_1$ then
$(s,u_1) = (l_0,u_1)$ is gentle, and
if $s$ lies on the $s_w$ side of $H_1$ then $(s,l_1) = (u_0,l_1)$ is gentle. 
Similarly, $t$ can only lie on the $n_e$, $e$, or $s_e$ sides of $H_n$; if
$t$ lies on the $s_w$ side of $H_n$ then $(t,l_{n-1}) = (u_n,l_{n-1})$ is 
gentle, and if $t$ lies on the $n_w$ side of $H_n$ then 
$(t,u_{n-1}) = (l_n,u_{n-1})$ is gentle. Note that this means that if $T_{1n}$ has
no gentle edge then it is regular.

We now informally describe the approach we use to prove (Main) 
Lemma~\ref{le:divide}. We first note that (Amortization) Lemma~\ref{le:mainlemmaB}
and (Gentle Path) Lemma~\ref{le:boundedGentelSections}
rely on (Technical) Lemma~\ref{le:mainlemmaA}. We will prove (Main) 
Lemma~\ref{le:divide} that
bounds the length of the shortest path in $T_{1n}$ from $s$ to $t$ as follows. 
If $T_{1n}$ does not contain a gentle path then it is regular and the proof follows
from (Amortization) Lemma~\ref{le:mainlemmaB}. If $T_{1n}$ contains a gentle path
then by (Gentle Path) Lemma~\ref{le:boundedGentelSections} it must contain a
canonical gentle path $\mathcal{G}$ from, in general, a right induction vertex of $T_i$ to
a left induction vertex of $T_j$, where $0 \leq i < j \leq n$. We can assume, 
using (Gentle Path) Lemma~\ref{le:boundedGentelSections}, that $\mathcal{G}$ is maximal in
the sense that it is not a subpath of any other gentle path in $T_{1n}$.  
The maximality of $\mathcal{G}$ will guarantee that neither $T_{1i}$ nor $T_{jn}$ contains
a gentle path whose endpoint is the base vertex of $T_i$ or the base
vertex of $T_j$, respectively. Therefore $T_{1i}$ and $T_{jn}$ are standard and
we then proceed by induction to prove a ``more general'' version of (Main)
Lemma~\ref{le:divide} for $T_{1i}$ and $T_{jn}$. The obtained bounds on the
lengths of shortest paths from $s$ to the right induction vertex of $T_i$ and 
from the left induction vertex of $T_j$ to $t$ are combined
with the bound on the length of gentle path $\mathcal{G}$ to complete the proof of (Main)
Lemma~\ref{le:divide}. Our reliance on induction means that we need to restate
the Main Lemma so it is amenable to an inductive proof:
\begin{lemma}[The Generalized Main Lemma]
Let $s,t \in P$ such that $0 < m_{st} < \frac{1}{\sqrt{3}}$ and let $T_{1n}$
be the linear sequence of triangles that segment $[st]$ intersects. If
$T_{ij}$, for some $i,j$
such that $1 \leq i \leq j \leq n$, is standard, $p$ is the left vertex
of $T_i$, and $q$ is the right vertex of $T_j$ then
\[d_{T_{ij}}(p,q) \leq  \max\{\cT,\sqrt{3}+m_{st}\}d_x(p,q).\]
\end{lemma}
Note that (Main) Lemma \ref{le:divide} is a special case of this
statement when $i=1$ and $j=n$ since $T_{1n}$ is (trivially) standard, 
$s$ is the left vertex of $T_1$, and $t$
is the right vertex of $T_n$.

\begin{proof}

We proceed by induction on $j-i$. If $T_{ij}$ is standard and there
is no gentle path in $T_{ij}$ (the base case) then, by Lemma~\ref{lem:props},
the linear sequence of triangles in $T_{ij}$ is regular and thus, by (Amortization)
Lemma~\ref{le:mainlemmaB}, we have $d_T(p,q) \leq (\cT) d_x(p,q)$.

If $T_{ij}$ is standard and there is a gentle path in $T_{ij}$,
then, by Lemma~\ref{le:boundedGentelSections}, there
exist points $u_{i'}$ and $l_{j'}$ in $T_{ij}$ such that there is a canonical
gentle path between $u_{i'}$ and $l_{j'}$ in $T_{ij}$. We also
assume that the canonical path between $u_{i'}$ and $l_{j'}$ is maximal in the
sense that it is not a proper subpath of a gentle path in $T_{ij}$. W.l.o.g.,
we assume that $u_{i'}$ occurs before $l_{j'}$, and so 
$i-1 \leq i' \leq j' \leq j$, ${\tt x}(u_{i'}) < {\tt x}(l_{j'})$, and
$d_T(u_{i'},l_{j'}) \leq \sqrt{3}d_x(u_{i'},l_{j'})-({\tt y}(u_{i'})-{\tt y}(l_{j'}))$.
Since $u_{i'}$ is either $s$
or above $st$ and $l_{j'}$ is either $t$ or below $st$, it follows that
$-({\tt y}(u_{i'})-{\tt y}(l_{j'}))\leq m_{st}d_x(u_{i'},l_{j'})$. Therefore,
$d_T(u_{i'},l_{j'}) \leq (\sqrt{3}+m_{st})d_x(u_{i'},l_{j'})$.

Since the gentle path from $u_{i'}$ to $l_{j'}$ is canonical, either $u_{i'}$ 
is a right induction vertex of $T_{i'}$ and $i' \geq i$ or $u_{i'}=u_{i-1}$.
In the first case, because $u_{i'}$ is on side $e$ of $H_{i'}$ the base
vertex $l_{i'-1}=l_{i'}$ of $T_{i'}$ must satisfy 
${\tt x}(l_{i'}) < {\tt x}(u_{i'})$. Suppose that
$l_{i'}$ is the endpoint of a gentle path in $T_{ii'}$ from, say, point
$u_{i''}$ then we would have
\begin{align*}
d_{T_{ij}} (u_{i''},l_{j'}) & \leq d_{T_{ij}}(u_{i''},l_{i'}) + d_2(l_{i'},u_{i'}) + d_{T_{ij}}(u_{i'},l_{j'}) \\
             & \leq \sqrt{3}d_x(u_{i''},l_{i'})-({\tt y}(u_{i''})-{\tt y}(l_{i'})) + \sqrt{3}d_x(l_{i'},u_{i'}) - ({\tt y}(l_{i'})-{\tt y}(u_{i'})) \\
             & \quad + \sqrt{3}d_x(u_{i'},l_{j'})) - ({\tt y}(u_{i'})-{\tt y}(l_{j'})) \\
             & \leq \sqrt{3}d_x(u_{i''},l_{j'})-({\tt y}(u_{i''})-{\tt y}(l_{j'}).
\end{align*}
This contradicts the maximality of the canonical gentle path from $u_{i'}$ to
$l_{j'}$. This means that $l_{i'}$ is not the endpoint of a gentle path in
$T_{ii'}$. Therefore $T_{ii'}$ is standard, the inductive hypothesis
applies, and $d_T(p,u_{i'}) \leq \max\{\cT, \sqrt{3}+m_{st}\}d_x(p,u_{i'})$. 
In the second case, because $T_{ij}$ is standard, $u_{i'}$
cannot be the base vertex of $T_i$ and so $u_{i'}=p$. The 
inequality $d_T(p,u_{i'}) \leq \max\{\cT, \sqrt{3}+m_{st}\}d_x(p,u_{i'})$ then
holds trivially. 

Similarly, we can show that
$d_T(l_{j'},q) \leq \max\{\cT, \sqrt{3}+m_{st}\}d_x(l_{j'},q)$. Thus: 
\begin{align*} d_T(p,q) & \leq d_T(p,u_{i'})+d_T(u_{i'},l_{j'})+d_T(l_{j'},q) \\
                        & \leq \max\{\cT, \sqrt{3}+m_{st}\}(d_x(p,u_{i'}) + d_x(l_{j'},q)) + (\sqrt{3}+m_{st})d_x(u_{i'},l_{j'}) \\
                        & \leq \max\{\cT , \sqrt{3}+m_{st}\}d_x(p,q)\end{align*}\end{proof}

\section{Proof of (Gentle Path) Lemma~\ref{le:boundedGentelSections}}
\label{sec:gentle}
\iftoggle{abstract}
{The main idea behind the proof of the Gentle Path lemma is that a gentle path
between $u_r \in U$ and $l_s \in L$ (where, say, $r \leq s$ and 
${\tt x}(u_r) < {\tt x}(u_s)$) in $T_{1n}$ can 
be extended using edge $(u_{r-1},u_r)$, unless $r=0$ or $u_r$ is a right
induction vertex of $T_{r}$, or using edge $(l_{s},l_{s+1})$, unless $s=n$ or
$l_s$ is a left induction vertex of $T_{s+1}$. 
In other words, a gentle path from $r$ to $s$ is either canonical
or can be extended to a canonical path from $u_{r'}$ to $l_{s'}$ as illustrated
in Fig.~\ref{fig:caseA}.

\begin{figure}[h]
\center{\caseA}
\caption{Illustration of the proof of 
Lemma~\ref{le:boundedGentelSections} in the case when the gentle path from $u_r$ to $l_s$ is just a gentle
edge. For every $i$ such that $r' < i \leq r$ and $u_i$ is
the right vertex of $H_i$, hexagon $H_i$ and the edge $(u_{i-1},u_i)$ are shown
in red. Each edge $(u_{i-1},u_i)$ has slope greater than $-\frac{1}{\sqrt{3}}$
and therefore has length bounded by
$\sqrt{3}d_x(u_{i-1},u_i) - ({\tt y}(u_{i-1})-{\tt y}(u_i))$, a value equal
to the total length of the two intersecting, red, dashed segments going north
from $u_{i-1}$ and north-west from $u_i$. The total length of the two dashed
blue segments is an upper bound on the length of the edge $(u_r,l_s)$ and
the total length of the dotted red line segments represent the upper bound
$\sqrt{3}d_x(p,q) - ({\tt y}(p)-{\tt y}(q))$ on the length of the path 
$p=u_{r'}, \dots, u_r, l_{s}, l_{s+1}, \dots, l_{s'}=q$.}
\label{fig:caseA}
\end{figure}

}
{

\begingroup
\def\thetheorem{\ref{le:boundedGentelSections}}
\begin{lemma}[The Gentle Path Lemma] 
Let $T_{1n}$ be a linear sequence of triangles with respect to a line $st$ with
slope $m_{st}$ such that $0 < m_{st} < \frac{1}{\sqrt{3}}$. If $T_{1n}$ is
standard and contains a gentle path then the path can be extended
to a canonical gentle path in $T_{1n}$.
\end{lemma}
\addtocounter{theorem}{-1}
\endgroup

The main idea behind the proof of this lemma is that a gentle path between
$u_r \in U$ and $l_s \in L$ (where, say, $r \leq s$ and 
${\tt x}(u_r) < {\tt x}(u_s)$) in $T_{1n}$ can 
be extended using edge $(u_{r-1},u_r)$, unless $r=0$ or $u_r$ is a right
induction vertex of $T_{r}$, or using edge $(l_{s},l_{s+1})$, unless $s=n$ or
$l_s$ is a left induction vertex of $T_{s+1}$. In other words, a gentle path
can be extended unless it is canonical.

\begin{proof}
We break the proof into several cases.

{\bf Case A.}
We first prove the claim in the case when $T_{1n}$ contains no irregular gentle
edge. In that case, because $T_{1n}$ is standard and by Lemma~\ref{lem:props},
the left vertex of $T_1$, whether it is $l_{0}$ or $u_{0}$ (or $l_0=u_0$),
lies on side $w$ of $H_1$ and the right vertex of $T_n$, whether it is $l_n$
or $u_n$ (or $u_n=l_n$), lies on side $e$ of $H_n$.

Let $T_{1n}$ contain a gentle path between $u_r \in U$ and
$l_s \in L$. W.l.o.g., we assume that $u_r$ occurs before $l_s$ 
(i.e., $r \leq s$ and ${\tt x}(u_r) < {\tt x}(l_r)$), 
as the case when $l_s$ occurs before $u_r$ can be argued
using a symmetric argument. 
Consider the sequence of points $u_{0}, u_1, \dots, u_r$ and
let $p=u_{r'}$ be the last point that is a right induction vertex (of $H_{r'}$)
or $p=u_{r'}=u_{0}$ if there are no right induction vertices in the sequence.
Note that in the second case, $p$ cannot be a base vertex of $T_1$ 
(since $T_{1n}$ is standard) and so $p$ must be the left vertex of
$T_1$.
The edges of path $u_{r'}, \dots u_r$ are edges $(u_{i-1},u_i)$ for every $i$
such that $r' < i \leq r$ and $u_i$ is a right vertex of $H_i$. Because $T_{1n}$
contains no irregular edges and the fact that $u_i$ cannot lie on the $e$ side
of $H_i$ (because if it did $u_i$ would be a right induction vertex) or
the $s_e$ side of $H_i$ (by Lemma~\ref{lem:props}), the endpoints $u_{i-1}$ 
and $u_i$ lie on the $w$ or $n_w$
and $n_w$ or $n_e$ sides, respectively, of $H_i$. This implies that
${\tt x}(u_{i-1}) < {\tt x}(u_i)$ and that the length of $(u_{i-1},u_i)$ is
bounded by
$\sqrt{3}d_x(u_{i-1},u_i) - ({\tt y}(u_{i-1})-{\tt y}(u_i))$, as illustrated in 
Fig.~\ref{fig:caseA}. 
Therefore, ${\tt x}(p=u_{r'}) < {\tt x}(u_r)$ and the length of the path 
$p=u_{r'}, \dots, u_r$ (path $\mathcal{P}$)
in $T_{1n}$ is at most $\sqrt{3}d_x(p,u_{r})-({\tt y}(p)-{\tt y}(u_r))$. 

\begin{figure}[!b]
\center{\caseA}
\caption{Illustration of case A in the proof of Lemma~\ref{le:boundedGentelSections}.
Shown is the case when the gentle path from $u_r$ to $l_s$ is just a gentle
edge. For every $i$ such that $r' < i \leq r$ and $u_i$ is
the right vertex of $H_i$, hexagon $H_i$ and the edge $(u_{i-1},u_i)$ are shown
in red; the length of $(u_{i-1},u_i)$ is bounded by
$\sqrt{3}d_x(u_{i-1},u_i) - ({\tt y}(u_{i-1})-{\tt y}(u_i))$, a value equal
to the total length of the two intersecting, red, dashed segments going north
from $u_{i-1}$ and north-west from $u_i$. The total length of the two dashed
blue segments is an upper bound on the length of the edge $(u_r,l_s)$ and
the total length of the dotted red line segments represent the upper bound
$\sqrt{3}d_x(p,q) - ({\tt y}(p)-{\tt y}(q))$ on the length of the path 
$p=u_{r'}, \dots, u_r, l_{s}, l_{s+1}, \dots, l_{s'}=q$.}
\label{fig:caseA}
\end{figure}

Similarly, we consider the sequence of points $l_s, l_{s+1}, \dots, l_n$ and
set $q=l_{s'}$ to be the first point that is a left induction vertex (of $H_{s'+1}$)
or $q = l_{s'} = l_n$ if there are no left induction vertices in the sequence.
For every $i$ such that $s \leq i < s'$, if $l_i$ is a left vertex of $H_{i+1}$
then edge $(l_{i},l_{i+1})$ has endpoints $l_{i}$ and $l_{i+1}$ lying on the 
$s_w$ or $s_e$ and $s_e$ and $e$ sides, respectively, of $H_{i+1}$. This implies
that ${\tt x}(l_{i}) < {\tt x}(l_{i+1})$ and that the length of
$(l_{i},l_{i+1})$ is bounded by 
$\sqrt{3}d_x(l_{i},l_{i+1}) - ({\tt y}(l_{i})-{\tt y}(l_{i+1}))$. Therefore, 
${\tt x}(l_s) < {\tt x}(l_{s'}=q)$ and the length of the path 
$l_{s}, l_{s+1}, \dots, l_{s'}=q$ (path $\mathcal{Q}$) is at most
$\sqrt{3}d_x(l_{s},q)-({\tt y}(l_s)-{\tt y}(q))$, as shown in 
Fig.~\ref{fig:caseA}. 
We combine the paths $\mathcal{P}$ and $\mathcal{Q}$ with the gentle path from
$u_r$ to $l_s$ and this combined
path is gentle, and therefore canonical, since $r' < s'$,  
${\tt x}(p=u_{r'}) < {\tt x}(l_{s'}=q)$, and:
\begin{align*}
d_T (p, q) & \leq d_T(p, u_r) + d_T(u_r, l_s) + d_T(l_s, q) \\
& \leq \sqrt{3}d_x(p,u_r)-({\tt y}(p)-{\tt y}(u_r))+ \sqrt{3}d_x(u_r,l_s))-({\tt y}(u_r)-{\tt y}(l_s)) \\
 & \quad + \sqrt{3}d_x(l_s, q)) - ({\tt y}(l_s) - {\tt y}(q)) \\
& = \sqrt{3}d_x(p,q)-({\tt y}(p)-{\tt y}(q)).
\end{align*}

{\bf Case B.} 
We now consider the case when $T_{1n}$ contains an irregular gentle edge. 
We assume
w.l.o.g. that the gentle edge has an endpoint $u_r$ that lies on side $s_w$ of
$H_{r+1}$ and also that the gentle edge is the first such edge (in the sense
that for all $i$ such that $0 \leq i < r$, $u_i$ does not lie on side $s_w$
of $H_{i+1}$). (For the case when $u_r$ lies on side $s_e$ of $H_r$ we would
consider the last such edge and use an argument that is symmetric to the one
we make below; the cases when the irregular gentle edge has an endpoint $l_r$
on side $n_w$ of $H_{r+1}$ or on side
$n_e$ of $H_r$ are symmetric to the cases when $u_r$ is on side $s_w$
of $H_{r+1}$ and side $s_e$ of $H_r$, respectively.)

By Lemma~\ref{lem:props}, the irregular gentle edge under consideration is 
$(u_r,l_{r+1})$, with $l_{r+1}$ lying on side $s_e$ or $e$ of $H_{r+1}$. Let $p$
and $q$ be as defined in Case A. Because of our assumption that $(u_r,l_{r+1})$
is first, the case A arguments can be applied to obtain 
${\tt x}(p) < {\tt x}(u_r)$ and bound the distance from $p$ to
$u_r$ by $\sqrt{3}d_x(p,u_r)-({\tt y}(p)-{\tt y}(u_r))$. We cannot do the same
to bound the distance from $l_{r+1}$ to $q = l_{s'}$ because it is possible that for
some $i$ such that $r+1 < i \leq s'$, $l_i$ lies on side $n_e$ of $H_i$ and
${\tt x}(l_{i-1}) > {\tt x}(l_i)$. So we proceed instead with induction and
prove that
\begin{equation*}
\label{eq:induction}
{\tt x}(u_r) < {\tt x}(l_i) \mbox{ and } d_T(u_r, l_i) \leq \sqrt{3}d_x(u_r, l_i) - ({\tt y}(u_r) - {\tt y}(l_i)) \mbox{ for every $i$ such that } r < i \leq s'
\end{equation*}
which will complete the proof.

The base case $i=r+1$ holds because $(u_r,l_{r+1})$ is gentle and 
${\tt x}(u_r) < {\tt x}(l_{r+1})$. For the induction step, we assume that 
${\tt x}(u_r) < {\tt x}(l_i)$ and 
$d_T(u_r, l_i) \leq \sqrt{3}d_x(u_r, l_i)-({\tt y}(u_r)-{\tt y}(l_i))$ 
for all $i$ such that $r < i < s \leq s'$ and show that the
inequality holds for $i = s$ as well. If $l_s = l_{s-1}$, that is trivially
true. Otherwise, $l_s$ is a right vertex of $H_s$. If
$l_s$ lies on the $e$ or $s_e$ side of $H_s$ then we use the arguments from
Case A to get that ${\tt x}(l_{s-1}) < {\tt x}(l_s)$, that the length of edge
$(l_{s-1}, l_s)$ is less than
$\sqrt{3}d_x(l_{s-1}, l_s) - ({\tt y}(l_{s-1}) - {\tt y}(l_s))$, 
and that therefore the inductive step again easily holds.

\begin{figure}
\center{\caseB}
\vspace{-1.2cm}
\center{(a) \hspace{6.35cm} (b)}

\caption{Illustration of {\em Case B.1} in the proof of 
Lemma~\ref{le:boundedGentelSections}. (a) $T_{(r+1)s}$. Point $u_r$ lies on side 
$s_w$ of $H_{r+1}$, point $l_s$ lies on side $n_e$ of $H_s$, and no
$(l_t, u_t)$, for $r < t < s$, has positive slope. (b) $\mathcal{R}(T_{(r+1)s})$
satisfies the conditions of Lemma~\ref{le:mainlemmaB} with $p = u_r$ and
$q = l_s$; because the slope of the line passing through $u_r$ and $l_s$ is less
than $-\frac{1}{\sqrt{3}}$, the length of the red, dashed segments (shown also in (a))
 is an upper bound on the length of a path in $T_{(r+1)s}$ from $u_r$ to $l_s$.}
\label{fig:caseB}
\end{figure}

If, however, $l_s$ lies on side $n_e$ of $H_s$ then we cannot use the Case A
argument because ${\tt x}(l_{s-1}) < {\tt x}(l_s)$ is not necessarily true.
In that case, by Lemma~\ref{lem:props} $(u_{s-1},l_s)$ must be an
(irregular) gentle edge with $u_{s-1}$ lying on the $w$ or $n_w$ side of $H_s$
(as illustrated in 
Fig.~\ref{fig:caseB}-(a)). First we note that ${\tt x}(u_r) < {\tt x}(l_s)$
holds because otherwise $l_s$ would appear below and to the right of $u_r$ and
then $(u_{s-1},l_s)$ would appear before $(u_r,l_{r+1})$
when traveling along segment $[st]$ from $s$ to $t$, a contradiction.
In order to bound $d_T(u_r,l_s)$, 
we consider indices $r+1,...,s-1$ and set $t$ to be
the last one such that edge $(l_t,u_t)$ has positive slope 
(i.e., ${\tt x}(l_t) < {\tt x}(u_t)$ and ${\tt y}(l_t) < {\tt y}(u_t))$, 
if one exists. We consider two subcases:

{\em Case B.1.} If no $(l_t, u_t)$, for $r < t < s$,
has positive slope  (as is the case in Fig.~\ref{fig:caseB}-(a)) we consider
the transformation of the linear sequence 
$T_{(r+1)s}$ obtained by rotating the plane clockwise by
an angle of $\pi/3$, 
as illustrated in Fig.~\ref{fig:caseB}-(b). 


In the rotated $T_{(r+1)s}$, which we denote $\mathcal{R}(T_{(r+1)s})$, 
we have $-\sqrt{3} < m_{st} < -\frac{1}{\sqrt{3}}$, $u_r$ is a left induction 
vertex of $T_{r+1}$, and $l_s$ is a right induction vertex of $T_s$. We show
next that $\mathcal{R}(T_{(r+1)s})$ is regular. If,
in $\mathcal{R}(T_{(r+1)s})$, $u_t$ were to lie on a $s$ side of $H_t$ for some
$r < t < s$ then the slope of $(l_t,u_t)$ would have to be greater than
$m_{st}$ and less than the slope of the $s_e$ side of hexagon $H_t$. So,
the slope of $(l_t,u_t)$ in $\mathcal{R}(T_{(r+1)s})$ would have to be
between $-\sqrt{3}$ and $\frac{1}{\sqrt{3}}$ which would imply that
$(l_t,u_t)$ has positive slope in $T_{(r+1)s}$, a contradiction. Note
that if $u_{t-1}$ were to lie on a $s$ side of $H_t$ then so would $u_t$ so
we need not consider that case. We similarly show that no $l_{t-1}$ or $l_t$
lies on a $n$ side of $H_t$. Finally, we note that a gentle edge
$(l_t, u_t)$ in $\mathcal{R}(T_{(r+1)s})$ would have to have positive slope in
$T_{(r+1)s}$, a contradiction. So $\mathcal{R}(T_{(r+1)s})$ contains no gentle
edge and it is thus regular.

All the conditions of (Technical) Lemma~\ref{le:mainlemmaA} therefore apply to
$\mathcal{R}(T_{(r+1)s})$. By the lemma, the distance from $u_r$ to $l_s$
in $\mathcal{R}(T_{(r+1)s})$ (and thus in the original $T_{(r+1)s}$ as well) is
bounded by 
$\frac{4}{\sqrt{3}} d^{\mathcal{R}}_x(u_r,l_s)$, where $d^{\mathcal{R}}_x(u_r,l_s)$
is the difference
between the abscissas of $u_r$ and $l_s$ in $\mathcal{R}(T_{(r+1)s})$). Note that
$\frac{4}{\sqrt{3}} d^{\mathcal{R}}_x(u_r,l_s)$ is the sum of the lengths of
two sides of an equilateral triangle of height $d^{\mathcal{R}}_x(u_r,l_s)$.
Therefore, because
the slope of the line through $u_r$ and $l_s$ in $\mathcal{R}(T_{(r+1)s})$ is
less than $-\frac{1}{\sqrt{3}}$, 
$\frac{4}{\sqrt{3}} d^{\mathcal{R}}_x(u_r,l_s)$ is less than the
length of the piecewise linear curve, shown in Fig.~\ref{fig:caseB}-(b),
consisting of a (longer) vertical segment down from $u_r$ followed by a
(shorter) segment with slope $\frac{1}{\sqrt{3}}$ to $l_s$. The length of that
curve is exactly $\sqrt{3}d_x(u_r,l_s)-({\tt y}(u_r)-{\tt y}(l_s))$ in
$T_{(r+1)s}$ as illustrated in Fig.~\ref{fig:caseB}-(a), which completes the
proof in this case.

{\em Case B.2.} If $(l_t,u_t)$ exists, as illustrated in
Fig.~\ref{fig:othercase}, then ${\tt x}(l_t) < {\tt x}(u_t)$ and 
$d_T(l_t,u_t) \leq \sqrt{3}d_x(l_t,u_t)-({\tt y}(l_t)-{\tt y}(u_t))$. Also,
by induction, ${\tt x}(u_r) \leq {\tt x}(l_t)$ and 
$d_T(u_r,l_t) \leq \sqrt{3}d_x(u_r,l_t)-({\tt y}(u_r)-{\tt y}(l_t))$.
What remains to be done is to show that ${\tt x}(u_t) < {\tt x}(l_s)$ and 
that $d_T(u_t,l_s) \leq \sqrt{3}d_x(u_t,l_s)-({\tt y}(u_t)-{\tt y}(l_s))$. 
We now have two cases to consider. Note first that, because $t$ is last, for 
every $j$ such that $t < j < s$, $u_j$ cannot lie on side $e$ or $s_e$ of $H_j$
(because otherwise $(l_j,u_j)$ would have to have positive slope); this means
that $u_j$ cannot be a right induction point of $H_j$.

\begin{figure}
\center{\othercase}

\caption{Illustration of {\em Case B.2} in the proof of 
Lemma~\ref{le:boundedGentelSections}. If $(l_t,u_t)$ has positive slope then
${\tt x}(l_t) < {\tt x}(u_t)$ and 
$d_T(l_t,u_t) \leq \sqrt{3}d_x(l_t,u_t)-({\tt y}(l_t)-{\tt y}(u_t))$.
By induction, ${\tt x}(u_r) \leq {\tt x}(l_t)$ and 
$d_T(u_r,l_t) \leq \sqrt{3}d_x(u_r,l_t)-({\tt y}(u_r)-{\tt y}(l_t))$. In
{\em Case B.2.i}, no $u_{j-1}$ lies on side $s_w$ of $H_j$, where $t<j<s$, and so
${\tt x}(u_t) \leq \dots \leq {\tt x}(u_{s-1}) \leq {\tt x}(l_s)$ and 
$d_T(u_u,l_s) \leq \sqrt{3}d_x(u_t,l_s)-({\tt y}(u_t)-{\tt y}(l_s))$.}
\label{fig:othercase}
\end{figure}

{\em Case B.2.i.} If we also have that no $u_{j-1}$ lies on side $s_w$ of $H_j$,
where $t<j<s$ and as illustrated in Fig.~\ref{fig:othercase}, 
then, using the arguments from Case A, we can show that
${\tt x}(u_t) \leq \dots \leq {\tt x}(u_{s-1}) \leq {\tt x}(l_s)$ and 
can bound the length of the path $u_t, \dots, u_{s-1}, l_s$ with 
$\sqrt{3}d_x(u_t,l_s)-({\tt y}(u_t)-{\tt y}(l_s))$. 


{\em Case B.2.ii.} Finally, suppose that $u_{t'-1}$, for some $t'$ such that
$t < t' < s$, lies on side $s_w$ of $H_{t'}$ and let us assume that $t'$ is
leftmost (in the sense that no $u_{j-1}$ lies on side $s_w$ of $H_j$ for $j$
such that $t < j < t'$). Using arguments from Case A we obtain that 
${\tt x}(u_t) < {\tt x}(u_{t'})$ and 
$d_T(u_t, u_{t'}) \leq \sqrt{3}d_x(u_t, u_{t'})) - ({\tt y}(u_t) -
{\tt y}(u_{t'} ))$. Using the approach from {\em Case B.1} we get
${\tt x}(u_{t'}) < {\tt x}(l_s)$ and 
$d_T(u_{t'}, l_s) \leq \sqrt{3}d_x(u_{t'}, l_s) - ({\tt y}(u_{t'})-{\tt y}(l_s))$.
We complete the proof of the lemma by combining all these inequalities.
\end{proof}

}

\section{Proof of (Technical) Lemma \ref{le:mainlemmaA}}
\label{sec:proofA}
\iftoggle{abstract}
{We prove this lemma via a framework that uses continuous versions
of the discrete functions ($p_N$, $p_S$, etc.) informally introduced in
Subsection~\ref{sub:keylemma}. We start by defining functions $H(x)$, $T(x)$, $u(x)$,
$\ell(x)$, ${\tt r}(x)$, $w(x)$, and $e(x)$ for
${\tt x}(p) \leq x \leq {\tt x}(q)$ as illustrated in
Fig.~\ref{fig:notation}.

\begin{figure}[h]
\center{\notation}
\caption{Let point $c_i$ be the center of hexagon $H_i$, for $i=1,\dots,n$. For $x$ such that
${\tt x}(c_i) \leq x < {\tt x}(c_{i+1})$, $H(x)$ is the hexagon whose center has abscissa
$x$ and that has points $u_i=u(x)$ and $l_i=\ell(x)$ on its boundary.
Intuitively, function $H(x)$ from $x= {\tt x}(c_i)$ to $x= {\tt x}(c_{i+1})$
models the ``pushing''
of hexagon $H_i$ through $u_i$ and $l_i$ up until it becomes $H_{i+1}$. Function
${\tt r}(x)$ is the minimum radius of $H(x)$ and $w(x) = x-{\tt r}(x)$ and
$e(x) = x + {\tt r}(x)$ are the abscissas of the $w$ and $e$ sides, respectively,
of $H(x)$. Finally, we define $T(x) = T_{1i}$ when ${\tt x}(c_i) \leq x < {\tt x}(c_{i+1})$.} 
\label{fig:notation}
\end{figure}

For a point $o$ on a side of $H(x)$, we define functions $p_N(o, x)$ and 
$p_S(o, x)$ as the {\em signed} shortest distances around the perimeter of
$H(x)$ to the $N$ vertex and $S$ vertex, respectively, with sign 
$\sgn(x - x(o))$. As Fig.~\ref{fig:dNdS}-(a) and Fig.~\ref{fig:dNdS}-(b)
illustrate, these signs are positive for $o$ on the $n_w$, $w$, or $s_w$
sides of $H(x)$ and negative for $o$ on $n_e$, $e$, or $s_e$ sides. 
We omit $o$ and use the shorthand notation $p_N(x)$ if $o = u(x)$ and
$p_S(x)$ if $o = \ell(x)$. 

\begin{figure}
\dNdS

\hspace{1.45cm} (a) \hspace{4.35cm} (b) \hspace{4.55cm} (c)
\caption{(a) The values of $p_N(o, x)$ are shown, for various points $o$
lying on the boundary of $H(x)$, as signed hexagon arc lengths. (b) The 
values of $p_S(o, x)$ are shown similarly. (c) The length of edge
$(u_{i-1},u_i)$, with $u_{i-1},u_i$ lying on the boundary of $H(x)=H_i$, 
is bounded by $p_N(u_{i-1}, x) - p_N(u_i,x)$.}
\label{fig:dNdS}
\end{figure}

Functions $U(x)$ and $L(x)$, used to bound the length of the shortest path from $p$ to $q$ and illustrated in Fig.~\ref{fig:UandL}-(a), are defined as follows for ${\tt x}(p) \leq x \leq {\tt x}(q)$:
\begin{align*}
& U(x) = d_{T(x)}(p,u(x)) + p_N(x) \mbox{\hspace{2cm}} L(x) = d_{T(x)}(p,\ell(x)) + p_S(x)
\end{align*}
We note that $U({\tt x}(q))+L({\tt x}(q))$ is exactly twice the distance in 
$T_{1n}$ from $p$ to $q$. We will compute an upper bound for function $U+L$ 
by bounding its growth rate.

Functions ${\bar U}(x)$ and ${\bar L}(x)$, used to bound the lengths of the upper and lower paths in $T_{1n}$ and illustrated in Fig.~\ref{fig:UandL}-(b), are defined as follows for ${\tt x}(c_i) \leq x \leq {\tt x}(c_{i+1})$:

\begin{figure}[!b]
\begin{center}\UandL

(a) \hspace{6cm} (b)
\end{center}
\caption{(a) Definition of $U(x)$ and $L(x)$. For example, $U(x)$ for 
${\tt x}(c_i) \leq x < {\tt x}(c_{i+1})$ is the sum
of the length of the shortest path from $p$ to $u_i$ in $T_{1i}$ 
(illustrated as the red dashed path) and $p_N(x)$ (of negative value and 
represented as a red arrow). (b) Definition of $\bar{U}(x)$ and $\bar{L}(x)$. 
When ${\tt x}(c_i) \leq x < {\tt x}(c_{i+1})$ for example,
$\bar{U}(x) - p_N(x)$ is an upper bound (equal to the length of the 
sequence of red dashed
hexagon arcs going from $p$ to $u_i$) on the length
of the upper path $p, u_0, u_1, \dots, u_{i-1}, u_i$.}
\label{fig:UandL}
\end{figure}

\begin{align*}
\bar{U}(x) & = \sum_{j=1}^{i} (p_N(u({\tt x}(c_{j-1})), {\tt x}(c_{j})) - p_N(u({\tt x}(c_{j})),{\tt x}(c_{j}))) + p_N(x) \\
\bar{L}(x) & = \sum_{j=1}^{i} (p_S(\ell({\tt x}(c_{j-1})), {\tt x}(c_{j})) - p_S(\ell({\tt x}(c_{j})),{\tt x}(c_{j}))) + p_S(x)
\end{align*}

Functions ${\bar U}(x)$ and ${\bar L}(x)$, 
as well as $U(x)$ and $L(x)$, have rates of growth when
${\tt x}(c_i) < x < {\tt x}(c_{i+1})$ that depend solely on the
last term ($p_N(x)$ or $p_S(x)$).
We show that functions $p_N$ and 
$p_S$ are monotonically increasing piecewise linear 
 and bound the rate of growth of $p_N$ and $p_S$ using elementary geometric 
arguments illustrated in Fig.~\ref{fig:growth}. Figure~\ref{fig:growth}-(c)
illustrates a case when the growth rate of $p_N+p_S$, and therefore also of 
${\bar U}+{\bar L}$ and of $U+L$, is $\frac{8}{\sqrt{3}}$.

\begin{figure}
\growth

\hspace{2cm} (a) \hspace{3.8cm} (b) \hspace{4.2cm} (c)
\caption{Constructions demonstrating growth rates, with respect to 
$\Delta x = 1$, of $p_N$, $p_S$ and other
functions for three different placements of $u(x)=u$ and $\ell(x)=l$ on the
boundary of $H(x)$. 
}
\label{fig:growth}
\end{figure}

}
{\begingroup
\def\thetheorem{\ref{le:mainlemmaA}}
\begin{lemma}[The Technical Lemma]
If $T_{1n}$ is a regular linear sequence of triangles then there is a path in
$T_{1n}$ from
the left induction vertex $p$ of $T_1$ to the right induction vertex $q$ of
$T_n$ of length at most  $\frac{4}{\sqrt{3}} d_x(p,q)$.
\end{lemma}
\addtocounter{theorem}{-1}
\endgroup

To prove this lemma we develop a framework that uses continuous versions
of the discrete functions ($p_N$, $p_S$, etc.) informally introduced in
Subsection~\ref{sub:keylemma}. We start by defining functions $H(x)$, $u(x)$,
and $\ell(x)$ for
${\tt x}(p) \leq x \leq {\tt x}(q)$:
\begin{itemize}
\item If $c_i$ is the center of hexagon $H_i$,
for $i = 1, \dots, n$, we define $H({\tt x}(c_i)) = H_i$. We also define
$u({\tt x}(c_i))$ and $\ell({\tt x}(c_i))$ to be $u_i$ and $l_i$,
respectively. 
\item Then, for every $i = 1, \dots, n-1$ and $x$ such that
${\tt x}(c_i) < x < {\tt x}(c_{i+1})$,
we define $H(x)$ to be the hexagon whose center has abscissa $x$ and that
has points $u_i$ and $l_i$ on its boundary; we also define $u(x)$ to be $u_i$
and $\ell(x)$ to be $l_i$ (see Fig.~\ref{fig:notation}).
\iftoggle{abstract}
{\begin{figure}}
{\begin{figure}}
\center{\notation}
\caption{For $x$ such that
${\tt x}(c_i) < x < {\tt x}(c_{i+1})$, $H(x)$ is the hexagon whose center has abscissa
$x$ and that has points $u_i=u(x)$ and $l_i=\ell(x)$ on its boundary.
Intuitively, function $H(x)$ from $x= {\tt x}(c_i)$ to $x= {\tt x}(c_{i+1})$
models the ``pushing''
of hexagon $H_i$ through $u_i$ and $l_i$ up until it becomes $H_{i+1}$. Function
${\tt r}(x)$ is the minimum radius of $H(x)$ and $w(x) = x-{\tt r}(x)$ and
$e(x) = x + {\tt r}(x)$ are the abscissas of the $w$ and $e$ sides, respectively,
of $H(x)$. $N(x)$ and $S(x)$ are the vertices $N$ and $S$ of $H(x)$.}
\label{fig:notation}
\end{figure}
We note that $H(x)$ is
uniquely defined because $T_{1n}$ contains no gentle edges and so $u(x)$ and
$\ell(x)$ cannot lie on sides $e$ and $w$, respectively, or on sides $w$ and
$e$, respectively, of $H(x)$.
\item
As we will soon see, function $H(x)$ has a specific growth pattern that depends
on what sides of $H(x)$ points $u(x)$ and $\ell(x)$ lie on. In order to
simplify our presentation, we define $H(x)$ when 
${\tt x}(p) \leq x < {\tt x}(c_1)$ and ${\tt x}(c_n) < x \leq {\tt x}(q)$ 
in a way that fits that pattern. Let $W_S^*$ be vertex
$W_S$ of $H_1=H({\tt x}(c_1))$ and let $H^*$ be the hexagon with $p$ and $W_S^*$
as its $W_N$ and $W_S$ vertices, respectively. Let $c^*$ be the center of $H^*$.
When ${\tt x}(p) \leq x \leq {\tt x}(c^*)$ we define $H(x)$ to be the hexagon 
whose center has abscissa $x$ and that has point $p$ as its $W_N$ vertex; 
we also define $u(x) = \ell(x) = p$. When 
${\tt x}(c^*) \leq x < {\tt x}(c_1)$ we define $H(x)$ to be the hexagon whose
center has abscissa $x$ and that has point $W_S^*$ as its $W_S$ vertex; we also
define $u(x) = \ell(x) = p$. We define
$H(x)$ when ${\tt x}(c_n) < x \leq {\tt x}(q)$ in a symmetric fashion
with $u(x) = \ell(x) = q$ in that case.
\end{itemize}

Next, we define ${\tt r}(x)$ to be the minimum radius of $H(x)$ (i.e., the
distance between the center of $H(x)$ to its $w$ side). Note that hexagons 
$H({\tt x}(p))$ and
$H({\tt x}(q))$ both have radius $0$ and define their centers to be $c_0 = p$
and $c_{n+1}=q$. We also extend the notation $T_{ij}$ to include $T_{10} = c_0$
and define $T(x) = T_{1i}$ when ${\tt x}(c_i) \leq x < {\tt x}(c_{i+1})$ and
$T({\tt x}(c_{n+1}))=T_{1n}$.
We define $N(x)$ and $S(x)$ to be the $N$ and $S$ vertex, respectively, of
$H(x)$. Finally, we define functions $w(x) = x-{\tt r}(x)$ and
$e(x) = x + {\tt r}(x)$ that keep track of the abscissa of the $w$ and $e$
sides, respectively, of $H(x)$ (refer to Fig.~\ref{fig:notation}). Note that
functions ${\tt r}(x), w(x), e(x)$ as well as functions ${\tt y}(N(x))$ and
${\tt y}(S(x))$ (the ordinates of $N(x)$ and $N(y)$, resp.) are continuous.

\begin{figure}[!b]
\dNdS

\vspace{0.3cm}
\hspace{1.45cm} (a) \hspace{4.35cm} (b) \hspace{4.55cm} (c)
\caption{(a) The values of $p_N(o, x)$ are shown, for various points $o$
lying on the boundary of $H(x)$, as signed hexagon arc lengths. (b) The 
values of $p_S(o, x)$ are shown similarly. (c) The length of edge
$(u_{i-1},u_i)$, with $u_{i-1},u_i$ lying on the boundary of $H(x)=H_i$, 
is bounded by $p_N(u_{i-1}, x) - p_N(u_i,x)$.}
\label{fig:dNdS}
\end{figure}

For a point $o$ on a side of $H(x)$, we define functions $p_N(o, x)$ and 
$p_S(o, x)$ as the {\em signed} shortest distances around the perimeter of
$H(x)$ to the $N$ vertex and $S$ vertex, respectively, with sign 
$\sgn(x - x(o))$. As Fig.~\ref{fig:dNdS}-(a) and Fig.~\ref{fig:dNdS}-(b)
illustrate, these signs are positive for $o$ on the $n_w$, $w$, or $s_w$
sides of $H(x)$ and negative for $o$ on $n_e$, $e$, or $s_e$ sides. 
We omit $o$ and use the shorthand notation $p_N(x)$ if $o = u(x)$ and
$p_S(x)$ if $o = \ell(x)$. 

\begin{figure}
\begin{center}\UandL

\vspace{0.6cm}
(a) \hspace{6cm} (b)
\end{center}
\caption{(a) Definition of $U(x)$ and $L(x)$. For example, $U(x)$ for 
${\tt x}(c_i) \leq x < {\tt x}(c_{i+1})$ is the sum
of the length of the shortest path from $p$ to $u_i$ in $T_{1i}$ 
(illustrated as the red dashed path) and $p_N(x)$ (of negative value and 
represented as a red arrow). (b) Definition of $\bar{U}(x)$ and $\bar{L}(x)$. $\bar{U}(x) - p_N(x)$,
for example, is an upper bound (equal to the length of the 
sequence of red dashed
hexagon arcs going from $p$ to $u_i$) on the length
of the upper path $p, u_0, u_1, \dots, u_{i-1}, u_i$.}
\label{fig:UandL}
\end{figure}

The key functions $U(x)$ and $L(x)$, illustrated in Fig.~\ref{fig:UandL}-(a), 
are defined as follows for
${\tt x}(p) = {\tt x}(c_0) \leq x \leq {\tt x}(c_{n+1}) = {\tt x}(q)$:
\begin{equation*}
U(x) = d_{T(x)}(p,u(x)) + p_N(x) \mbox{\hspace{2cm}} L(x) = d_{T(x)}(p,\ell(x)) + p_S(x)
\end{equation*}
Finally, we define the {\em potential function} $P(x)$ to be $U(x) + L(x)$. 
We note that
$P({\tt x}(q))$ is exactly twice the distance in $T_{1n}$ from $p$ to $q$. The 
main goal of this paper is to compute an upper bound for function $P(x)$.
We will do this by bounding its growth rate.

For ${\tt x}(c_i) < x < {\tt x}(c_{i+1})$, the terms 
$d_{T(x)}(p,u(x))$ and $d_{T(x)}(p,\ell(x))$ in the definitions of 
$U(x)$ and $L(x)$, respectively, are constant. 
This means that for such $x$ the rate of growth of
functions $U(x)$ and $L(x)$ is determined solely by the terms 
$p_N(x)$ and $p_S(x)$, respectively.

We note that the length of each edge $(u_{i-1},u_i)$ 
(assuming $u_{i-1} \not= u_i$) can be bounded by the distance
from $u_{i-1}$ to $u_i$ when traveling clockwise along the sides of $H_i$.
This distance is exactly $p_N(u_{i-1}, {\tt x}(c_i)) - p_N(u_i,{\tt x}(c_i))$ as
illustrated in Fig.~\ref{fig:dNdS}-(c).
This, and a similar observation about each edge $(l_{i-1},l_i)$, motivates
the following definitions of functions $\bar{U}(x)$ and $\bar{L}(x)$ that we
will use to bound lengths of subpaths of the upper path $p,u_0, \dots,u_n,q$ 
and the lower path $p, l_0, \dots, l_n,q$ (see also Fig.~\ref{fig:UandL}-(b)). 
When ${\tt x}(c_i) \leq x < {\tt x}(c_{i+1})$ or $x = {\tt x}(c_{n+1}) = {\tt x}(q)$, we define
\begin{align*} 
\bar{U}(x) & = \sum_{j=1}^{i} (p_N(u({\tt x}(c_{j-1})), {\tt x}(c_{j})) - p_N(u({\tt x}(c_{j})),{\tt x}(c_{j}))) + p_N(x) \\
\bar{L}(x) & = \sum_{j=1}^{i} (p_S(\ell({\tt x}(c_{j-1})), {\tt x}(c_{j})) - p_S(\ell({\tt x}(c_{j})),{\tt x}(c_{j}))) + p_S(x)
\end{align*}
and let ${\bar P}(x)={\bar U}(x) + {\bar L}(x)$. We note that
${\bar P}(x)$ is an upper bound for $P(x)$ and that ${\bar P}({\tt x}(q))$
bounds the sum of the lengths of paths $p,u_0, \dots,u_n,q$ 
and $p, l_0, \dots, l_n,q$. Functions ${\bar U}(x)$ and ${\bar L}(x)$, 
just like $U(x)$ and $L(x)$, have rates of growth when
${\tt x}(c_i) < x < {\tt x}(c_{i+1})$ that are determined solely by the
last term ($p_N(x)$ or $p_S(x)$). We will show that functions 
$p_N(x)=p_N(u(x),x)$ and $p_S(x) = p_S(\ell(x),x)$ are monotonically
increasing piecewise linear functions whose rates of growth depend solely
on the sides of $H(x)$ that $u(x)$ and $\ell(x)$ lie on. In order to capture
precisely this rate of growth, we define the {\em transition
function} $t(x)$ to be {\em transition} $t_{ij}$ if $\ell(x)$ lies in the interior
of side $i$ and $u(x)$ lies in the interior of side $j$. We use the wildcard
notations $t_{\ast j}$ and $t_{i \ast}$ to refer to any
transition with $u(x)$ on side $j$ and $\ell(x)$ on side $i$, respectively,
of $H(x)$. 

\newpage
\begin{lemma}
\label{lem:growthrates}
Given the assumptions of Lemma~\ref{le:mainlemmaA}, for every 
$x \in [{\tt x}(p), {\tt x}(q)]$:
\begin{itemize}
\item $t(x)$, when defined, is one of
\begin{equation*}
t_{wn_w},t_{wn_e},t_{s_ww},  t_{s_wn_w}, t_{s_wn_e},  t_{s_we},  t_{s_ew},  t_{s_en_w},  t_{s_en_e},  t_{s_ee},  t_{en_w},  t_{en_e}.
\end{equation*}
\item Functions ${\tt y}(N(x)), {\tt y}(S(x)), {\tt r}(x), w(x), e(x)$
and, for $x \in [{\tt x}(c_i), {\tt x}(c_{i+1}))$ and \newline $i=0,\dots,n-1$,
functions $p_N(x)$, $p_S(x)$, and ${\bar P}(x)$ are all
piecewise linear functions with the following growth rates where defined:
\end{itemize}\vspace{-0.5cm}\begin{center}
{\tabulinesep=1.2mm
\begin{tabu} to \textwidth{X[2l]X[r]X[r]X[r]X[r]X[r]X[r]X[r]X[r]X[r]X[r]X[r]X[r]}
& \transitiononezero & \transitiononefive & \transitiontwoone & \transitiontwozero & \transitiontwofive & \transitiontwofour & \transitionthreeone & \transitionthreezero & \transitionthreefive & \transitionthreefour & \transitionfourzero & \transitionfourfive \\
{$\bm{t(x)}$} &  {$\bm{t_{wn_w}}$} & {$\bm{t_{wn_e}}$} &  {$\bm{t_{s_ww}}$} & {$\bm{t_{s_wn_w}}$} &  {$\bm{t_{s_wn_e}}$} & {$\bm{t_{s_we}}$} &  {$\bm{t_{s_ew}}$} & {$\bm{t_{s_en_w}}$} &  {$\bm{t_{s_en_e}}$} & {$\bm{t_{s_ee}}$} &  {$\bm{t_{en_w}}$} & {$\bm{t_{en_e}}$} \\ \hline
{$\bm{\frac{\Delta {\bar P}(x)}{\Delta x}}$ } & $\frac{6}{\sqrt{3}}$ & $\frac{8}{\sqrt{3}}$ & $\frac{6}{\sqrt{3}}$ & $\frac{4}{\sqrt{3}}$ & $\frac{4}{\sqrt{3}}$ & $\frac{8}{\sqrt{3}}$ & $\frac{8}{\sqrt{3}}$ & $\frac{4}{\sqrt{3}}$ & $\frac{4}{\sqrt{3}}$ & $\frac{6}{\sqrt{3}}$ & $\frac{8}{\sqrt{3}}$ & $\frac{6}{\sqrt{3}}$ \\ \hline
{$\bm{\frac{\Delta p_N(x)}{\Delta x}}$ } & $\frac{2}{\sqrt{3}}$ & $\frac{2}{\sqrt{3}}$ & $\frac{4}{\sqrt{3}}$ & $\frac{2}{\sqrt{3}}$ & $\frac{2}{\sqrt{3}}$ & $\frac{6}{\sqrt{3}}$ & $\frac{6}{\sqrt{3}}$ & $\frac{2}{\sqrt{3}}$ & $\frac{2}{\sqrt{3}}$ & $\frac{4}{\sqrt{3}}$ & $\frac{2}{\sqrt{3}}$ & $\frac{2}{\sqrt{3}}$ \\ \hline
{$\bm{\frac{\Delta {\tt y}(N(x))}{\Delta x}}$ } & $\frac{1}{\sqrt{3}}$ & $-\frac{1}{\sqrt{3}}$ & $\frac{3}{\sqrt{3}}$ & $\frac{1}{\sqrt{3}}$ & $-\frac{1}{\sqrt{3}}$ & $-\frac{5}{\sqrt{3}}$ & $\frac{5}{\sqrt{3}}$ & $\frac{1}{\sqrt{3}}$ & $-\frac{1}{\sqrt{3}}$ & $-\frac{3}{\sqrt{3}}$ & $\frac{1}{\sqrt{3}}$ & $-\frac{1}{\sqrt{3}}$ \\ \hline
{$\bm{\frac{\Delta p_S(x)}{\Delta x}}$ } & $\frac{4}{\sqrt{3}}$ & $\frac{6}{\sqrt{3}}$ & $\frac{2}{\sqrt{3}}$ & $\frac{2}{\sqrt{3}}$ & $\frac{2}{\sqrt{3}}$ & $\frac{2}{\sqrt{3}}$ & $\frac{2}{\sqrt{3}}$ & $\frac{2}{\sqrt{3}}$ & $\frac{2}{\sqrt{3}}$ & $\frac{2}{\sqrt{3}}$ & $\frac{6}{\sqrt{3}}$ & $\frac{4}{\sqrt{3}}$ \\ \hline
{$\bm{\frac{\Delta {\tt y}(S(x))}{\Delta x}}$ } & $-\frac{3}{\sqrt{3}}$ & $-\frac{5}{\sqrt{3}}$ & $-\frac{1}{\sqrt{3}}$ & $-\frac{1}{\sqrt{3}}$ & $-\frac{1}{\sqrt{3}}$ & $-\frac{1}{\sqrt{3}}$ & $\frac{1}{\sqrt{3}}$ & $\frac{1}{\sqrt{3}}$ & $\frac{1}{\sqrt{3}}$ & $\frac{1}{\sqrt{3}}$ & $\frac{5}{\sqrt{3}}$ & $\frac{3}{\sqrt{3}}$ \\ \hline
{$\bm{\frac{\Delta {\tt r}(x)}{\Delta x}}$ } & $1$ & $1$ & $1$ & $\frac{1}{2}$ & $0$ & $-1$ & $1$ & $0$ & $-\frac{1}{2}$ & $-1$ & $-1$ & $-1$ \\ \hline
{$\bm{\frac{\Delta w(x)}{\Delta x}}$ } & $0$ & $0$ & $0$ & $\frac{1}{2}$ & $1$ & $2$ & $0$ & $1$ & $\frac{3}{2}$ & $2$ & $2$ & $2$ \\ \hline
{$\bm{\frac{\Delta e(x)}{\Delta x}}$ } & $2$ & $2$ & $2$ & $\frac{3}{2}$ & $1$ & $0$ & $2$ & $1$ & $\frac{1}{2}$ & $0$ & $0$ & $0$ \\ \hline
\end{tabu}}
\end{center}
\end{lemma}

\begin{figure}[!b]
\growth

\hspace{2cm} (a) \hspace{3.8cm} (b) \hspace{4.2cm} (c)
\caption{Constructions demonstrating Lemma~\ref{lem:growthrates} for
transitions (a) $t_{s_wn_e}$, (b) $t_{wn_w}$, and (c) $t_{s_we}$. In all three
cases the growths shown are with respect to $\Delta x = 1$. $\Delta {\tt r}(x)$
can be obtained from $\frac{1}{2}(\Delta e(x) - \Delta w(x))$ and
$\Delta P(x)$ from $\Delta p_N(x) + \Delta p_S(x)$. Note that the growth
$\Delta p_N(x)$ when $t(x)$ is a single transition can be represented using
a piecewise linear curve of length $\Delta p_N(x)$.}
\label{fig:growth}
\end{figure}

\begin{proof}
The first part follows from the definition of a regular linear sequence of
triangles. The growth rates for transitions
$t_{s_wn_e}$, $t_{wn_w}$, and $t_{s_we}$ follow from elementary
geometric constructions illustrated in Fig.~\ref{fig:growth}. The constructions
for the remaining transitions are similar.
\end{proof}


We now consider the behavior of $U(x)$, $L(x)$, ${\bar U}(x)$, and ${\bar L}(x)$
at $x = {\tt x}(c_i)$ for $i=1,\dots,n+1$. Note that 
\begin{align*}
\bar{U}({\tt x}(c_i)) & = \sum_{j=1}^{i} (p_N(u({\tt x}(c_{j-1})), {\tt x}(c_{j})) - p_N(u({\tt x}(c_{j})),{\tt x}(c_{j}))) + p_N(u_i, {\tt x}(c_i)) \\
& = \sum_{j=1}^{i-1} (p_N(u({\tt x}(c_{j-1})), {\tt x}(c_{j})) - p_N(u({\tt x}(c_{j})),{\tt x}(c_{j}))) + p_N(u({\tt x}(c_{i-1})), {\tt x}(c_i))
\end{align*}
which is the limit for ${\bar U}(x)$ when $x \rightarrow {\tt x}(c_i)$ from 
the left.
So ${\bar U}(x)$ is continuous from ${\tt x}(c_0)$ to ${\tt x}(c_{n+1})$ and,
similarly, so is ${\bar L}(x)$ and therefore ${\bar P}(x)$ as well.
Since, by Lemma~\ref{lem:growthrates}, ${\bar P}({\tt x}(q))$ is bounded by
$\frac{8}{\sqrt{3}}d_x(p,q)$, it follows that the sum of the lengths of the upper path
$p,u_0, \dots,u_n,q$ and the lower path $p, l_0, \dots, l_n,q$ is at most 
$\frac{8}{\sqrt{3}}d_x(p,q)$. Therefore, one of the two paths has length
bounded by $\frac{4}{\sqrt{3}}d_x(p,q)$ which proves the Technical Lemma:

\begin{proof}[Proof of (Technical) Lemma~\ref{le:mainlemmaA}]
$d_{T_{1n}}(p,q) \leq \frac{1}{2}{\bar P}({\tt x}(q)) \leq \frac{4}{\sqrt{3}}d_x(p,q)$
by Lemma~\ref{lem:growthrates}.
\end{proof}

In order to prove the stronger bound of Amortization Lemma~\ref{le:mainlemmaB},
in addition to edges of the upper and lower path we will need to include edges
$(l_i,u_i)$ in the analysis and obtain a tighter bound on $P(x)$.
For this reason, we complete the growth rate analysis of
functions $U(x)$, $L(x)$, and $P(x)$. We show that they are not necessarily 
continuous at $x = {\tt x}(c_i)$ but, when discontinuous, they do not increase:
\begin{lemma} 
\label{lem:discontinuity}
Functions $U(x)$, $L(x)$, and $P(x)$, when discontuous at $x={\tt x}(c_i)$
for some $i = 1, \dots, n$, do not increase at $x$. 
\end{lemma}
\begin{proof}
First note that either 
$u_i \not= u_{i-1}$ and $u({\tt x}(c_{i-1})) = u_{i-1}$, or
$u_i = u_{i-1}$ and $i > 1$ and $u({\tt x}(c_{i-1})) = u_{i-1}$, or 
$u_i = u_{i-1}$ and $i = 1$ and $u({\tt x}(c_{i-1})) = p$.  

In the second case, we note that $l_{i-1} \not= l_i$ and that point $l_i$
cannot be on the shortest path from $p$ to $u_i$ in $T_{1i}$. Therefore 
$U({\tt x}(c_i))=d_{T_{1i}}(p,u_i) + p_N(u_i, {\tt x}(c_i)) = d_{T_{1{i-1}}}(p,u_{i-1}) + p_N(u_{i-1}, {\tt x}(c_i))$. The last term is the limit for $U(x)$ when $x \rightarrow {\tt x}(c_i)$ from
the left so $U(x)$ is continuous around $x = {\tt x}(c_i)$ which completes the 
proof for this case.

For the first and third cases we set $u^* =  u({\tt x}(c_{i-1}))$. Then
\begin{align*} 
U({\tt x}(c_i)) & = d_{T_{1i}}(p,u_i) + p_N(u_i, {\tt x}(c_i)) \\ 
       & \leq d_{T_{1(i-1)}}(p,u^*) + d_2(u^*,u_i) + p_N(u_i, {\tt x}(c_i)) \\ 
       & \leq d_{T_{1(i-1)}}(p,u^*) + p_N(u^*,{\tt x}(c_i)) - p_N(u_i,{\tt x}(c_i)) + p_N(u_i, {\tt x}(c_i)) \\ \nonumber
       & = d_{T_{1(i-1)}}(p,u^*) + p_N(u^*, {\tt x}(c_i)) \nonumber
\end{align*}
The last term is the limit for $U(x)$ when $x \rightarrow {\tt x}(c_i)$
from the left and so the claim holds for $U(x)$.
The claim for $L(x)$ holds using equivalent arguments, and the claim for 
$P(x)$ follows from $P(x) = L(x) + U(x)$.
\end{proof}

Before we end this section, we note that we have not defined $t(x)$ at values
of $x$ when $\ell(x)$ or $u(x)$
is a vertex of $H(x)$, which is when $p_S(x)$ or $p_N(x)$, and thus
$\bar{L}(x)$ or $\bar{U}(x)$, respectively, $L(x)$ or $U(x)$, respectively, 
and $P(x)$ are not smooth
and differentiable. In what follows, for clarity of presentation we will
sometimes abuse our definition of $t(x)$ to include such points. Since there
are only a finite number of such points, they do not affect our analysis.

}

\section{Proof of (Amortization) Lemma~\ref{le:mainlemmaB}}
\label{sec:proofB}
\iftoggle{abstract}
{The proof of the lemma builds on the framework discussed in the previous
section and on a careful analysis of the growth rates of $p_N$ and $p_S$
when $T_{1n}$ contains no gentle path.  We show that in that case
the average growth rate of $U+L$ is at most $2\left(\cT\right)$. 

Our main approach is to spread (i.e., amortize) the 
``extra'' $\frac{2}{\sqrt{3}}$ of the $\frac{8}{\sqrt{3}}$ growth rate
over wider intervals of time that, as we show, include time intervals during
which the growth rate is smaller. To achieve our tight bound of 
$2\left(\cT\right)$, however, we need to do more and also include 
``cross-edges'' $(l_i,u_i)$ as illustrated in Fig.~\ref{fig:pathcost}.

\begin{figure}[h]
\center{\anotherpathcost}

\caption{Illustrated is a situation in which the growth rate of $p_S(x)$ is 
$\frac{6}{\sqrt{3}}$ between $x=x_l$ and $x=x_r$. In that case 
the growth rate of $p_N(x)$ is $\frac{2}{\sqrt{3}}$. For large enough such
intervals $[x_l,x_r]$,  
the path $l_i,u_i,u_{i+1},\dots,u_j,l_j$ is a shortcut for 
$l_i, l_{i+1}, \dots,l_j$ and therefore $L(x_r)$ is smaller than what the growth
rate of $p_S(x)$ would indicate. The stretch factor bound
we obtain is the result of a min-max optimization between the two subpaths 
from $l_i$ to $l_j$, and it is tight as we show in Section~\ref{sec:conclusion}.
}
\label{fig:pathcost}
\end{figure}


}
{\begingroup
\def\thetheorem{\ref{le:mainlemmaB}}
\begin{lemma}[The Amortization Lemma]
Let $T_{1n}$ be a regular linear sequence with respect to line $st$ with slope
$m_{st}$. If $0 < m_{st} < \frac{1}{\sqrt{3}}$ and if $T_{1n}$ contains no
gentle path then there is a path in $T_{1n}$ from 
the left induction vertex $p$ of $T_1$ to the right induction vertex $q$ of
$T_n$ of length at most  $\left(\cT\right) d_x(p,q)$.
\end{lemma}
\addtocounter{theorem}{-1}
\endgroup

The proof of the lemma builds on the framework developed in the previous
section and on a careful analysis of the growth rates shown in
Lemma~\ref{lem:growthrates} when $T_{1n}$ contains no gentle path. 
We show that in that case
the average growth rate of $P(x)=U(x)+L(x)$ is at most $2\left(\cT\right)$.

As Lemma~\ref{lem:growthrates} demonstrates, the growth rate of $P(x)$ at a
particular time $x$ (we find it useful to use the time intuition and think of
$x$ as time going from time ${\tt x}(p)$ to time ${\tt x}(q)$) 
is one of $\frac{4}{\sqrt{3}}$, $\frac{6}{\sqrt{3}}$, or
$\frac{8}{\sqrt{3}}$. We will refer to transitions 
$t_{wn_e},t_{s_we},t_{s_ew},t_{en_w}$---transitions for which
the growth rate $\frac{\Delta P(x)}{\Delta x}$ is $\frac{8}{\sqrt{3}}$---as
{\em bad}. Note that when $t(x)$ is not bad, the growth rate of $P(x)$ is at
most $\frac{6}{\sqrt{3}}$, well under the desired growth rate of $2\left(\cT\right)$. 
Our goal is to, when possible, spread (i.e., amortize) the 
``extra'' $\frac{2}{\sqrt{3}}$ of the growth rate
of $P(x)$ when $t(x)$ is bad over wider intervals of time. 
To do this we define intervals of time
during which a particular bad transition takes place 
(see also Fig.~\ref{fig:badintervals}):


\begin{definition}
Given $x_l, x_r$ such that ${\tt x}(p) \leq x_l < x_r \leq {\tt x}(q)$, the interval
$[x_l,x_r]$ is a 
\begin{itemize}
\item $t_{wn_e}$-interval if $t(x_l)= t_{wn_e}$ and $t(x) \not= t_{\ast w}$ for all
$x_l < x < x_r$ and a strict $t_{wn_e}$-interval if, in addition,
$t(x) \not= t_{en_w}, t_{\ast e}$ when $x_l < x < x_r$.
\item $t_{s_ew}$-interval if $t(x_l)= t_{s_ew}$ and $t(x) \not= t_{w \ast}$ for all
$x_l < x < x_r$ and a strict $t_{s_ew}$-interval if, in addition,
$t(x) \not= t_{s_we}, t_{e \ast}$ when $x_l < x < x_r$.
\item $t_{s_we}$-interval if $t(x_r)= t_{s_we}$ and $t(x) \not= t_{e \ast}$ for all
$x_l < x < x_r$ and a strict $t_{s_we}$-interval if, in addition,
$t(x) \not= t_{s_ew}, t_{w \ast}$ when $x_l < x < x_r$
\item $t_{en_w}$-interval if $t(x_r)= t_{en_w}$ and $t(x) \not= t_{\ast e}$ for all
$x_l < x < x_r$ and a strict $t_{en_w}$-interval if, in addition,
$t(x) \not= t_{wn_e}, t_{\ast w}$ when $x_l < x < x_r$
\end{itemize}
\end{definition}

\begin{figure}[!b]
\begin{center}
\badintervals
\end{center}

\caption{Illustrations of (left) $t_{wn_e}$- and $t_{s_ew}$-intervals and (right)
$t_{s_we}$- and $t_{en_w}$-intervals. In a $t_{wn_e}$-interval $[x_l,x_r]$, for
example, $t(x) \not= t_{\ast w}$; if the interval is strict then 
$t(x) \not= t_{en_w}, t_{\ast e}$ as well.}
\label{fig:badintervals}
\end{figure}

Note that if $t(x) = t_{ij}$ is bad for some $x$ in $[{\tt x}(p), {\tt x}(q)]$
then a strict $t_{ij}$-interval contains $x$.
We refer to $t_{wn_e}$- and $t_{s_ew}$-intervals as {\em left intervals} and to
$t_{s_we}$- and $t_{en_w}$-intervals as {\em right intervals}. Note that left
intervals do not intersect each other and neither do the right intervals.
A maximal strict left interval can intersect at most one maximal strict right
interval and vice versa and when that is the case the left interval is to the
left of the right interval.
We will take advantage of the limited interaction between maximal
strict intervals 
to amortize the bad growth taking place within a strict interval over
intervals of time
that do not overlap or, if they do, overlap in a restricted way.

We introduce notation that will help us keep track of the relative horizontal
positions of $u(x)$ and $\ell(x)$: the {\em forward} abcsissa 
$f(x) = \max\{{\tt x}(\ell(x)), {\tt x}(u(x))\}$ and the {\em back} abcsissa
$b(x) = \min\{{\tt x}(\ell(x)), {\tt x}(u(x))\}$. The following facts will
be used heavily without reference throughout this section:
\begin{proposition}
\label{lem:bandf}
Given the assumptions of Lemma~\ref{le:mainlemmaB}, for $i = 1,\dots,n$:
\renewcommand{\labelenumi}{(\alph{enumi})}
\begin{enumerate}
\item ${\tt x}(l_{i-1}) \leq {\tt x}(l_i)$ and ${\tt x}(u_{i-1}) \leq {\tt x}(u_i)$.
\end{enumerate}
and for ${\tt x}(p) \leq x_l < x_r \leq {\tt x}(q)$:
\begin{enumerate}
\addtocounter{enumi}{1}
\item $u(x_l) \leq u(x_r)$ and $\ell(x_l) \leq \ell(x_r)$.
\item $w(x_l) \leq w(x_r)$ and $e(x_l) \leq e(x_r)$.
\item $b(x_l) \leq b(x_r)$ and $f(x_l) \leq f(x_r)$.
\end{enumerate}
\end{proposition}

\begin{proof}
Note that if $l_{i-1}=l_i$ is the base vertex of $T_i$ then $u_{i-1}, u_i$ lie 
on sides $w, n_w$ or $w, n_e$ or $n_w, n_e$ or $n_w,e$ or $n_e, e$ of $H_i$; 
if $u_{i-1}=u_i$ is the base vertex of $T_i$ then $l_{i-1}, l_i$ lie on sides
$w, s_w$ or $w, s_e$ or $s_w, s_e$ or $s_w,e$ or $s_e, e$ of $H_i$.
Part {\it (b)} follows from this and parts {\it (c)}
and {\it (e)} follow from  {\it (b)}. Part {\it (d)} follows from 
Lemma~\ref{lem:growthrates}.
\end{proof}

%

Finally, we define key functions $\delta_f(x) = e(x) - f(x)$ and 
$\delta_b(x) = b(x) - w(x)$ that will be used when $t(x)$ is bad.
Note that when $t(x)= t_{wn_e},t_{s_ew}$ we have 
$\delta_f(x) \leq {\tt r}(x)$ and when $t(x)=t_{s_we},t_{en_w}$ 
then $\delta_b(x) \leq {\tt r}(x)$. 

We now prove several technical lemmas that we need. Unlike our bounds for
${\bar U}(x)$ and ${\bar L}(x)$ that only made use of edges of the upper and
lower path, we will sometimes make use of {\em cross-edges} $(u_i,l_i)$ in 
computing tighter bounds for $U(x)$ or $L(x)$, and therefore $P(x)$. We will,
in particular, use edges  $(\ell(x),u(x))$ for values of $x$ when $t(x)$ is bad.
To motivate the next lemma which bounds the cost of using such edges, consider,
for example, a  $t_{wn_e}$-interval starting at $x_l$ and the edge 
$(\ell(x_l), u(x_l)) = (l_{i-1}=l_i, u_i)$ between the $w$ and $n_e$ sides of 
$H(x)$. Then
\begin{align*} 
U(x_l) & = d_{T_{1i}}(p,u_i) + p_N(u_i, x_l) \nonumber \\
       & \leq \min \begin{cases}
                 d_{T_{1(i-1)}}(p,u_{i-1}) + d_2(u_{i-1},u_i) + p_N(u_i, x_l) \\ 
                 d_{T_{1(i-1)}}(p,l_{i-1}) + d_2(l_{i-1},u_i) + p_N(u_i, x_l) \\
                   \end{cases} \\ \nonumber
       & \leq L(x_l) - p_S(l_{i-1},x_l) + d_2(l_{i-1},u_i) + p_N(u_i, x_l)
\end{align*}

The term $- p_S(l_{i-1},x_l) + + d_2(l_{i-1},u_i) + p_N(u_i, x_l)$ can be seen
as the cost of using edge $(l_{i-1},u_i)$. The following lemma provides a bound
on this cost in terms of $\delta_f(x)$ as well as the cost of using other 
{\em bad cross-edges} (see also Fig.~\ref{fig:gentlepath}-(a)):
\begin{lemma}
\label{lem:switchpath}
Given the assumptions of Lemma~\ref{le:mainlemmaB}, let $t(x)$ be bad for some
$x \in [{\tt x}(p), {\tt x}(q)]$.
%
\begin{itemize}
\item If $t(x) = t_{wn_e}$ then
$-p_S(x) + d_2(\ell(x),u(x)) + p_N(x) \leq \left(2 - \frac{2}{\sqrt{3}}\right)\delta_f(x)$
\item If $t(x) = t_{s_ew}$ then
$p_S(x) + d_2(\ell(x),u(x)) - p_N(x) \leq \left(2 - \frac{2}{\sqrt{3}}\right)\delta_f(x)$
\item If $t(x) = t_{s_we}$ then
$p_S(x) + d_2(\ell(x),u(x)) - p_N(x) \leq \left(2 - \frac{2}{\sqrt{3}}\right)\delta_b(x)$
\item If $t(x) = t_{en_w}$ then
$-p_S(x) + d_2(\ell(x),u(x)) + p_N(x) \leq \left(2 - \frac{2}{\sqrt{3}}\right)\delta_b(x)$
\end{itemize}
\end{lemma}
\begin{proof}
Assuming that $t(x) = t_{wn_e}$ (the other cases follow
by symmetry),
%
%
we let $w$ to be the SW corner of $H(x)$ and let $\theta = \angle u_i w N(x)$
as illustrated in Fig.~\ref{fig:gentlepath}-(a). We then
obtain
$p_S(x)=p_S(\ell(x),x) \geq \frac{2}{\sqrt{3}}{\tt r}(x)$,
$d_2(\ell(x), u(x)) \leq \frac{2}{\cos(\theta)}{\tt r}(x)$, and
$p_N(x)=p_N(u(x),x) = -2\tan(\theta){\tt r}(x)$.
Noting that
$0 \leq \theta \leq \frac{\pi}{6}$ and
	using
$\frac{1}{\cos(\theta)} \leq 1 + (2-\sqrt{3})\tan(\theta)$ when
$\theta \in [0,\frac{\pi}{6}]$, 
we have
\begin{align*}
-p_S(x) + d_2(\ell(x),u(x)) + p_N(x) & \leq \left(-\frac{2}{\sqrt{3}} + \frac{2}{\cos(\theta)} -2\tan(\theta)\right){\tt r}(x)\\
& \leq \left(-\frac{2}{\sqrt{3}} + 2 + 2(2-\sqrt{3})\tan(\theta) -2\tan(\theta)\right){\tt r}(x) \\
& = \left(2-\frac{2}{\sqrt{3}}\right)\left(1-\sqrt{3}\tan(\theta)\right){\tt r}(x)
\end{align*}
Since ${\tt r}(x) - \delta_f(x) = \sqrt{3}\tan(\theta){\tt r}(x)$, the proof
follows from the last inequality.
\end{proof}

\iftoggle{abstract}
{\begin{figure}[!b]}
{\begin{figure}}
\gentlepath

\hspace{1.9cm}(a) \hspace{2.6in} (b)

\caption{(a) Illustration and proof of Lemma~\ref{lem:switchpath} in the case
of $t(x) = t_{wn_e}$. (b) 
Lemma~\ref{lem:gentlepath} in the case of a $t_{wn_e}$-interval 
$[x_l,x_r]$. The dotted curve from $N(x_l)$ to $N(x_r)$ consists of a
sequence of piecewise linear curves each representing the growth rate
$\Delta p_N(x)$ over an interval $\Delta x$ during which $t(x)$ is a fixed
transition. A few examples of such piecewise linear curves are shown in bold in
Fig.~\ref{fig:growth}. The dotted curve has total length 
$\bar{U}(x_r) - \bar{U}(x_l)$ and its slope is $-\frac{1}{\sqrt{3}}$ exactly
when $t(x)$ is either $t_{\ast n_e}$ (which includes bad transition $t_{wn_e}$) 
or $t_{\ast e}$. Lemma~\ref{lem:gentlepath} states that the time (shown with
thick blue segments) spent in those transitions within interval $[x_l, x_r]$
is bounded by ${\tt r}(x_l) + (\sqrt{3} - 1)\delta_f(x_l)$, if
$T_{1n}$ contains no gentle path.}
\label{fig:gentlepath}
\end{figure}

The following important lemma bounds the amount of time
certain transitions, and especially bad transition $t_{ij}$, can take place 
within a
$t_{ij}$-interval $[x_l,x_r]$ if $T_{1n}$ contains no gentle path. By {\it
``amount of time a transition $t_{ij}$ takes place in an interval $[x_l,x_r]$''}
we mean the Lebesgue measure of the set 
$\{x \in [x_l,x_r]: t(x) = t_{ij}\}$
which is the union of a finite number of disjoint open intervals inside
$[x_l,x_r]$. The lemma we state below captures 
the following insight. When $t(x)$ is a bad transition
and the growth rate of $P(x)$ is $\frac{8}{\sqrt{3}}$ then
by Lemma~\ref{lem:growthrates} one of $U(x)$ or $L(x)$ has growth rate 
$\frac{2}{\sqrt{3}}$. More specifically, if, say, $t(x) = t_{wn_e}$ then 
$\frac{\Delta U(x)}{\Delta x} = \frac{2}{\sqrt{3}}$, meaning that the
upper path fragment $u(x_l)=u_i, u_{i+1}, \dots, u(x_r)$ has a
``very gentle'' average slope. Therefore if $t(x) = t_{wn_e}$ for a sufficiently
long enough time within $t_{wn_e}$-interval $[x_l,x_r]$ then the path 
$\ell(x_l), u(x_l), \dots, u(x_r)$ would have to be a gentle path.
We actually need the contrapositive of this insight (see also Fig.~\ref{fig:gentlepath}-(b)):

\begin{lemma}
\label{lem:gentlepath}
Given the assumptions of Lemma~\ref{le:mainlemmaB},
let ${\tt x}(p) \leq x_l \leq x_r \leq {\tt x}(q)$,
and let $z_{ij}$ be the amount of time within interval $[x_l,x_r]$ spent
in transition $t_{ij}$. If $[x_l,x_r]$ is
%
%
\begin{itemize}
\item a $t_{wn_e}$-interval then
$z_{\ast n_e} + z_{\ast e} \leq {\tt r}(x_l) + (\sqrt{3} - 1)\delta_f(x_l)$
\item a $t_{s_ew}$-interval then
$z_{s_e \ast} + z_{e \ast} \leq {\tt r}(x_l) + (\sqrt{3} - 1)\delta_f(x_l)$
\item a $t_{s_we}$-interval then
$z_{w \ast} + z_{s_w \ast} \leq {\tt r}(x_r) + (\sqrt{3} - 1)\delta_b(x_r)$
\item a $t_{en_w}$-interval then
$z_{\ast n_w} + z_{\ast w} \leq {\tt r}(x_r) + (\sqrt{3} - 1)\delta_b(x_r)$
\end{itemize}
where wildcard notations $z_{\ast j}$ and $z_{i \ast}$ refer to the amount
of time spent in transitions $t_{\ast j}$ and $t_{i \ast}$, respectively.
\end{lemma}

\begin{proof}
We assume that $[x_l,x_r]$ is a $t_{wn_e}$-interval (the other cases follow
by symmetry). We also assume that $t(x_r) \not= t_{\ast n_w}$ implying that
$u(x_r)$ lies either on the N vertex,
side $n_e$, the NE vertex, or side $e$ of $H(x_r)$ (or, with some abuse of
the definition of $t(x)$, $t(x_r) = t_{\ast n_e},t_{\ast e})$.  If that is not the
case, we simply consider the
shorter interval $[x_l,x_q]$ where $x_q$ is such that $u(x_{q})$ is the N
vertex of $H(x_q)$ and $t(x) = t_{\ast n_w}$ for $x_q < x < x_r$.

Consider the path in $T_{1n}$ from $\ell(x_l)$ to $u(x_r)$ that starts with
the edge $(\ell(x_l), u(x_l)) = (l_i,u_i)$ and then visits the vertices
$u_i, u_{i+1}, \dots, u_{i+s} = u(x_r)$ in order
(refer to Fig.~\ref{fig:gentlepath}-(b)). 
Since $T_{1n}$ contains no gentle paths, the length of this path is at least
$\sqrt{3}d_x(l_i,u_{i+s}) - ({\tt y}(u_{i+s})-{\tt y}(l_i))$ and so:
\begin{align}
\label{eq:nogentle}
d_2(l_i,u_i)+ \sum_{t=i}^{i+s-1} d_2(u_t,u_{t+1}) + ({\tt y}(u_{i+s})-{\tt y}(l_i)) & \geq \sqrt{3}d_x(l_i,u_{i+s}) 
\end{align}
%
%
Note that ${\tt y}(N(x_l)) - {\tt y}(l_i) + p_S(l_i,x_l) = \frac{5}{\sqrt{3}} {\tt r}(x_l)$ and so, by
Lemma~\ref{lem:growthrates},
\begin{align}
\label{eq:crossvertical}
{\tt y}(u_{i+s})-{\tt y}(l_i) & = ({\tt y}(N(x_l)) - {\tt y}(l_i)) + ({\tt y}(N(x_r)) - {\tt y}(N(x_l))) 
- ({\tt y}(N(x_r) - {\tt y}(u_{i+s}))  \notag \\
		   & \leq \frac{5}{\sqrt{3}} {\tt r}(x_l) - p_S(l_i,x_l) + 
                          \left(\frac{1}{\sqrt{3}}z_{\ast n_w} - \frac{1}{\sqrt{3}}z_{\ast n_e} -
                        \frac{3}{\sqrt{3}}z_{s_ee} - \frac{5}{\sqrt{3}}z_{s_we}\right) \notag \\ 
                   & \quad     - \left(\frac{1}{\sqrt{3}}({\tt x}(u_{i+s})-x_r) + \max\{0,{\tt y}(v) - {\tt y}(u_{i+s})\}\right)
\end{align}
where $v$ is the NE vertex of $H(x_r)$. The $\max$ term in (\ref{eq:crossvertical}) is 0 or
positive depending on whether $u_{i+s}$ is on
side $n_e$ or $e$, respectively, of $H(x_r)$. 
Each edge $(u_t,u_{t+1})$ on the path $u_i, u_{i+1}, \dots, u_{i+s}$
can be bounded by $p_N(u_t, {\tt x}(c_{t+1})) - p_N(u_{t+1},{\tt x}(c_{t+1}))$ and so:
\begin{align}
\label{eq:upperpath}
\sum_{t=i}^{i+s-1} d_2(u_t,u_{t+1}) & \leq \sum_{t=i}^{i+s-1} (p_N(u_t, {\tt x}(c_{t+1})) - p_N(u_{t+1},{\tt x}(c_{t+1}))) \notag \\
	& = (\bar{U}(x_r) - p_N(u_{i+s},x_r)) - (\bar{U}(x_l) - p_N(u_i,x_l)) \notag \\
	& = p_N(u_i,x_l) + (\bar{U}(x_r)- \bar{U}(x_l)) - p_N(u_{i+s},x_r)  \notag \\
	& = p_N(u_i,x_l) + \left(\frac{2}{\sqrt{3}}(z_{\ast n_w} + z_{\ast n_e}) + \frac{4}{\sqrt{3}}z_{s_ee} + \frac{6}{\sqrt{3}}z_{s_we}\right) \notag \\
   & \quad  + \left(\frac{2}{\sqrt{3}}({\tt x}(u_{i+s})-x_r) + \max\{0,{\tt y}(v) - {\tt y}(u_{i+s})\}\right)
\end{align}
Substituting the left-hand side of
(\ref{eq:nogentle}) with  (\ref{eq:crossvertical}) and 
(\ref{eq:upperpath}) gives us 
\begin{multline*}
d_2(l_i,u_i) + \frac{5}{\sqrt{3}} {\tt r}(x_l) - p_S(l_i,x_l) + p_N(u_i,x_l) +
\frac{3}{\sqrt{3}}z_{\ast n_w}+\frac{1}{\sqrt{3}}(z_{\ast n_e}+z_{\ast e}+{\tt x}(u_{i+s})-x_r)\\
\geq \frac{3}{\sqrt{3}}({\tt x}(u_{i+s})-{\tt x}(l_i)) 
= \frac{3}{\sqrt{3}}({\tt x}(u_{i+s})- x_l) + \frac{3}{\sqrt{3}}{\tt r}(x_l) 
\end{multline*}
where we use ${\tt x}(l_i) = x_l - {\tt r}(x_l)$. Using
$z_{\ast n_w} + z_{\ast e} + z_{\ast n_e} = x_r - x_l$ and ${\tt x}(u_{i+s}) \geq x_r$, and
applying Lemma~\ref{lem:switchpath},
we get 
\[
\frac{2}{\sqrt{3}} {\tt r}(x_l) + \left(2 - \frac{2}{\sqrt{3}}\right)\delta_f(x_l) \geq
\frac{2}{\sqrt{3}}(z_{\ast n_e}+z_{\ast e})
\]
which is equivalent to the lemma statement.
%
\end{proof}

The following two lemmas build on Lemmma~\ref{lem:gentlepath}. They show
different ways that the ``extra'' $\frac{2}{\sqrt{3}}$ growth (i.e., the
growth above $\frac{6}{\sqrt{3}}$) of a bad transition $t_{ij}$ within a
$t_{ij}$-interval can be spread (amortized) over a wider time interval (refer
also to Fig.~\ref{fig:linearworstcase}).

\begin{lemma}  
\label{lem:linearworstcase}
Given the assumptions of Lemma~\ref{lem:gentlepath}, if $[x_l,x_r]$ is a 
$t_{wn_e}$-, $t_{s_ew}$-, $t_{en_w}$-, or $t_{s_we}$-interval, if
$t(x) \not= t_{en_w}$, $t_{s_we}$, $t_{wn_e}$, or $t_{s_ew}$, respectively, when
$x_l \leq x \leq x_r$, and if 
$t(x_r)=t_{\ast e}$, $t_{e \ast}$, $t_{\ast w}$, or $t_{w \ast}$, respectively, then
%
%
\begin{align}
\label{eq:LWC}
\frac{2}{\sqrt{3}}z_{wn_e}, \frac{2}{\sqrt{3}}z_{s_ew} & \leq \left(\frac{2}{\sqrt{3}}-1\right)(e(x_r)-e(x_l)+2\delta_f(x_l)) \\
\label{eq:LWC2}
& \leq \left(\frac{4}{\sqrt{3}}-2\right)(x_r-x_l+\delta_f(x_l)) \\
\frac{2}{\sqrt{3}}z_{s_we}, \frac{2}{\sqrt{3}}z_{en_w} & 
\leq \left(\frac{2}{\sqrt{3}}-1\right)(w(x_r)-w(x_l)+2\delta_b(x_r)) \nonumber \\ 
&  \leq \left(\frac{4}{\sqrt{3}}-2\right)(x_r-x_l+\delta_b(x_r)), \nonumber 
\end{align}
respectively.
\end{lemma}

\begin{figure}[!b]
\linearworstcase

\caption{Illustration of Lemmas~\ref{lem:linearworstcase}
and~\ref{lem:softrecovery} for the case of a
$t_{wn_e}$-interval $[x_l,x_r]$. (a) Lemma~\ref{lem:linearworstcase}: the
$\frac{2}{\sqrt{3}}$ extra cost of the bad transition $t_{wn_e}$ taking place
within the interval $[x_l,x_r]$ can be amortized in one of two ways:
either as cost $\frac{1}{2}(\frac{4}{\sqrt{3}}-2)$ spread over the interval
$[e(x_l)-2\delta_f(x_l),e(x_r)]$ (shown in blue) or as cost
$\frac{4}{\sqrt{3}}-2$ spread over the interval $[x_l-\delta_f(x_l),x_r]$
(shown in red).
(b) 
Lemma~\ref{lem:softrecovery}: the $\frac{2}{\sqrt{3}}$ extra cost of the
bad transition $t_{wn_e}$ within the interval $[x_l,x_r]$ can be amortized as
cost $\frac{1}{2}(\frac{4}{\sqrt{3}}-2)$ spread over the interval
$[e(x_l)-2\delta_f(x_l),w(x_r)]$ (shown in blue) {\em plus} cost
$\frac{1}{\sqrt{3}}$ spread over intervals of time $z_{s_wn_w}$ and $z_{s_en_w}$
(contained within the dashed red interval).}
\label{fig:linearworstcase}
\end{figure}

\begin{proof}
We prove the lemma for the case of a $t_{wn_e}$-interval; the other 3 cases
follow by symmetry. By Lemma~\ref{lem:growthrates}, if
$[x_l,x_r]$ is a $t_{wn_e}$-interval satisfying the conditions of the lemma then
${\tt r}(x_l)+z_{wn_e}+z_{wn_w}+\frac{1}{2}z_{s_wn_w} = {\tt r}(x_r) + \frac{1}{2}z_{s_en_e}+z_{en_e}+z_{s_we} + z_{s_ee}$
which implies
\begin{equation}
\label{eq:radii}
\sqrt{3} z_{wn_e} - \frac{\sqrt{3}}{2}z_{s_en_e} \leq \sqrt{3}({\tt r}(x_r) - {\tt r}(x_l) + z_{en_e}+z_{s_we} + z_{s_ee})
\end{equation}
Since $t(x_r) = t_{\ast e}$, we have 
$t(x) \not= t_{\ast w}, t_{\ast n_w}$ for $x_r \leq x \leq e(x_r)$, and thus, the interval
$[x_l,e(x_r)]$ is a $t_{wn_e}$-interval such that $t(x) = t_{\ast n_e}$ or
$t(x) = t_{\ast e}$ when $x_r \leq x \leq e(x_r)$.  Therefore, 
by Lemma~\ref{lem:gentlepath} we have 
%
%
${\tt r}(x_r) + z_{\ast n_e} + z_{\ast e} \leq {\tt r}(x_l) + (\sqrt{3}-1)\delta_f(x_l)$
which implies
\[
z_{wn_e} + z_{s_en_e} \leq (\sqrt{3}-1)\delta_f(x_l) - ({\tt r}(x_r) - {\tt r}(x_l) + z_{en_e}+z_{s_we} + z_{s_ee})
	\]
together with inequality~(\ref{eq:radii}), we get
\[(\sqrt{3}+1)z_{wn_e} \leq (\sqrt{3}-1)({\tt r}(x_r) - {\tt r}(x_l) +z_{en_e}+z_{s_we}+z_{s_ee}+\delta_f(x_l))\]
and multiplying both sides by $(\sqrt{3} - 1)/\sqrt{3}$
\begin{equation}
\label{eq:radii2}
\frac{2}{\sqrt{3}}z_{wn_e} \leq \left(\frac{4}{\sqrt{3}}-2\right)({\tt r}(x_r) - {\tt r}(x_l) +z_{en_e}+z_{s_we}+z_{s_ee}+\delta_f(x_l))
\end{equation}
%
%
By Lemma~\ref{lem:growthrates}, we have
$z_{en_e}+z_{s_we}+z_{s_ee} \leq \frac{1}{2}(w(x_r) - w(x_l))$ and 
so~(\ref{eq:radii2}) implies inequality~(\ref{eq:LWC}).
Because $0 \leq w(x_r) - w(x_l)$ inequality~(\ref{eq:LWC}) implies
inequality~(\ref{eq:LWC2}). 
\end{proof}

\begin{lemma}
\label{lem:softrecovery}
(Refer to Fig.~\ref{fig:linearworstcase}-(b)) 
Given the assumptions of Lemma~\ref{lem:gentlepath}, 
if $[x_l,x_r]$ is a $t_{wn_e}$-, $t_{s_ew}$-, $t_{en_w}$-, or $t_{s_we}$-interval,
if $t(x) \not= t_{en_w}$, $t_{s_we}$, $t_{wn_e}$, or $t_{s_ew}$, respectively, when
$x_l \leq x \leq x_r$, and if $t(x_r) = t_{\ast w}$, $t_{w \ast}$, $t_{\ast e}$,
or $t_{e \ast}$, respectively, then
\begin{align}
& \quad \frac{2}{\sqrt{3}}z_{wn_e} - \frac{1}{\sqrt{3}}(z_{s_wn_w}+z_{s_en_w}), 
\frac{2}{\sqrt{3}}z_{s_ew} - \frac{1}{\sqrt{3}}(z_{s_wn_w}+z_{s_wn_e})  \nonumber \\
\leq & \quad \left(\frac{2}{\sqrt{3}}-1\right)(w(x_r) - e(x_l) + 2\delta_f(x_l)) \label{eq:softrecovery} \\
& \quad \frac{2}{\sqrt{3}}z_{en_w} - \frac{1}{\sqrt{3}}(z_{s_wn_e}+z_{s_en_e}),
\frac{2}{\sqrt{3}}z_{s_we} - \frac{1}{\sqrt{3}}(z_{s_en_w}+z_{s_en_e}) \nonumber \\
\leq & \quad \left(\frac{2}{\sqrt{3}}-1\right)(w(x_r) - e(x_l) + 2\delta_b(x_r)), \nonumber
\end{align}
respectively.
\end{lemma}

\begin{proof}
We prove the lemma for the case of a $t_{wn_e}$-interval only; the other 3 cases
can be seen to follow by symmetry. Then, 
since $[x_l,x_r]$ is a $t_{wn_e}$-interval and $t(x_r)=t_{\ast w}$, point $u(x_r)$
must be the NW point of $H(x_r)$ as shown in
Fig.~\ref{fig:linearworstcase}-(b). The point $u(x_r)$ must lie on side $n_w$ of 
$H(x)$ for $w(x_r) \leq x \leq x_r$ and so  $t(x) = t_{\ast n_w}$ when
$w(x_r) \leq x \leq x_r$. This, together with Lemma~\ref{lem:growthrates} on the growth of $w(\cdot)$
and the fact that $t(x) \not= t_{en_w}$ when $w(x_r) \leq x \leq x_r$,
implies that $w(x_r) - w(w(x_r)) = {\tt r}(w(x_r)) \leq \frac{1}{2}z_{s_wn_w}+z_{s_en_w} \leq z_{s_wn_w}+z_{s_en_w}$.
It also implies that all $t_{\ast n_e}$ and $t_{\ast e}$ transitions
in interval $[x_l,x_r]$ take place within interval $[x_l,w(x_r)]$ and so by
Lemma~\ref{lem:growthrates} on the growth of $e(\cdot)$,  $2z_{wn_e} \leq e(w(x_r)) - e(x_l)$.
Therefore, $2z_{wn_e}-(z_{s_wn_w}+z_{s_en_w}) \leq e(w(x_r))-{\tt r}(w(x_r))-e(x_l) = w(x_r)-e(x_l)$ and:
\begin{equation}
\label{eq:softrecovery2}
\left(\frac{4}{\sqrt{3}} - 2\right)z_{wn_e}-\left(\frac{2}{\sqrt{3}} - 1\right)(z_{s_wn_w}+z_{s_en_w}) \leq\
\left(\frac{2}{\sqrt{3}}-1\right)(w(x_r)-e(x_l))
\end{equation}

By Lemma~\ref{lem:gentlepath} we have
$z_{\ast e}+z_{\ast n_e} \leq {\tt r}(x_l)+(\sqrt{3}-1)\delta_f(x_l)$ and by
Lemma~\ref{lem:growthrates} on the growth rate of ${\tt r}(\cdot)$, we have 
${\tt r}(x_l)+z_{wn_e}-(z_{s_wn_e}+z_{s_en_e}+z_{en_e}+z_{s_we}+z_{s_ee}) \leq {\tt r}(w(x_r))$.
Using ${\tt r}(w(x_r)) \leq z_{s_wn_w}+z_{s_en_w}$ from above, we get
$2z_{wn_e} - (z_{s_wn_w}+z_{s_en_w}) \leq (\sqrt{3}-1)\delta_f(x_l)$
and multiplying both sides by $1 - \frac{1}{\sqrt{3}}$
\[\left(2-\frac{2}{\sqrt{3}}\right)z_{wn_e} - \left(1-\frac{1}{\sqrt{3}}\right)(z_{s_wn_w}+z_{s_en_w}) \leq
	\left(\frac{4}{\sqrt{3}}-2\right)\delta_f(x_l)\]
Summing~(\ref{eq:softrecovery2}) with this inequality,
we obtain~(\ref{eq:softrecovery}).
\end{proof}

We define a time $x$ in the interval $[{\tt x}(p), {\tt x}(q)]$ to be
{\em left critical} 
or {\em right critical} if $x$ is the left boundary of a maximal strict 
left interval or the right boundary of a maximal strict 
right interval, respectively. We also define the start and end
coordinates ${\tt x}(p)$ and ${\tt x}(q)$ to be both left and right critical. A
time $x$ is said  to be critical if it is left or right critical.
This lemma, illustrated in Fig.~\ref{fig:induction}, 
 bounds $P(x)$ by making use of amortization
Lemmas~\ref{lem:linearworstcase} and~\ref{lem:softrecovery}:

\begin{figure}
\induction

\caption{Illustration of Lemma~\ref{lem:criticalpoints}. (a) If $x$ is a left
critical point then $P(x)$ is no greater than
$(\frac{10}{\sqrt{3}}-2)x - (\frac{4}{\sqrt{3}} - 2)\delta_f(x)$.
In the case of a $t_{wn_e}$-interval starting at $x$, this allows some of the
extra $\frac{2}{\sqrt{3}}$ growth of bad transition $t_{wn_e}$ within the interval
to be charged to the ``past'' as cost $\frac{4}{\sqrt{3}}-2$ over the interval
$[x-\delta_f(x),x]$. 
(b) If $x$ is a right critical point then $P(x)$ is no greater than
$(\frac{10}{\sqrt{3}}-2)x + (\frac{4}{\sqrt{3}} - 2)\delta_b(x)$.
In the case of a $t_{en_w}$-interval ending at $x$, this means that some
of the extra $\frac{2}{\sqrt{3}}$ growth of bad transition $t_{en_w}$ within the
interval has been charged to the ``future'' as cost $\frac{4}{\sqrt{3}}-2$
over the interval $[x, x+\delta_b(x)]$.}
\label{fig:induction}
\end{figure}

\begin{lemma}
\label{lem:criticalpoints}
Given the assumptions of Lemma~\ref{le:mainlemmaB}:
\begin{itemize}
\item If $x$ is a left critical point then $P(x) \leq \left(\frac{10}{\sqrt{3}}-2\right)x - \left(\frac{4}{\sqrt{3}}-2\right)\delta_f(x)$
\item If $x$ is a right critical point then $P(x) \leq \left(\frac{10}{\sqrt{3}}-2\right)x + \left(\frac{4}{\sqrt{3}}-2\right)\delta_b(x)$
\end{itemize}
\end{lemma}

\begin{proof}
We proceed by induction, using  the left to right ordering of critical points.
The first critical point (i.e., the base case) is ${\tt x}(p)$. Since
$P(x) = 0 = {\tt x}(p) = \delta_f({\tt x}(p)) = \delta_b({\tt x}(p))$,
the claim holds. We assume now that the claim holds
for critical point $x_l$ and prove that it holds for the next critical point
$x_r$ (and so $x_l < x_r$ and no critical point exists between $x_l$ and $x_r$).

\noindent{\bf Case LL:} We first consider the case when $x_l$ and $x_r$ are both left critical points. We will show that in that case:
\begin{align}
\label{eq:leftleft}
\Delta P(x_l,x_r) \leq \left(\frac{10}{\sqrt{3}}-2\right)(x_r - x_l) - \left(\frac{4}{\sqrt{3}}-2\right)
({\tt r}(x_r) - \delta_f(x_l))
\end{align}
Combining the above inequality with the inductive hypothesis for $P(x_l)$,
we obtain $P(x_r) = P(x_l) + \Delta P(x_l,x_r) \leq (\frac{10}{\sqrt{3}}-2)x_r
- (\frac{4}{\sqrt{3}}-2){\tt r}(x_r) \leq (\frac{10}{\sqrt{3}}-2)x_r
- (\frac{4}{\sqrt{3}}-2)\delta_f(x_r)$, as illustrated in Fig.~\ref{fig:proofA},
and complete the proof for this case.

We now show that~(\ref{eq:leftleft}) holds. 	
Both $x_l$ and $x_r$ being left critical implies that transitions $t_{s_we}$ and $t_{en_w}$ do not take place
within interval $[x_l, x_r]$. W.l.o.g. we assume that $x_l$ is the left
boundary of a maximal strict $t_{wn_e}$-interval. If $x_q$ is the right boundary
of this interval then $x_q \leq x_r$, $t(x_q) = t_{\ast w}$ or $t(x_q) = t_{s_ee}$,
and $t(x)$ is not bad when $x_q < x < x_r$. Let $z_{ij}$ be the time within
interval $[x_l,x_q]$ spent in transition $t_{ij}$, for every transition $t_{ij}$,
and let $z = \sum z_{ij} = x_q-x_l$.

If $t(x_q) = t_{\ast w}$ then by Lemma~\ref{lem:growthrates} we have:
\begin{align*}
\Delta P(x_l,x_r) & \leq \frac{8}{\sqrt{3}}z_{wn_e} + \frac{4}{\sqrt{3}}(z_{s_wn_w}+z_{s_en_w}) + \frac{6}{\sqrt{3}}(z - z_{wn_e} - z_{s_wn_w} - z_{s_en_w} + x_r - x_q) \\
                  & \leq \frac{6}{\sqrt{3}}(x_r - x_l) + \frac{2}{\sqrt{3}}z_{wn_e} - \frac{2}{\sqrt{3}}(z_{s_wn_w}+z_{s_en_w}) \\
                  & \leq \frac{6}{\sqrt{3}}(x_r - x_l) + \left(\frac{4}{\sqrt{3}}-2\right)(w(x_q)-x_l+\delta_f(x_l))
\end{align*}
The last inequality follows from Lemma~\ref{lem:softrecovery} and from 
$x_l < e(x_l)$ and $w(x_q)-x_l > {\tt x}(u(x_q)) - {\tt x}(u(x_l)) \geq 0$. 
Using $w(x_q) \leq w(x_r) = x_r - {\tt r}(x_r)$, we prove (\ref{eq:leftleft}).
\begin{figure}
\proofA
\caption{The proof of Lemma~\ref{lem:criticalpoints} in the case when $x_l$ and $x_r$ are both left critical and
$x_l$ is the left boundary of a $t_{wn_e}$-interval. By induction,
$P(x_l)$ is amortized as illustrated in Fig.~\ref{fig:induction}-(a) and shown
above in blue. We show that
$\Delta P(x_l,x_r)$ can be amortized as cost $\frac{10}{\sqrt{3}}-2$ over
interval $[x_l,w(x_r)]$ {\em plus} cost
$\frac{6}{\sqrt{3}}$ over interval $[w(x_r), x_r]$, shown in red.
This implies that $P(x_r)$ can be amortized as shown in Fig.~\ref{fig:induction}-(a).}
\label{fig:proofA}
\end{figure}

If $t(x_q) = t_{s_ee}$ then by Lemma~\ref{lem:growthrates} we have:

\begin{align*}
\Delta P(x_l,x_r) & \leq \frac{8}{\sqrt{3}}z_{wn_e} + \frac{6}{\sqrt{3}}(z - z_{wn_e} + x_r - x_q) \leq \frac{6}{\sqrt{3}}(x_r - x_l) + \frac{2}{\sqrt{3}}z_{wn_e} \\
                  & \leq \frac{6}{\sqrt{3}}(x_r - x_l) + \left(\frac{4}{\sqrt{3}}-2\right)(x_q-x_l+\delta_f(x_l))
\end{align*}
The last inequality follows from Lemma~\ref{lem:linearworstcase}.
Now, since $t(x_q) = t_{s_ee}$, then $x_q \leq b(x_q) \leq b(x_r) = w(x_r) = x_r - {\tt r}(x_r)$, we prove (\ref{eq:leftleft}) and complete the proof in this case as well.

\noindent{\bf Case RR:} The case when $x_l$ and $x_r$ are both right critical
points can be handled using an argument that is symmetric to the one we used for
case LL.
 
\noindent{\bf Case RL:} If $x_l$ is right critical and $x_r$ is left critical then no bad transitions
take place within the interval $[x_l,x_r]$ and
$\Delta P(x_l,x_r) \leq \frac{6}{\sqrt{3}}(x_r-x_l)$. If $x_l+\delta_b(x_l) \leq x_r-\delta_f(x_r)$ then by induction:
\begin{align*}
	P(x_r) & = P(x_l) + \Delta P(x_l,x_r) \\
               & \leq \left(\frac{10}{\sqrt{3}}-2\right)x_l+\left(\frac{4}{\sqrt{3}}-2\right)\delta_b(x_l)
	       + \frac{6}{\sqrt{3}}(x_r-x_l) \\
               & \leq \left(\frac{10}{\sqrt{3}}-2\right)x_r-\left(\frac{4}{\sqrt{3}}-2\right)\delta_f(x_r)
\end{align*}

If $x_l+\delta_b(x_l) > x_r-\delta_f(x_r)$, consider the interval
$I = [x_r-\delta_f(x_r), x_l+\delta_b(x_l)]$.
We argue that this interval
is contained within interval $[x_l, x_r]$ and that $\Delta P$ has a growth
rate of just $\frac{4}{\sqrt{3}}$ in interval $I$. To do
this, we assume w.l.o.g. that $t(x_l) = t_{en_w}$ with $u(x_l) = u_s$ and
$\ell(x_l) = l_s$, as illustrated in Fig.~\ref{fig:proofB}-(a). 
Since $x_r$ is left critical, one of $u(x_r)$ or $\ell(x_r)$ must lie on the
$w$ side of $H(x_r)$. However, the point on the $w$ side of $H(x_r)$ and 
having abscissa $w(x_r)$ cannot be a point $\ell(x_r) = l_t$ in $L$ 
(with $t \geq s$) because 
we would have
$w(x_r) = {\tt x}(l_t) \geq {\tt x}(l_s) = e(x_l)$ which contradicts 
$x_l+\delta_b(x_l) > x_r-\delta_f(x_r)$.
Therefore, the point on the $w$ side of $H(x_r)$ must be 
a point $u_t \in U$ for some $t \geq s$ and $t(x) = t_{\ast n_w}$ for $x \in
[{\tt x}(u_t),x_r]$.  Furthermore, letting $x' = \max\{x_l, {\tt x}(u_t)\}$,
$t(x)$ must be $t_{wn_w}$, $t_{s_wn_w}$, or $t_{s_en_w}$ for $x$ in interval
$[x',x_r]$, and by
Lemma~\ref{lem:growthrates}, $\frac{\Delta {\tt r}(x)}{\Delta x} \geq 0$ for
$x$ in that interval. Thus, $x_r = w(x_r) + {\tt r}(x_r) \geq {\tt x}(u_t) + {\tt r }(x')$.
If $x' = x_l$ then $u_t = u_s$ as illustrated in Fig.~\ref{fig:proofB}-(a), and
then $x_r \geq {\tt x}(u_s) + {\tt r}(x_l) = x_l + \delta_b(x_l)$. 
If, on the other hand, $x' = {\tt x}(u_t)$ then $x_r \geq {\tt x}(u_t) + {\tt r}({\tt x}(u_t)) = e({\tt x}(u_t)) \geq e(x_l)$ 
and we still get $x_r \geq e(x_l) = x_l + {\tt r}(x_l) \geq x_l + \delta_b(x_l)$. 
%
A symmetric argument can be used to show
that $x_l \leq x_r-\delta_f(x_r)$ and therefore interval I is contained within
$[x_l,x_r]$.

\begin{figure}[!b]
\proofB

\caption{The proof of Lemma~\ref{lem:criticalpoints}. (a) The case when
$x_l$ is right critical, $x_r$ is left critical,
$x_l+\delta_b(x_l) > x_r-\delta_f(x_r)$, $t(x_l) = t_{en_w}$, and
$x' = \max\{x_l, {\tt x}(u_t)\} = x_l$. By induction,
$P(x_l)$ is amortized as illustrated in Fig.~\ref{fig:induction}-(b) and shown
above in blue. $\Delta P(x_l,x_r)$ can be amortized as cost 
$\frac{4}{\sqrt{3}}$ over interval 
$I=[x_r-\delta_f(x_r),x_l+\delta_b(x_l)]$ (shown in red) {\em plus} cost 
$\frac{6}{\sqrt{3}}$
over the interval $[x_l,x_r]$ not including $I$ (empty in the example). 
Thus $P(x_r)$ can be amortized as shown in Fig.~\ref{fig:induction}-(a).
(b) The case when $x_l$ is left critical and $x_r$ is right critical. By
induction, $P(x_l)$ is amortized as illustrated in Fig.~\ref{fig:induction}-(a)
and shown above in blue. $\Delta P(x_l,x_r)$ can be amortized as cost 
$\frac{10}{\sqrt{3}}-2$ over interval $[x_l,x_r]$ {\em plus} cost 
$\frac{4}{\sqrt{3}}-2$ over intervals $[x_l-\delta_f(x_l),x_l]$ and 
$[x_r,x_r+\delta_b(x_r)]$ (shown in red). Thus $P(x_r)$ can be amortized
as shown in Fig.~\ref{fig:induction}-(b).}
\label{fig:proofB}
\end{figure}

For every $x$ in interval $I$, 
we have
$x \geq w(x_r) = b(x_r) \geq b(x)$ and $x \leq e(x_l) \leq f(x_l) \leq f(x)$.
Given that no bad transitions take place
within $[x_l,x_r]$, $t(x)$ must be $t_{s_en_w}$ for $x \in I$ and therefore
$\Delta P$ has a growth rate bounded by $\frac{4}{\sqrt{3}}$ in $I$. Then,
as illustrated in Fig.~\ref{fig:proofB}-(a),
\begin{align*}
	P(x_r) & = P(x_l) + \Delta P(x_l,x_r) \\
	       & \leq \left(\frac{10}{\sqrt{3}}-2\right)x_l + \left(\frac{4}{\sqrt{3}}-2\right)\delta_b(x_l)
	       + \frac{4}{\sqrt{3}}|I| + \frac{6}{\sqrt{3}}(x_r-x_l-|I|) \\
	       & = \frac{6}{\sqrt{3}}x_r + \left(\frac{4}{\sqrt{3}}-2\right)(x_l + \delta_b(x_l))
	       - \frac{2}{\sqrt{3}}(x_l + \delta_b(x_l) - (x_r -\delta_f(x_r))) \\
	       & \leq \frac{6}{\sqrt{3}}x_r + \left(\frac{4}{\sqrt{3}}-2\right)(x_l + \delta_b(x_l))
	       - \left(\frac{4}{\sqrt{3}}-2\right)(x_l + \delta_b(x_l) - (x_r -\delta_f(x_r))) \\
               & = \left(\frac{10}{\sqrt{3}}-2\right)x_r - \left(\frac{4}{\sqrt{3}}-2\right)\delta_f(x_r)
\end{align*}




\noindent{\bf Case LR:} 
Finally, we consider the case when $x_l$ is left critical and $x_r$ is right
critical. We will show that in this case:
\begin{equation}
\label{eq:keyineq0}
\Delta P(x_l,x_r) \leq \left(\frac{10}{\sqrt{3}}-2\right)(x_r - x_l) + \left(\frac{4}{\sqrt{3}}-2\right)(\delta_f(x_l) + \delta_b(x_r))
\end{equation}
Note that if this inequality holds then (refer to Fig.~\ref{fig:proofB}-(b)):
\begin{align*}
P(x_r) & = P(x_l) + \Delta P(x_l,x_r) \leq \left(\frac{10}{\sqrt{3}}-2\right)x_r + \left(\frac{4}{\sqrt{3}}-2\right)\delta_b(x_r)
\end{align*}

Let $z_{ij}$ be the time within
interval $[x_l,x_r]$ spent in transition $t_{ij}$, for every transition $t_{ij}$,
and let $z = \sum z_{ij} = x_r-x_l$.
We assume w.l.o.g. that $t(x_l) = t_{wn_e}$. Note that either
$t(x_r) = t_{s_we}$ or $t(x_r) = t_{en_w}$.

\noindent{\bf Subcase LR.1:} We consider the case when
$t(x_r) = t_{s_we}$ first; this means that $t(x) \not= t_{s_ew}, t_{en_w}$
when $x \in [x_l, x_r]$. 
Either $[x_l,x_r]$ is a $t_{wn_e}$-interval with $t(x_r) = t_{s_we}$, and so
Lemma~\ref{lem:linearworstcase} applies to interval $[x_l,x_r]$, or
$[x_l,x_q]$ is a $t_{wn_e}$-interval for some $x_q < x_r$ with 
$t(x_q) = t_{s_ww}$,
and thus Lemma~\ref{lem:softrecovery} applies to interval $[x_l,x_q]$. In the
second case note that $t(x) \not= t_{wn_e}$ when $x \in [x_q, x_r]$ because, 
otherwise, there would be a left critical point between $x_l$ and $x_r$, 
a contradiction. This, together with $w(x_q) \leq w(x_r) < e(x_r)$,
gives:
\begin{equation}
\label{eq:keyineq}
\frac{2}{\sqrt{3}}z_{wn_e}-\frac{1}{\sqrt{3}}(z_{s_wn_w}+z_{s_en_w}) \leq \left(\frac{2}{\sqrt{3}}-1\right)(e(x_r)-e(x_l)+2\delta_f(x_l))
\end{equation}
Note that this inequality holds for the first case as well
(using Lemma~\ref{lem:linearworstcase}). Via symmetric arguments either $[x_l,x_r]$ is a
$t_{s_we}$-interval with $t(x_l) = t_{wn_e}$, and so Lemma~\ref{lem:linearworstcase}
applies to interval $[x_l,x_r]$, 
or $[x_{q'},x_r]$ is a $t_{s_we}$-interval for some $x_{q'} > x_l$ with
$t(x_{q'}) = t_{en_e}$, and thus Lemma~\ref{lem:softrecovery} applies to interval
$[x_{q'},x_r]$. Either way, the inequality
\[\frac{2}{\sqrt{3}}z_{s_we}-\frac{1}{\sqrt{3}}(z_{s_en_w}+z_{s_en_e})\leq \left(\frac{2}{\sqrt{3}}-1\right)(w(x_r)-w(x_l)+2\delta_b(x_r))\]
holds and it, together with inequality~(\ref{eq:keyineq}), gives:
\[\frac{2}{\sqrt{3}}(z_{wn_e} +z_{s_we})-\frac{2}{\sqrt{3}}(z_{s_wn_w}+z_{s_en_w}+z_{s_wn_e}+z_{s_en_e})\leq \left(\frac{4}{\sqrt{3}}-2\right)(x_r-x_l+\delta_b(x_r)+\delta_f(x_l))\]
We can now show that (\ref{eq:keyineq0}) holds:
\begin{align*} \Delta P(x_l, x_r) & \leq \frac{6}{\sqrt{3}}z + \frac{2}{\sqrt{3}}(z_{wn_e} +z_{s_we}) - \frac{2}{\sqrt{3}}(z_{s_wn_w}+z_{s_en_w}+z_{s_wn_e}+z_{s_en_e}) \\ 
	& \leq \frac{6}{\sqrt{3}}(x_r - x_l) + \left(\frac{4}{\sqrt{3}}-2\right)(x_r-x_l+\delta_b(x_r)+\delta_f(x_l)) \\
        & = \left(\frac{10}{\sqrt{3}}-2\right)(x_r - x_l) + \left(\frac{4}{\sqrt{3}}-2\right)(\delta_b(x_r) + \delta_f(x_l))
\end{align*}

\noindent{\bf Subcase LR.2:} We now consider the case when $t(x_l) = t_{wn_e}$ and $t(x_r) = t_{en_w}$. 
Note that this means that $t(x) \not= t_{s_ew}, t_{s_we}$ when $x \in [x_l, x_r]$. 
If 
\begin{equation}
\label{eq:keyineq2}
\frac{2}{\sqrt{3}}(z_{wn_e} +z_{en_w})-\frac{2}{\sqrt{3}}(z_{s_wn_w}+z_{s_en_w}+z_{s_wn_e}+z_{s_en_e})\leq
	\left(\frac{4}{\sqrt{3}}-2\right)(x_r-x_l+\delta_b(x_r)+\delta_f(x_l))
\end{equation}
then the above argument can be applied to obtain~(\ref{eq:keyineq0}). We will show next that if $t(x) = t_{\ast w}$ or $t(x) = t_{\ast e}$ for 
some $x \in [x_l, x_r]$ then inequality~(\ref{eq:keyineq2}) holds.

Consider the maximal strict $t_{wn_e}$-interval $[x_l,x_q]$ and the maximal
strict $t_{en_w}$-interval $[x_{q'},x_r]$. Recall that $t(x) \not= t_{en_w},t_{\ast w},t_{\ast e}$ for $x \in [x_l,x_q]$ and that $t(x) \not= t_{wn_e},t_{\ast w},t_{\ast e}$ for $x \in [x_{q'},x_r]$ and thus $x_l < x_q,x_{q'} < x_r$. 
If $t(x) = t_{\ast w}$ or $t(x) = t_{\ast e}$ 
for  some $x$ such that $x_l < x < x_r$, and using the assumption that there
are no critical points between $x_l$ and $x_r$, we must have that 
$x_l < x_q \leq x_{q'} < x_r$ and that $t(x) \not= t_{wn_e},t_{en_w}$ when
$x \in [x_q,x_{q'}]$. This also means that $t(x_q)$ is either $t_{s_ww}$ or 
$t_{s_ee}$ and that $t(x_{q'})$ is either $t_{s_ww}$ or $t_{s_ee}$. 

If $t(x_q)=t_{s_ww}$ then Lemma~\ref{lem:softrecovery} applies to interval
$[x_l,x_q]$. If $t(x_q) = t_{s_ee}$ then Lemma~\ref{lem:linearworstcase} applies
to interval $[x_l,x_q]$. Since $t(x) \not= t_{wn_e}$ when $x \in [x_q, x_r]$ 
and $w(x_q) \leq w(x_r) \leq e(x_r)$ and just as in case LR.1, it follows in 
both cases that
\begin{equation}
\label{eq:16}
\frac{2}{\sqrt{3}} z_{wn_e} - \frac{1}{\sqrt{3}}(z_{s_wn_w} + z_{s_en_w}) \leq \left(\frac{2}{\sqrt{3}} - 1\right) (e(x_r) - e(x_l) + 2 \delta_f(x_l))
\end{equation}
If $t(x_{q'})=t_{s_ee}$ then Lemma~\ref{lem:softrecovery}
applies to interval $[x_{q'},x_r]$. If $t(x_{q'}) = t_{s_ww}$ then Lemma~\ref{lem:linearworstcase} applies to interval $[x_{q'},x_r]$. Since $t(x) \not= t_{en_w}$ when $x \in [x_l, x_{q'}]$ and $w(x_l) \leq w(x_{q'}) \leq e(x_{q'})$ it follows 
in both cases that
\begin{equation}
\label{eq:17} 
\frac{2}{\sqrt{3}} z_{en_w} - \frac{1}{\sqrt{3}}(z_{s_wn_e} + z_{s_en_e}) \leq \left(\frac{2}{\sqrt{3}} - 1\right) (w(x_r) - w(x_l) + 2 \delta_b(x_r))
\end{equation}
By combining inequalities (\ref{eq:16}) and (\ref{eq:17}) we obtain inequality 
(\ref{eq:keyineq2}).

We assume now that $t(x) \not= t_{*w}, t_{*e}$ for all $x \in [x_l,x_r]$.
We can also assume that
inequality~(\ref{eq:keyineq2}) does not hold which implies
\[\frac{2}{\sqrt{3}}(z_{wn_e} +z_{en_w}) > \left(\frac{4}{\sqrt{3}}-2\right)(x_r - x_l + \delta_b(x_r)+\delta_f(x_l))
\]
which, using $z_{wn_e} +z_{en_w} \leq x_r - x_l$, implies:
\begin{equation}
\label{eq:keyineq3}
\left(4 - \frac{6}{\sqrt{3}}\right)(\delta_b(x_r)+\delta_f(x_l)) < \left(\frac{6}{\sqrt{3}}-2\right)(x_r - x_l)
\end{equation}

Now, $\Delta P(x_l,x_r) = U(x_r) - U(x_l) + L(x_r) - L(x_l)$. For all
$x \in [x_l,x_r]$, $t(x) = t_{*n_w}, t_{*n_e}$  and so, by
Lemma~\ref{lem:growthrates},
$\frac{\Delta p_N(x)}{\Delta x} = \frac{2}{\sqrt{3}}$. It follows that
$U(x_r) - U(x_l) \leq \frac{2}{\sqrt{3}}(x_r-x_l)$. In order to bound
$L(x_r) - L(x_l)$, we consider the following path $\mathcal{P}$ from point $\ell(x_l)$,
say $l_i$, that lies on side $w$ of $H(x_l)$ to point $\ell(x_r)$, say $l_j$
(where $i < j$), that lies on side $e$ of $H(x_r)$: 
$\ell(x_l)=l_i, u_i,u_{i+1},\dots,u_j,l_j=\ell(x_r)$ (refer to Fig.~\ref{fig:pathcost}).
Then, if $|\mathcal{P}|$ is the length of $\mathcal{P}$:
\begin{figure}
\center{\pathcost}

\caption{
The proof of Lemma~\ref{lem:criticalpoints}. Shown is the case when $x_l$ is
left critical, $x_r$ is right critical, $t(x_l) = t_{wn_e}$, $t(x_r)=t_{en_w}$,
and $t(x) \not= t_{\ast w}, t_{\ast e}$ for $x \in [x_l,x_r]$. The path
$\mathcal{P}$, defined as $l_i,u_i,u_{i+1},\dots,u_j,l_j$, is effectively
a shortcut for $l_i, l_{i+1}, \dots,l_j$.
}
\label{fig:pathcost}
\end{figure}
\begin{align*}
L(x_r) - L(x_l) & = d_{T_{1j}}(p,l_j) + p_S(l_j,x_r) - d_{T_{1i}}(p,l_i) - p_S(l_i,x_l) \\
& \leq d_{T_{1i}}(p,l_i) + |\mathcal{P}| + p_S(l_j,x_r) - d_{T_{1i}}(p,l_i) -p_S(l_i,x_l) \\
& = -p_S(l_i,x_l) + |\mathcal{P}| + p_S(l_j,x_r) \\
& \leq -p_S(l_i,x_l) + d_2(l_i,u_i) + p_N(u_i,x_l) + U(x_r) - U(x_l) \\
& \quad - p_N(u_j,x_r) + d_2(l_j,u_j) + p_S(l_j,x_r)
\end{align*}
using the fact that the length of the path $u_i,u_{i+1},\dots,u_j$ is at most
$p_N(u_i,x_l) + U(x_r) - U(x_l) - p_N(u_j,x_r)$. 
%
%
By Lemma~\ref{lem:switchpath},
$-p_S(l_i,x_l) + d_2(l_i,u_i) + p_N(u_i,x_l) \leq
 \left(2-\frac{2}{\sqrt{3}}\right)\delta_f(x_l)$ and
$-p_N(u_j,x_r) + d_2(l_j,u_j) + p_S(l_j,x_r) \leq \left(2-\frac{2}{\sqrt{3}}\right)\delta_b(x_r)$. Combining the bounds on
$U(x_r) - U(x_l)$ and $L(x_r) - L(x_l)$, we get
\[
\Delta P(x_l,x_r) \leq \left(2-\frac{2}{\sqrt{3}}\right)(\delta_b(x_r)+\delta_f(x_l))+\frac{4}{\sqrt{3}}(x_r-x_l)
\]
Summing the above inequality with inequality~(\ref{eq:keyineq3}) yields 
(\ref{eq:keyineq0}) and completes the inductive proof in this case as well.
\end{proof}

We can now provide a proof of (Amortization) Lemma~\ref{le:mainlemmaB}.
\begin{proof}[Proof of Lemma~\ref{le:mainlemmaB}]
Note that ${\tt x}(q)$ is a left critical point and that
$\delta_f({\tt x}(q)) = 0$.
	Therefore, by Lemma \ref{lem:criticalpoints},
$P({\tt x}(q)) \leq (\frac{10}{\sqrt{3}}-2){\tt x}(q)$. 
Since $2d_{T_{1n}}(p, q) = P({\tt x}(q))$, the lemma follows.
\end{proof}

}

\section{Conclusion}
\label{sec:conclusion}

The approach we use to bound the length of the shortest path in a Delaunay
triangulation $T$ between points $s$ and $t$ is to consider the linear
sequence $T_{1n}$ of triangles of $T$ that segment $[st]$ intersects.
We show that, in general, 
$T_{1n}$ can be split into 1) disjoint linear sequences of triangles
$T_{i_1j_1}, T_{i_2j_2}, \dots, T_{i_k,j_k}$ that contain no gentle path
and 2) $k-1$ gentle paths with a gentle path connecting the right vertex of
$T_{j_l}$ with the left vertex of $T_{i_{l+1}}$ for $l = 1, \dots, k-1$. 
The worst case stretch factor for the Delaunay triangulation is then the
maximum between the worst case stretch factors for 1) a path connecting the
leftmost and rightmost points in a linear sequence $T_{ij}$ that contains no
gentle path and 2) a gentle path.

(Main) Lemma~\ref{le:divide} and Lemma~\ref{le:lower_bound} show
that the worst case stretch factor for {\Large\varhexagon}-Delaunay
triangulations comes from gentle path constructions. It turns out that
similar conclusions can also be made regarding $\triangle$- and
$\square$-Delaunay triangulations.

For $\ocircle$-Delaunay triangulations, the situation seems to be different.
The lower bound construction by Bose et al.~\cite{BDLSV11} corresponds to
a gentle path construction and has stretch factor $1.5846$. The lower
bound construction by Xia and Zhang~\cite{XZ11} corresponds to a linear
sequence that contains no gentle path and has stretch factor $1.5932$.
We think that the worst case stretch factor for $\ocircle$-Delaunay 
triangulations will come from a construction similar to the one by Xia and 
Zhang~\cite{XZ11}. Therefore, to get a tight bound on the stretch factor of a
$\ocircle$-Delaunay triangulation one needs to develop techniques that
give tight bounds on the stretch factor of a linear sequence that contains
no gentle path.

\iftoggle{abstract}
{\begin{figure}[!b]}
{\begin{figure}}
\begin{center}
\mickey
\caption{The Mickey Mouse {\Large\varhexagon}-Delaunay triangulation. The
inradii of $H_1$ and $H_n$ are both set to $1$. Edges that belong to a
shortest path from $s$ to $t$ are in bold.}
\label{fig:mickey}
\end{center}\end{figure}

We have done so for {\Large\varhexagon}-Delaunay triangulations. 
Our (Amortization) Lemma~\ref{le:mainlemmaB} implies that for 
{\Large\varhexagon}-Delaunay triangulations the worst case
stretch factor for a linear sequence $T_{ij}$ with no gentle paths is $(\cT)$. 
It turns out that our analysis is tight: Figure~\ref{fig:mickey}
shows a construction--which we name the Mickey Mouse 
{\Large\varhexagon}-Delaunay triangulation--that, for any $\epsilon > 0$, 
can be extended to a
{\Large\varhexagon}-Delaunay triangulation whose shortest path between $s$ and
$t$ is at least $(\cT) d_x(s,t)-\epsilon$. Unsurprisingly, the construction
corresponds to the lower bound construction by Xia and Xhang~\cite{XZ11} for 
$\ocircle$-Delaunay triangulations.

Based on this we think that the techniques we developed for obtaining the 
tight bound in Lemma~\ref{le:mainlemmaB} will be useful in obtaining better
upper bounds for the stretch factor of other kinds of Delaunay triangulations.

\bibliography{paper}

\end{document}